%% file: v8_4.tex
\documentclass[a4paper,11pt]{article}

\usepackage[margin=0.75in]{geometry}

\usepackage{enumerate}
\usepackage{enumitem}

\usepackage[textsize=tiny]{todonotes}

\include{macros}

\title{ Computational Complexity and the Nature of of Quantum Mechanics}
\author{A. Benavoli, A. Facchini, M. Zaffalon\\
IDSIA, Manno, Switzerland}

\begin{document}

\maketitle

\begin{abstract}
Quantum theory (QT) has been confirmed by numerous experiments, yet we still cannot fully grasp the meaning of the theory. As a consequence, the quantum world appears  to us paradoxical.
Here we shed new light on QT by being based on two main postulates:
\begin{enumerate}
\item the theory should be logically consistent; 
\item inferences in the theory should be computable in polynomial time.
\end{enumerate}
The first postulate is what we  require to each well-founded mathematical theory. The computation postulate defines the physical component of the theory. 
We show that the computation postulate is the only true divide between QT, seen as a generalised theory of probability,
and classical probability. 
All quantum paradoxes, and entanglement
in particular, arise from the clash of trying to reconcile a computationally intractable, somewhat idealised,  
theory (classical physics) with a computationally tractable theory (QT) or, in other
words, from regarding physics as fundamental rather than computation.
\end{abstract}

\section{Introduction}

Quantum theory (QT) is one of the most fundamental, and accurate, mathematical descriptions of our physical world. It dates back to the 1920s, and in spite of nearly one century passed by since its inception, we do not have a clear understanding of such a theory yet. In particular, we cannot fully grasp its \emph{meaning}: why the theory is the way it is. As a consequence, we cannot come to terms with the many paradoxes it appears to lead to---the so-called ``quantum weirdness''.

Let us recall that QT is widely regarded as a ``generalised'' theory of probability. 
Since its foundation, there have been two main ways to explain the differences between QT and classical probability. The first one goes back to Birkhoff and von Neumann \cite{birkhoff1936logic} and explains the differences  with the premise that,
in QT, the Boolean algebra of events is taken over by the ``quantum logic'' of projection operators on a Hilbert space.
The second one is based on the view that the quantum-classical clash is due to the appearance of negative probabilities \cite{dirac1942bakerian, feynman1987negative}.

More recently, the so-called ``quantum reconstruction'' research effort aims at rebuilding   finite dimensional QT from  primitive postulates.
The search for alternative axiomatisations of QT has been approached following different avenues: extending
Boolean  logic \cite{birkhoff1936logic,mackey2013mathematical,jauch1963can}, using operational primitives and/or information-theoretic postulates \rednew{(the so called ``operational approach'')}\cite{rovelli1996rovelli,hardy2011foliable,hardy2001quantum,barrett2007information,rau2009quantum,fivel2012derivation,wilce2009four,chiribella2010probabilistic,barnum2011information,van2005implausible,pawlowski2009information,dakic2009quantum,fuchs2002quantum,brassard2005information,masanes2011derivation,mueller2016information},
 starting from the phenomenon of quantum nonlocality \cite{barrett2007information,van2005implausible,pawlowski2009information,popescu1998causality,navascues2010glance}, the categorical approach \cite{abramsky2004categorical,coecke2010quantum,cho2015introduction} and building upon the (subjective) foundation of probability \cite{caticha1998consistency,Caves02,Appleby05a,Appleby05b,Timpson08,longPaper, 
Fuchs&SchackII,mermin2014physics,holevo2011probabilistic,pitowsky2003betting,Pitowsky2006,goyal2010origin,khrennikov2009interpretations,khrennikov2013non,goyal2011quantum,goyal2014derivation,skilling2017symmetrical,benavoli2016quantum,benavoli2017gleason}. 
\rednew{The present work is particularly close to the latter approach, whose goal is to answer questions such as \cite{goyal2011quantum}: 
What is the relationship of QT probability calculus to classical probability theory} (CPT)? Is QT consistent with CPT?
Is it some kind of extension of CPT, and, if so, what is the nature and conceptual foundation of that extension?
Note that the different ways one can answer these questions critically depend on the used notion
 of ``being consistent''. 

In trying to understand and explain QT, we share with the operational approach  the aim of deriving the mathematical structure of  Hilbert spaces, but with two main differences.

Firstly,
in the operational approach,  the notions of measurement and outcome probability are primitive elements of the theory, and therefore the emphasis is on the predictive aspect of QT---it aims at computing the outcome probability. 
We instead take a different view and mainly focus on the explanatory aspect of QT, that is on its being a (probabilistic) model of reality. \rednew{Hence, from this perspective,  in the process of modelling say a Stern-Gerlach experiment, one does not start from the outcomes of the experiment but from 
 (hidden) variables representing the ``directions''  of the particle's spins, their domains (possibility space)
 and the functions of the variables we can observe (observables) and are interested in (queries).
 \\
 Secondly,  in this work we do not aim at characterising QT has being the \emph{unique} theory satisfying some set of principles, or postulates. Instead, we aim at showing that ultimately QT appears to us to behave like it behaves because,  
 once the appropriate variables,  possibility space,  observables and queries are specified, it constitutes an \emph{instance} of a theory based on two simple postulates:
\begin{description}[noitemsep]
 \item[(Coherence)] The theory should be logically consistent. 
\item[(Computation)] Inferences in the theory should be computable in polynomial time. 
\end{description}
The first postulate is what it is essentially required to each well-founded mathematical theory, be it physical or not: to be based on a few axioms and rules from which we can unambiguously derive its mathematical truths. 
The set of rules we refer to are at the foundation of CPT, and more generally of all \textit{rational choice theories} \cite{nau1991arbitrage}: decision theory, game theory, social choice theory and market theory. In this setting the term logical consistency has a very precise meaning: a theory is logically consistent (coherent)
 when it disallows ``Dutch books'' (``arbitrages'' in market theory). In gambling terms, this means that a bettor on a experiment cannot be made a sure loser by exploiting inconsistencies in her probabilistic assessments. 
\\
The second postulate is the central one. When seen through the prism of the coherence postulate, it tells us that QT is just a computationally efficient version of CPT, or stated otherwise a theory of computational rationality. In doing so, it explains why there are cases where QT is in agreement with CPT---that is precisely when inferences in CPT \emph{can} be computed in polynomial time---, and cases where QT cannot be reconciled with CPT---that is precisely when inferences in CPT \emph{cannot} be computed in polynomial time.
\\
But there is more to it: the setting in which our results are proved applies to any theory satisfying the coherence and computational postulates, and goes thence beyond the specific case of QT. As a consequence, Bell-like inequalities
and entanglement are shown to be phenomena that can arise in all theories of computational rationality, and are therefore not specific to QT.} To illustrate this fact, 
we devise an example of a computationally tractable theory of probability that is unrelated to QT but that admits entangled states.
An example of Bell's like inequalities and entanglement with classical coins is given in \cite{Benavoli2019d}.
%
Before listing our findings in detail, we introduce the model framework we will consider in the rest of the paper.


\subsection{Framework}
 The first step in defining a model is to explicit its variables, their domains (possibility space)
 and  their functions we can observe (observables). 
\begin{description}
 \item[Possibility space:] For a single particle ($n$-level system), we consider  the possibility space
$$
\overline{\complex}^{n}\coloneqq\{ x\in \complex^{n}: ~x^{\dagger}x=1\}.
$$
We can interpret an element $\tilde{x} \in \overline{\complex}^{n}$ as ``input data'' for some classical preparation procedure.
For instance,  in the case of the spin-$1/2$ particle ($n = 2$),  if $\theta = [\theta_1 , \theta_2 , \theta_3 ]$ is the direction of a filter
in the Stern-Gerlach experiment, then $\tilde{x}$ is its one-to-one mapping into $\Omega$ (apart from a phase term).
For spin greater than $1/2$, the variable $\tilde{x} \in \overline{\complex}^{n}$ 
 cannot directly be interpreted in terms only of ``filter direction''.
Nevertheless,  at least on the formal level, $\tilde{x}$ can play the role of a ``hidden variable'' in our model
and  $\overline{\complex}^{n}$ of the possibility space $\pspace$.
This hidden-variable model for QT was first introduced in \cite[Sec.~1.7]{holevo2011probabilistic}, where the author  explains
why this model does not contradict the existing ``no-go'' theorems for hidden variables.
\item[Cartesian product:] More generally, we can consider composite systems of $m$ particles, each one with $n_j$ degrees of freedom. The  possibility space is then the Cartesian product 
 \begin{equation}
   \label{eq:pspace}
\pspace=\times_{j=1}^m \overline{\complex}^{n_j}.
 \end{equation}
\end{description}
So far this framework is totally compatible with CPT, in fact we can define a classical probability distribution over  complex vectors.
We can simply interpret $\overline{\complex}^{n}$ (and $\times_{j=1}^m \overline{\complex}^{n_j}$) as a convenient feature space,
similarly to the Fourier domain (after the Fourier transform).
 \begin{description}
\item[Observables:] An observable is given as a function of the variables of the model. But since observables are problem dependent, we may restrict their collection to just a subspace of all possible such functions. In case of a single particle, observables correspond to  quadratic forms:
 \begin{equation}
   \label{eq:gamble_one}
  g(x)=x^\dagger G x,
 \end{equation}
with $G \in \He^{n \times n}$ ($\He^{n \times n}$ being the set of $n \times n$ Hermitian matrices).
For $m$ particles, we have that:
 \begin{equation}
   \label{eq:gamble_many0}
  g(x_1,\dots,x_m)=(\otimes_{j=1}^m x_j)^\dagger G (\otimes_{j=1}^m x_j),
 \end{equation}
 with $G \in \He^{n \times n}$, $n=\prod_{j=1}^m n_j$ and where $\otimes$ denotes the tensor product between vectors regarded as column matrices.
\end{description}

Note that Expression \eqref{eq:gamble_many0}, and thus the use of the tensor product in composite systems,  can be justified based on Expressions \eqref{eq:pspace} and \eqref{eq:gamble_one}. The argument goes roughly as follows (for a more in depth discussion see  Section \ref{sec:tensorproduct}).
In the theory of probability, structural judgements such as independence, are defined on product of observables on the single variables, e.g., two variables $X,Y$
are independent iff $E[f(X)g(X)]=E[f(X)]E[g(Y)]$ for all $f,g$, where $E[\cdot]$ denotes the expectation operator. In the specific case we are considering, 
the product of observables  have the form $\prod_{j=1}^m x_j^{\dagger}  G_j x_j$.
Now, it is not difficult to verify that such product is mathematically  the same as  $(\otimes_{j=1}^m x_j)^\dagger (\otimes_{j=1}^m G_j) (\otimes_{j=1}^m x_j)$. 
By closing the set of product gambles  under the operations of addition and scalar (real number) multiplication, 
we get the vector space whose domain coincide with the collection of gambles of the form as in  \eqref{eq:gamble_many0}. 
Hence,  in the framework advanced in the present work,  the tensor product is ultimately a derived notion, not a primitive one.

Again, notice that up to here we are still fully in agreement with CPT.  


\subsection{Quantum weirdness}
Since CPT can be defined over complex vectors (the possibility space in CPT can be any set), what originates the  ``quantum weirdness'' cannot be the use of complex numbers (vectors).
In this paper, we  show that the ``quantum weirdness'' comes from the postulate of computation.

The inference problem, \rednew{that is the problem of deciding whether a certain claim, or judgement, follows from a given set of assessments}, is in general undecidable in  CPT, whereas it is NP-hard  when we restrict inferences to observables  satisfying \eqref{eq:gamble_many0}.

If we impose the additional postulate of computation, the theory becomes one of ``computational rationality'': one that is consistent (or coherent), up to the degree that polynomial computation allows. \rednew{ \textbf{QT coincides with such a weaker, and hence more general, theory of probability}.  Moreover, the mathematical structure of Hilbert space usually associated with QT can be seen as a consequence of the adoption of such postulate}.

\rednew{As it will be made clear in what follows, the point is that for a subject living inside QT, all is coherent. Instead, for another subject, living in the classical, and somewhat idealised, probabilistic world (hence not subject to the computation postulate)}, QT displays some inconsistencies: \textbf{precisely those that cannot be fixed in polynomial time}. All quantum paradoxes, and entanglement in particular, arise from the clash of these two world views: i.e., from trying to reconcile an unrestricted theory (i.e., classical physics) with a  theory of computational rationality (quantum theory), \rednew{or, in other words, from regarding our common-sense notion of physics as fundamental rather than computation.} 

There is more, we highlight our findings in the following list.
\begin{description}
  \item[$\rho$ as a moment matrix:] Requiring the computation postulate is tantamount to defining a probabilistic model using only a finite number of moments, and therefore, implicitly, to defining the model as the set of all probabilities\footnote{Note that here we use `probabilities' to mean \emph{finitely additive probabilities}, alternatively called also `probability charges' (the term `distribution' is usually restricted to $\sigma$-additive probabilities, and therefore it would be misplaced in the present context).} compatible with the given moments.
  \item[Imprecise Probability:] A probability is not uniquely defined by a finite number of moments. Therefore, QT can be shown to be ``generalised'' also in another direction, as it is a theory of ``imprecise'' probability.
    \item[Continuous Probability:] Since the possibility space $\pspace$ is infinite with probabilities we mean continuous probabilities. In our context this explains why the set of states in QT has infinite extremes and, why determining whether a state is entangled or  not is NP-hard.
  \item[Signed Probabilities:] In QT, some of these compatible probabilities can actually be \emph{signed}, that is, they allow for ``negative probabilities'', which  have no intrinsic meaning beyond the fact of constituting a compatible probabilistic model.
    \item[$n$-level single particle system:] We show that QT is always compatible with CPT (classical and computational rationality agree).
    \item[CHSH experiment] We show that the witness function in the CHSH experiment is nothing else than a negative function whose negativity
cannot be assessed in polynomial time, whence it is not ``negative'' in QT. 
          \item[Entanglement Witness Theorem:] We reformulate the Entanglement Witness Theorem in terms of classical rationality versus computational rationality.
          \item[Entanglement outside QT:] We devise an example of a computationally tractable theory of probability that is unrelated to QT but that admits entangled states.
    \item[Tensor Product:] It is a derived notion in our context. We use  computational rationality to explain why it has a central role in the standard axiomatisation of QT. 
 \end{description}

\subsection{Outline of the paper}
Section~\ref{sec:dg} is concerned with the coherence principle. We recall how probability can be derived (via mathematical duality) from a set of simple logical axioms. Addressing consistency (coherence or rationality) in such a setting is a standard task in logic; in practice, it reduces to prove that a certain real-valued bounded function is non-negative. 

Section~\ref{sec:comp} details the computation postulate. We consider the problem of verifying the non-negativity of a function as above. This problem is generally undecidable or, when decidable, NP-hard. 
We make the problem polynomially solvable by redefining the meaning of (non-)negativity. We give our fundamental theorem (Theorem~\ref{th:fundamental}) showing that the redefinition is at the heart of the clash between classical probability and computational rationality.

We show in Section~\ref{sec:qt} that QT is a special instance of computational rationality and hence that Theorem~\ref{th:fundamental} is not only the sole difference between quantum and classical probability, but also  the distinctive reason for all  quantum paradoxes; this latter part is discussed in Section~\ref{sec:entang}. 
In particular, to give further insight about the quantum-classical clash, in Section~\ref{sec:local} we reconsider the question of local realism in the light of computational rationality; in Section~\ref{sec:witness} we show that the witness function, in the fundamental ``entanglement witness theorem'',  is nothing else than a negative function whose negativity cannot be assessed in polynomial time---whence it is not ``negative'' in QT.

Moreover, using Theorem \ref{th:fundamental}, in Section~\ref{sec:ent_not_only} we devise an example of a computationally tractable theory of probability that is unrelated to QT but that admits entangled states. This shows in addition that the ``quantum logic''  and the ``quasi-probability'' foundations of QT are two faces of the same coin, being natural consequences of the computation principle.

We finally discuss the results in Section~\ref{sec:discussions} and compare our approach to QBism and the operational reconstructions. 


\section{Coherence postulate}\label{sec:dg}

De Finetti's subjective foundation of probability \cite{finetti1937}   is  based on the notion   of rationality (consistency or coherence). This approach has then been further developed in \citep{williams1975,walley1991}, giving rise to the so-called \emph{theory of desirable gambles} (TDG). This is an equivalent reformulation of the well-known Bayesian decision theory (\`a la Anscombe-Aumann \cite{anscombe1963}) once it is extended to deal with incomplete preferences \cite{zaffalon2017a,zaffalon2018a}. In this setting probability is a derived notion in the sense that it can be inferred via mathematical duality from a set of logical axioms that one can interpret 
as rationality requirements in the way a subject, let us call her Alice, accepts gambles on the results of an uncertain experiment. It goes as follows. 

Let $\pspace$ denote the possibility space of an experiment (e.g., $\{Head, Tail\}$ 
or $\complex^n$ in QT). A gamble $g$ on $\pspace$ is a bounded real-valued function of $\pspace$, interpreted as
an uncertain reward. It plays the traditional role of variables or, using a physical parlance, of \emph{observables}. 
In the context we are considering, accepting a gamble $g$ by an agent is regarded as a
commitment to receive, or pay (depending on the sign), $g(\omega)$ \emph{utiles} (abstract units of utility, we can approximately identify it with money provided we deal with small amounts of it \cite[Sec.~3.2.5]{finetti1974}) whenever $\omega$ occurs. If by $\gambles$ we denote the set of all the gambles on $\pspace$, the subset of all non-negative gambles, that is, of gambles for which Alice never loses utiles, is given by  
$\nonnegative\coloneqq \{g \in \gambles: \inf g\geq0 \}$. Analogously, 
negative gambles, those gambles for which Alice will certainly lose some utiles, even an epsilon, is defined as $\negative\coloneqq \{g \in \gambles: \sup g < 0 \}$.
In what follows, with $\mathcal{G}:=\{g_1,g_2,\dots,g_{|\mathcal{G}|}\} \subset \gambles$ we denote a finite set of gambles that Alice finds desirable (we will comment on the case when $\mathcal{G}$ may not be finite): these are the gambles that she is willing to accept and thus commits herself to the corresponding transactions.

The crucial question is now to  provide a criterion for a set $\mathcal{G}$ of gambles representing assessments of desirability to be called \emph{rational}. Intuitively Alice is rational if she avoids sure losses: that is, if, by considering the implications of what she finds desirable, she is not forced to find desirable a negative gamble. This postulate of rationality is called ``no arbitrage'' in economics and ``no Dutch book'' in the subjective foundation of probability.
In  TDG we formulate it through the  notion of logical consistency which, despite the informal interpretation given above, is a purely syntactical (structural) notion. To show this, we need to define an appropriate logical calculus (characterising the set of gambles that Alice must find desirable as a consequence of having desired $\mathcal{G}$ in the first place) and based on it to characterise the family of consistent sets of assessments.

For the former, since non-negative gambles
may increase Alice's utility without ever decreasing it, we first have that:
\begin{enumerate}[label=\upshape A0.,ref=\upshape A0]
\item\label{eq:taut} $\nonnegative$ should always be desirable.
\end{enumerate}
This defines the tautologies of the calculus. Moreover,  whenever $f,g$ are desirable for Alice, then any positive linear combination of them should also be desirable (this amounts to assuming that Alice has a linear utility scale, which is a standard assumption in probability). Hence the corresponding deductive closure of a set $\assess$ is given by:
\begin{enumerate}[label=\upshape A1.,ref=\upshape A1]
\item\label{eq:NE} $\domain\coloneqq\posi(\nonnegative\cup \mathcal{G})$.
\end{enumerate}
Here ``$\posi$'' denotes the conic hull operator.\footnote{The conic hull of a set of gambles $\mathcal{A}$
is defined as $\posi(\mathcal{A})=\{\sum_i \lambda_ig_i: \lambda_i\in \reals^{\geq}, g_i\in \gambles\}$.} When $\mathcal{G}$ is not finite, \ref{eq:NE} requires in addition that $\domain$ is closed.

In the betting interpretation given above, a sure loss for an agent is represented by a negative gamble.  We therefore say that:

\begin{definition}[Coherence postulate]
\label{def:avs}
 A set $\domain$ of desirable gambles is \emph{coherent} if and only if
 \begin{enumerate}[label=\upshape A2.,ref=\upshape A2]
\item\label{eq:sl} $ \negative \cap  \domain=\emptyset$.
\end{enumerate}
\end{definition}
\noindent Note that $\domain$  is incoherent if and only if $-1 \in \domain$; therefore $-1$ can be regarded as playing the role of the Falsum and \ref{eq:sl} can be reformulated as $-1 \notin \domain$, see Appendix \ref{sec:onefalsum}.
An example that gives an intuition of the postulates is given in Figure~\ref{fig:main}.

\begin{figure}
\begin{subfigure}{.33\linewidth}
\centering
\includegraphics[scale=.33]{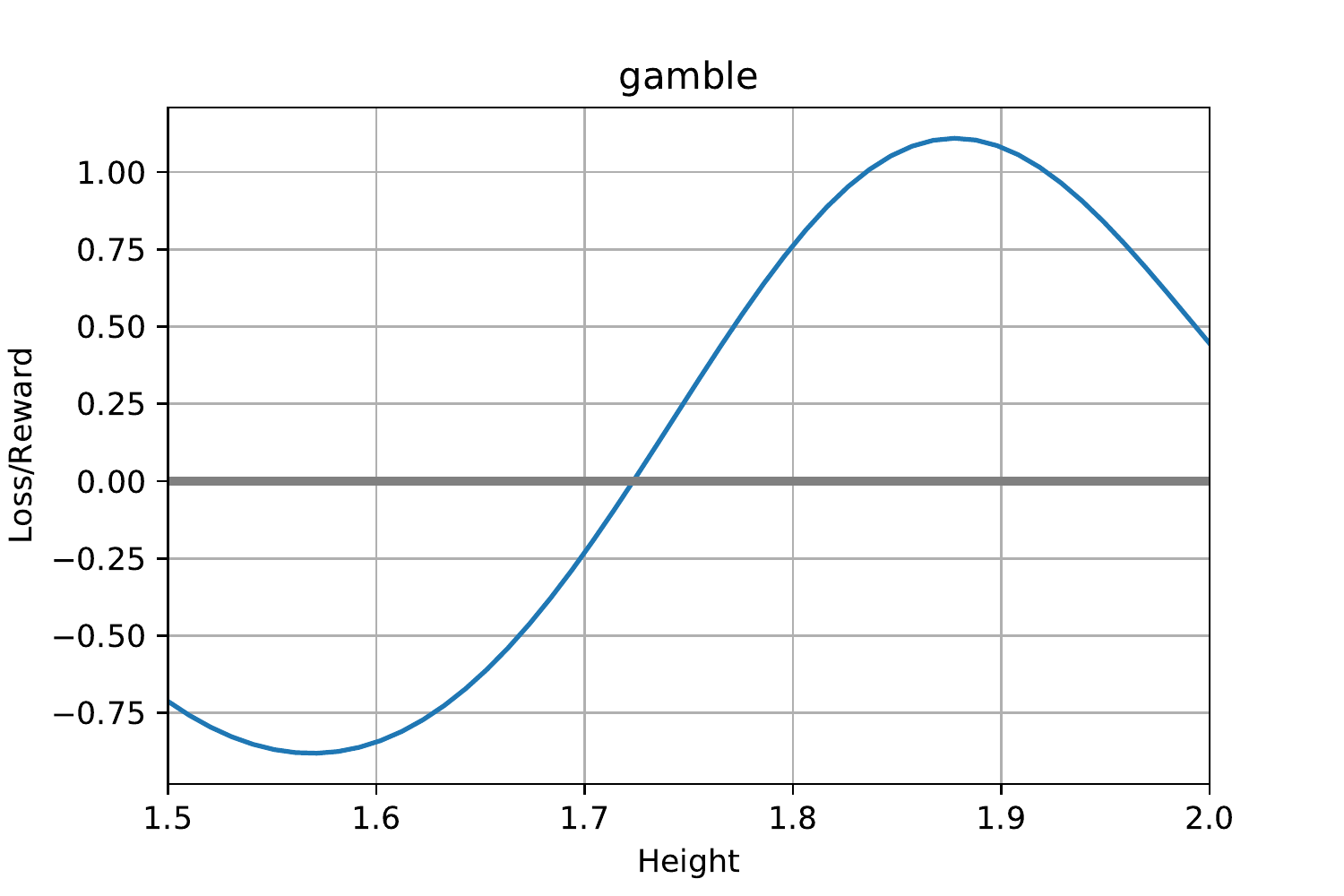}
\caption{}
\label{fig:sub1}
\end{subfigure}
\begin{subfigure}{.33\linewidth}
\centering
\includegraphics[scale=.33]{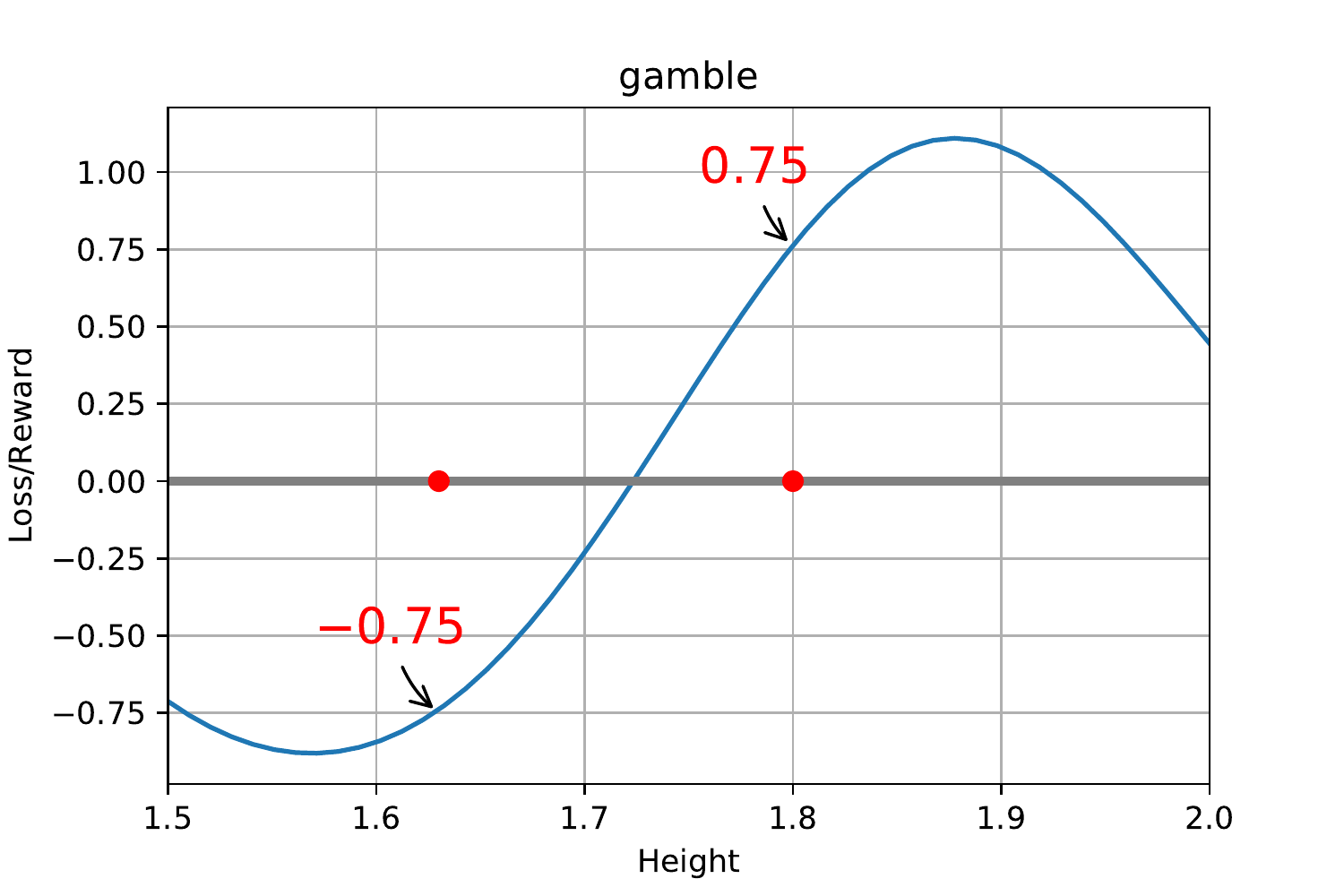}
\caption{}
\label{fig:sub1}
\end{subfigure}
\begin{subfigure}{.33\linewidth}
\centering
\includegraphics[scale=.33]{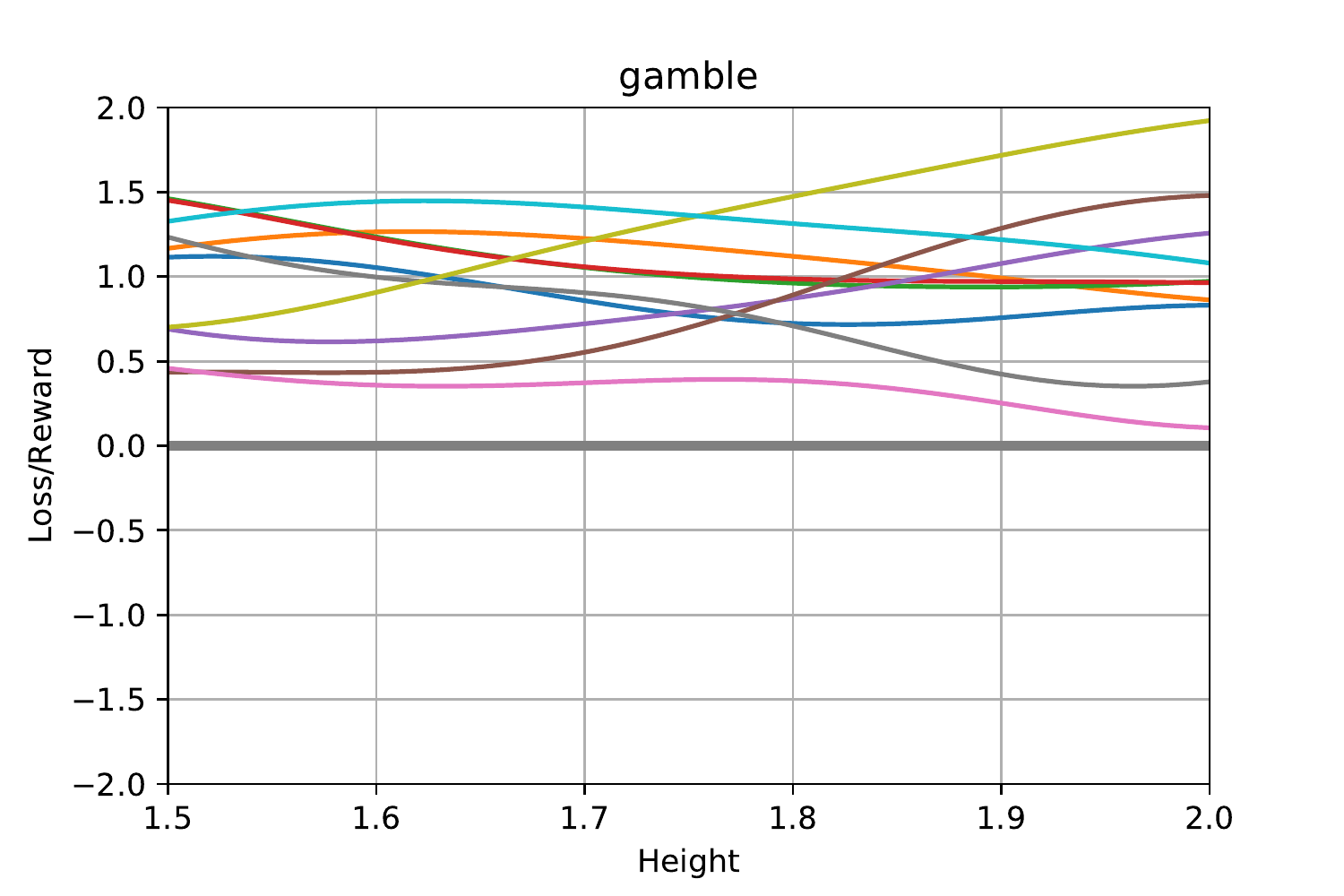}
\caption{}
\label{fig:sub1}
\end{subfigure}\\
\begin{subfigure}{.33\linewidth}
\centering
\includegraphics[scale=.33]{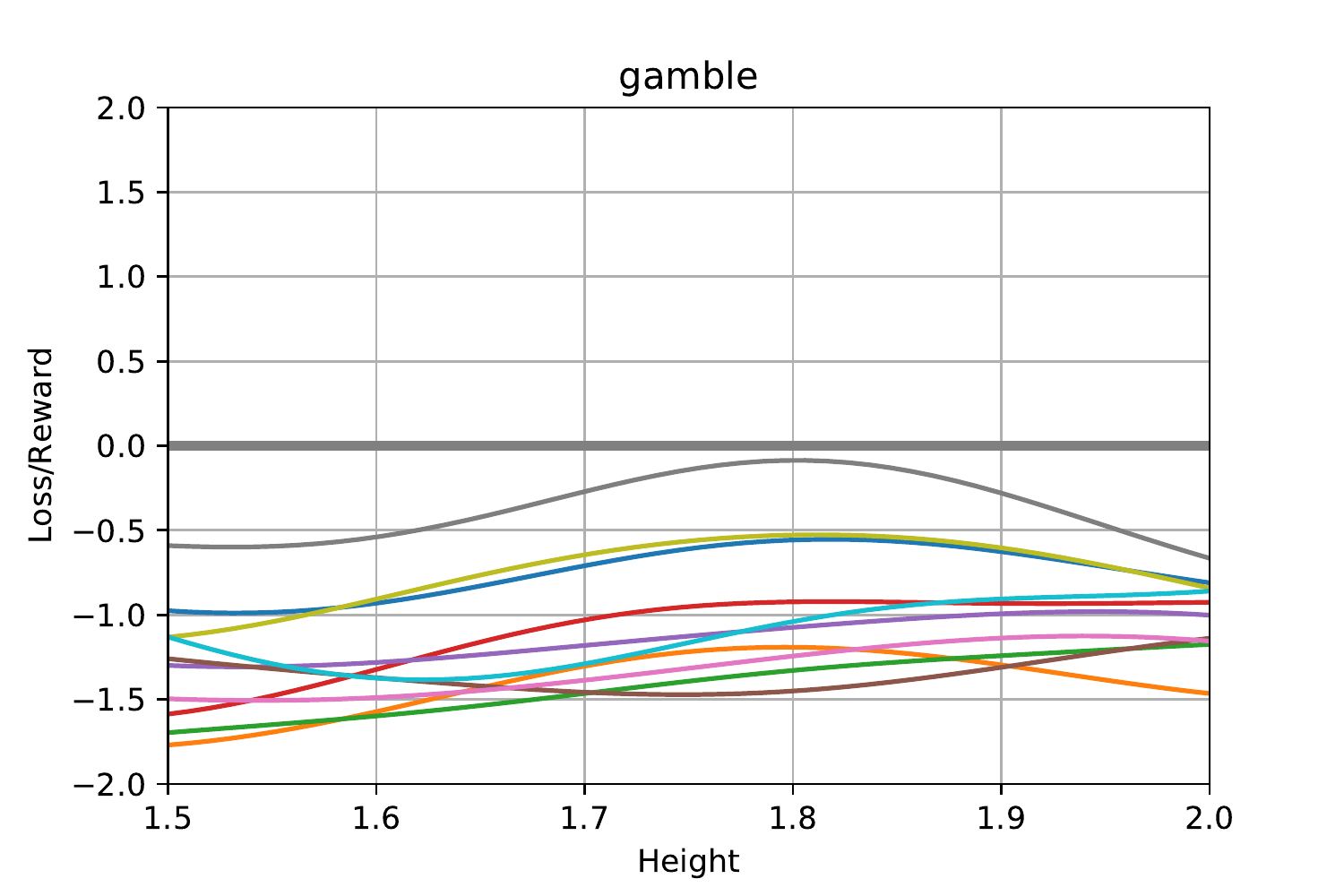}
\caption{}
\label{fig:sub1}
\end{subfigure}
\begin{subfigure}{.33\linewidth}
\centering
\includegraphics[scale=.33]{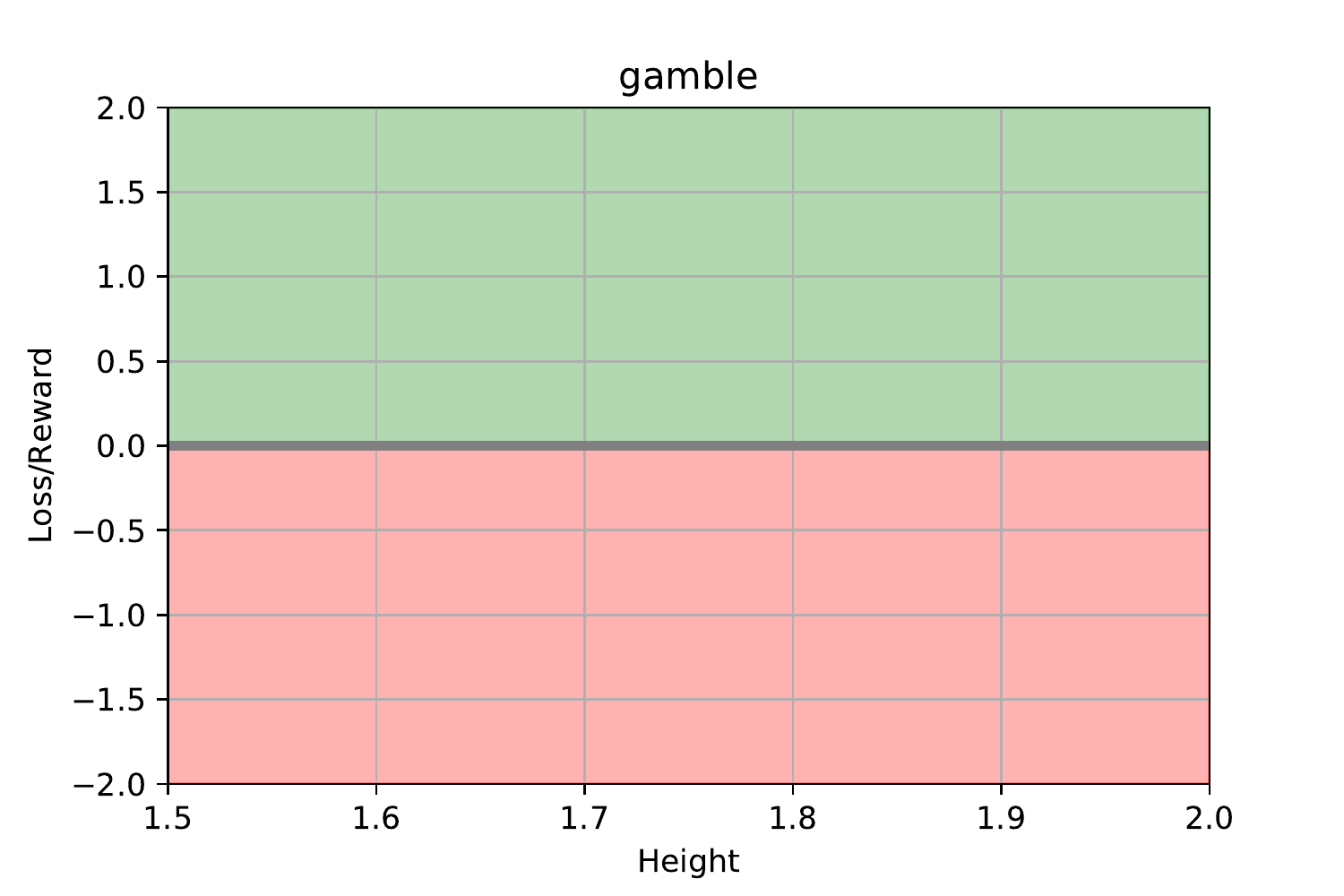}
\caption{}
\label{fig:sub1}
\end{subfigure}
\begin{subfigure}{.33\linewidth}
\centering
\includegraphics[scale=.33]{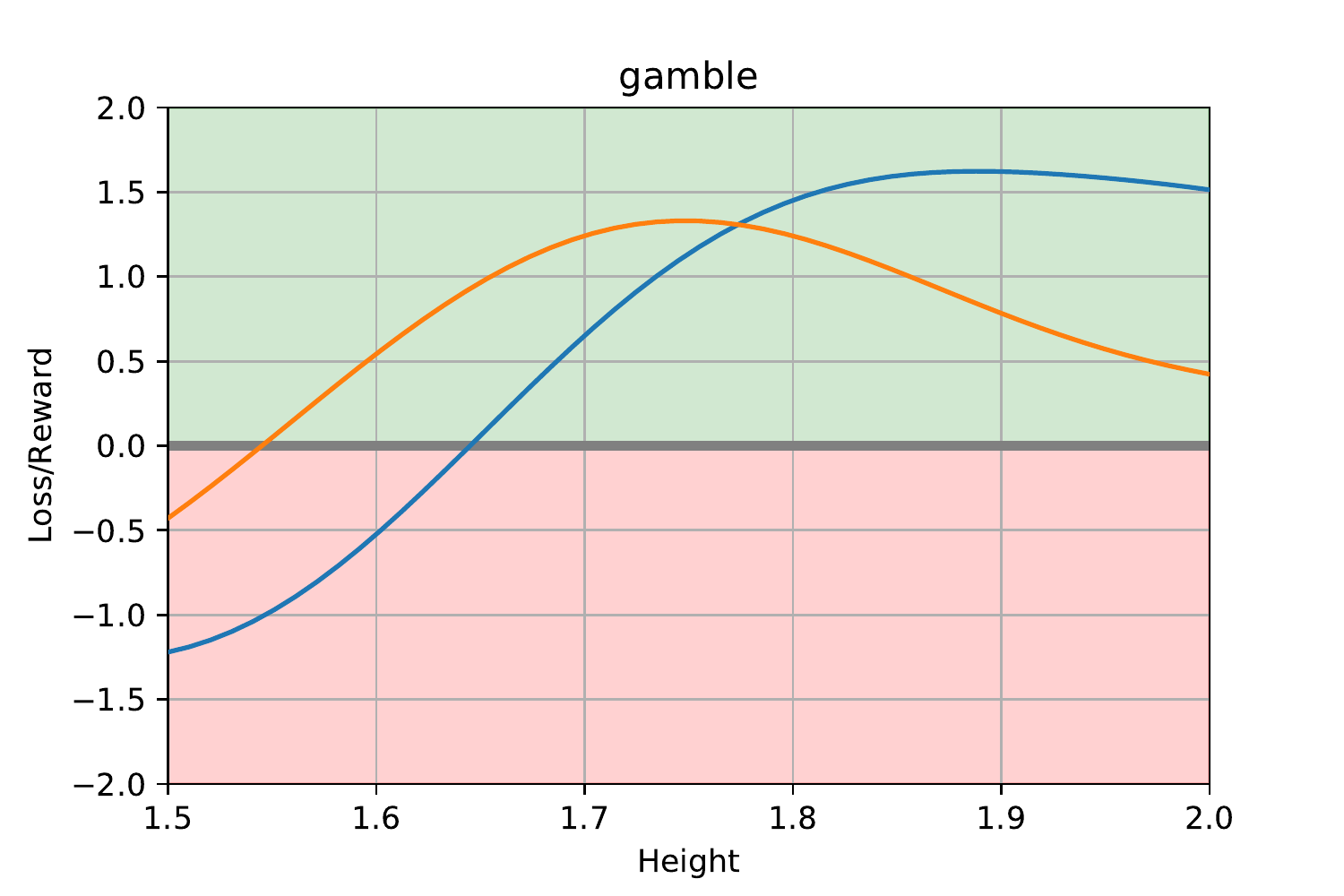}
\caption{}
\label{fig:sub1}
\end{subfigure}\\
\begin{subfigure}{.33\linewidth}
\centering
\includegraphics[scale=.33]{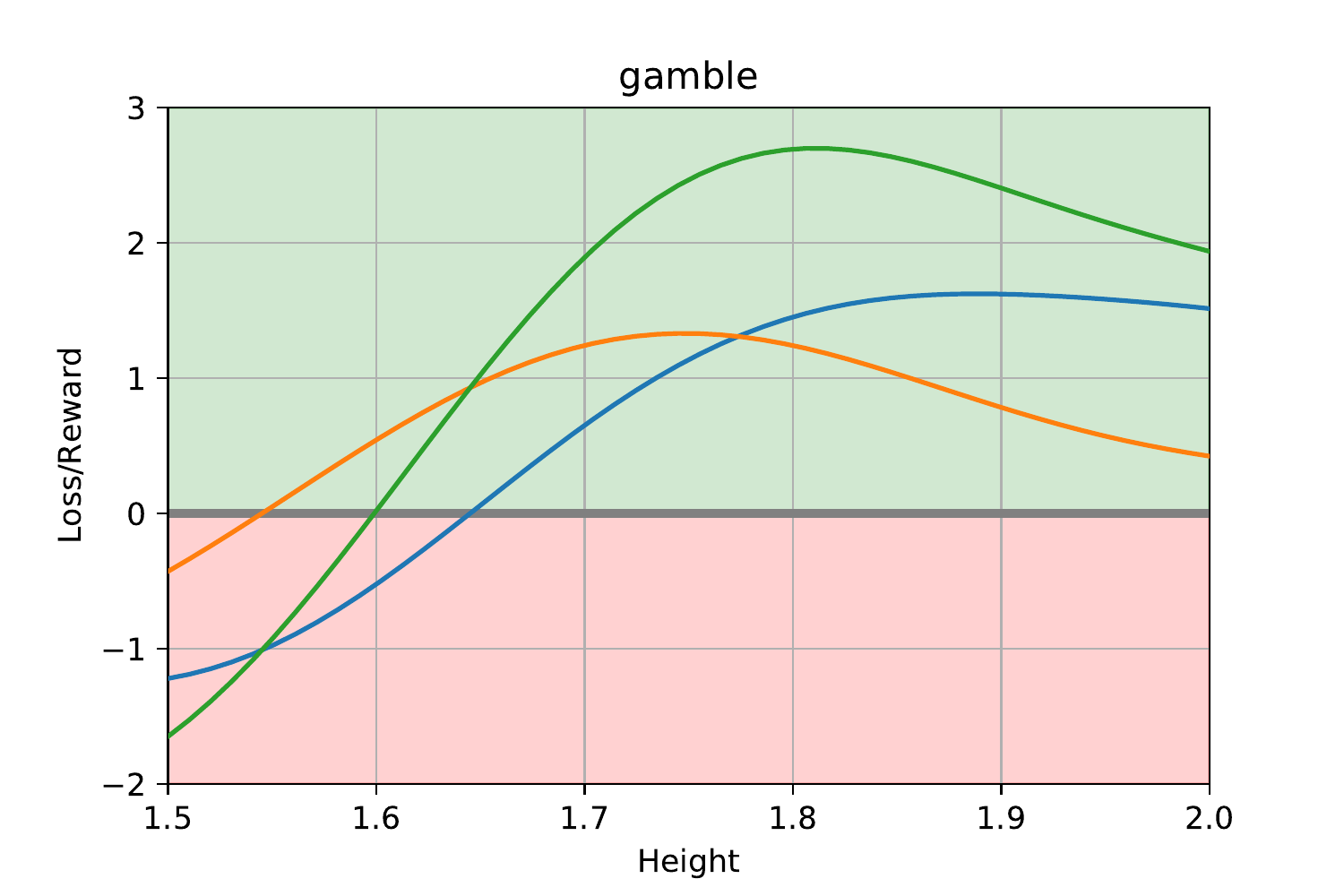}
\caption{}
\label{fig:sub1}
\end{subfigure}
\begin{subfigure}{.33\linewidth}
\centering
\includegraphics[scale=.33]{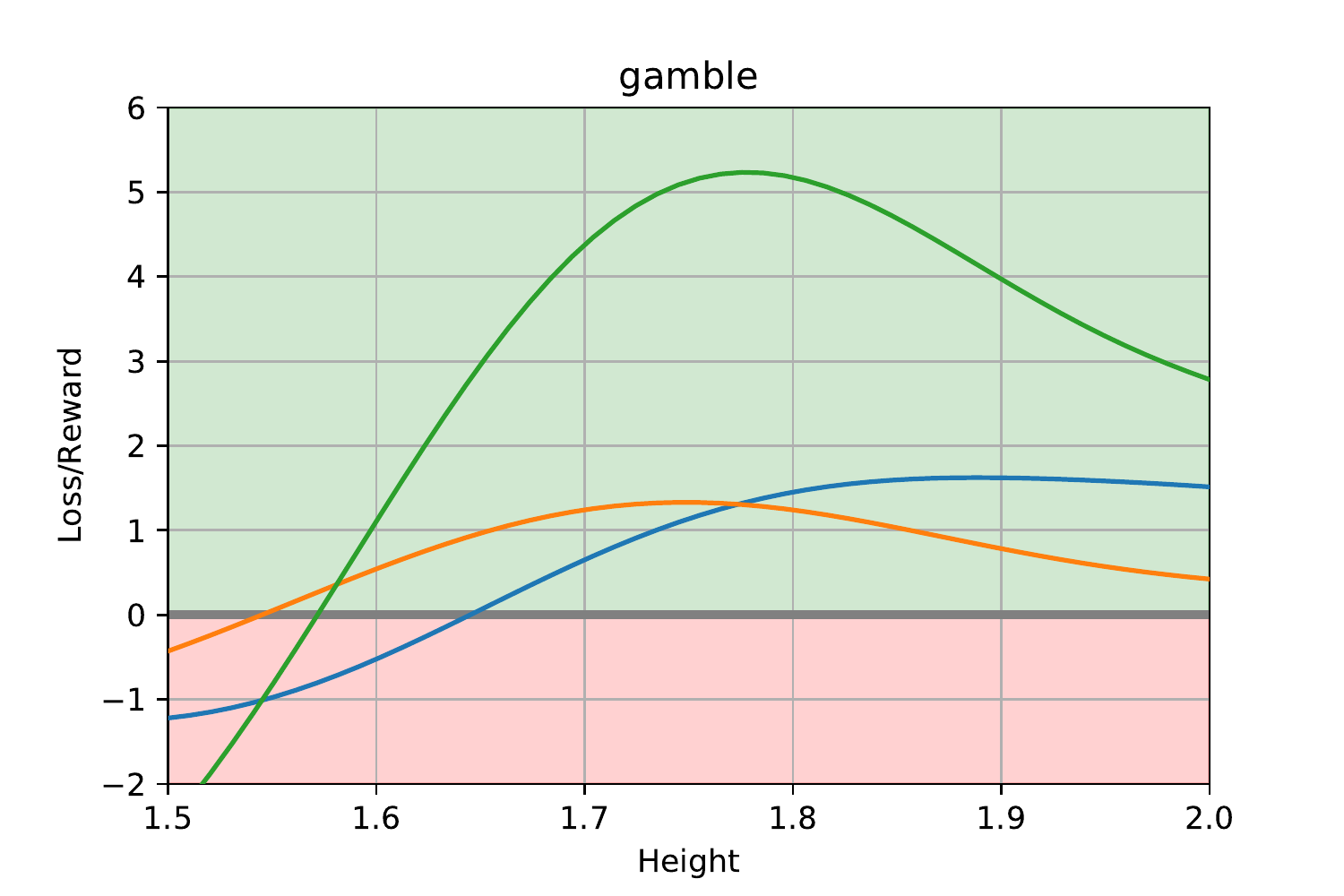}
\caption{}
\label{fig:sub1}
\end{subfigure}
\begin{subfigure}{.33\linewidth}
\centering
\includegraphics[scale=.33]{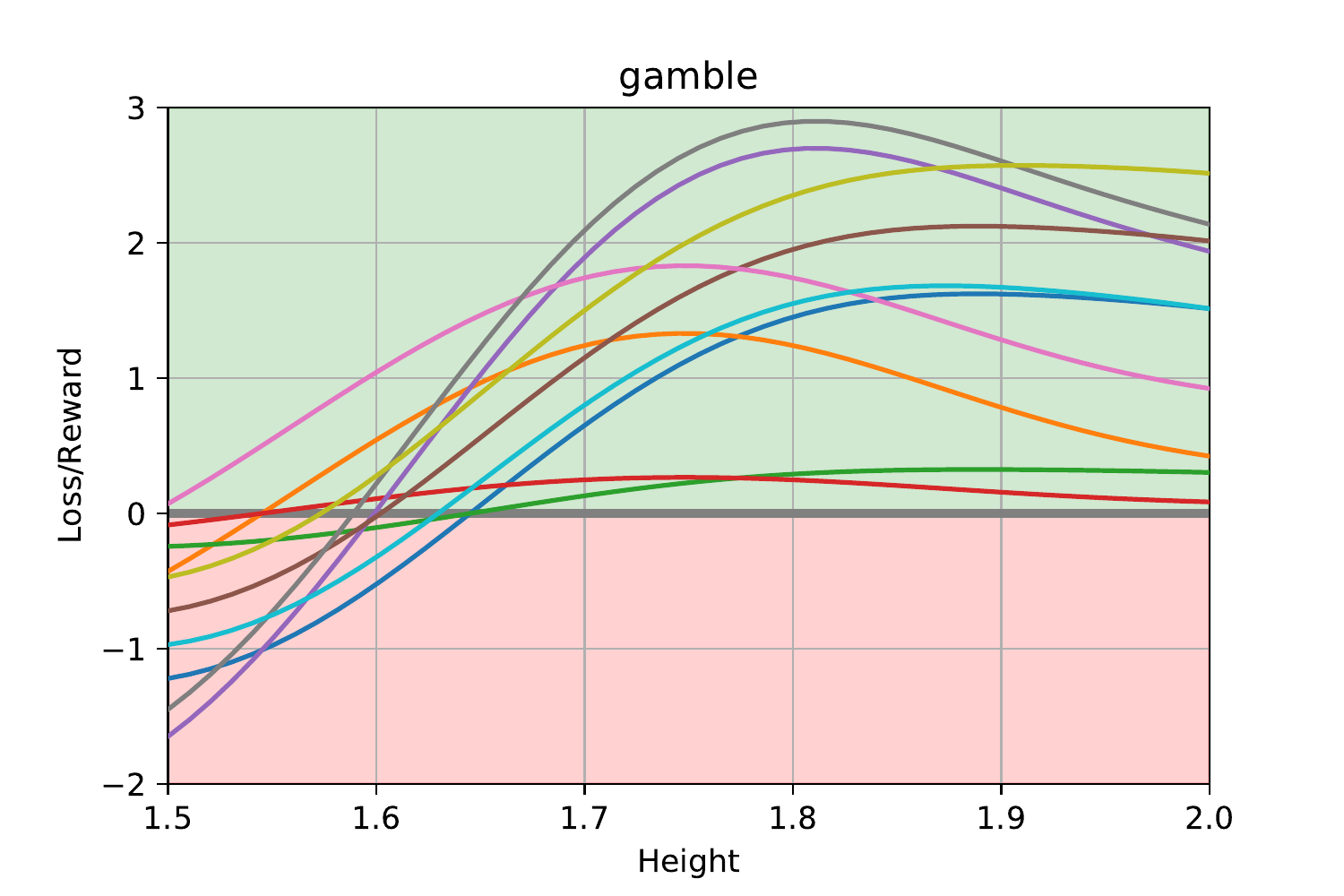}
\caption{}
\label{fig:sub1}
\end{subfigure}\\
\hspace{2cm}\begin{subfigure}{.5\linewidth}
\centering
\includegraphics[scale=.33]{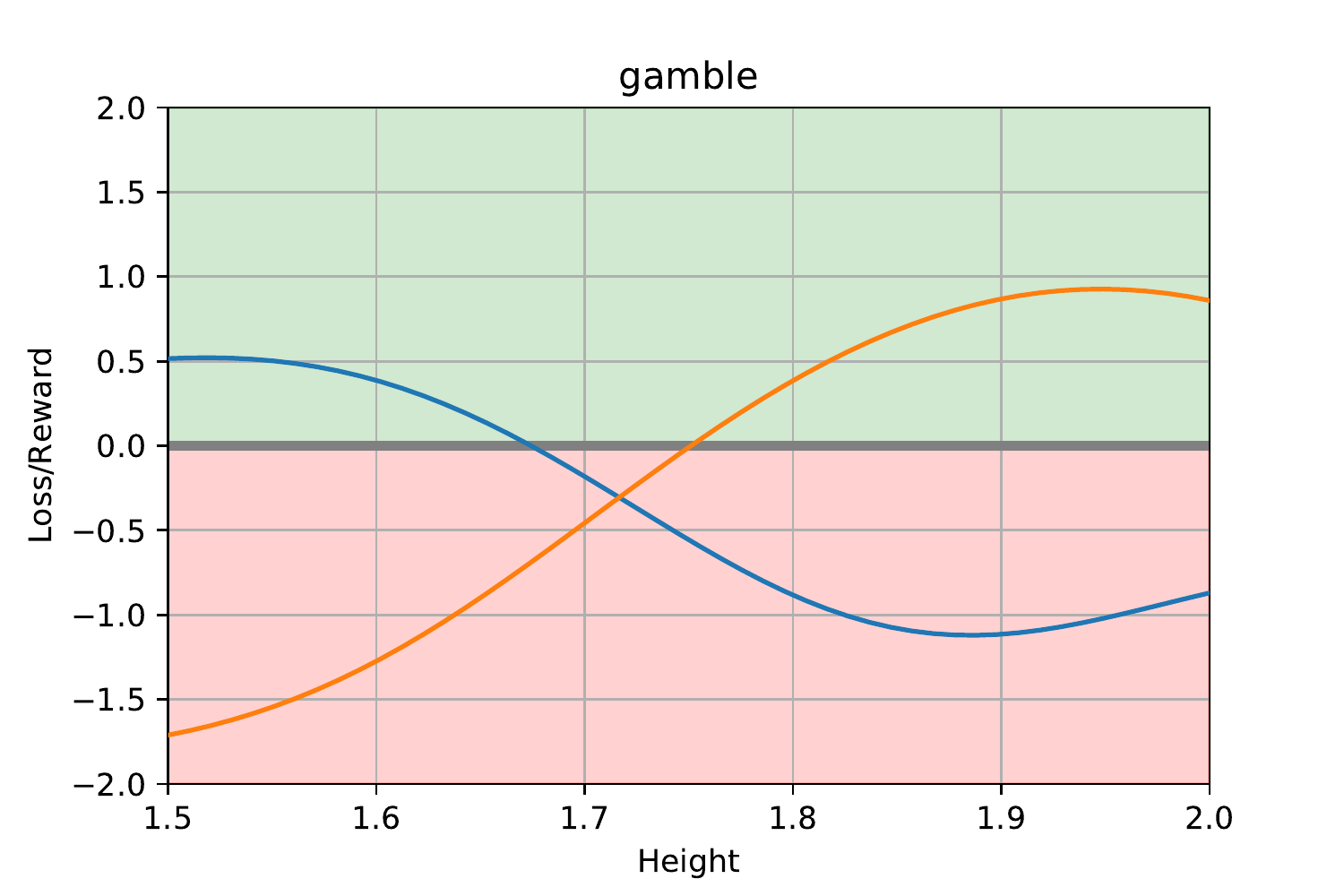}
\caption{}
\label{fig:sub1}
\end{subfigure}
\begin{subfigure}{.5\linewidth}
\centering
\includegraphics[scale=.33]{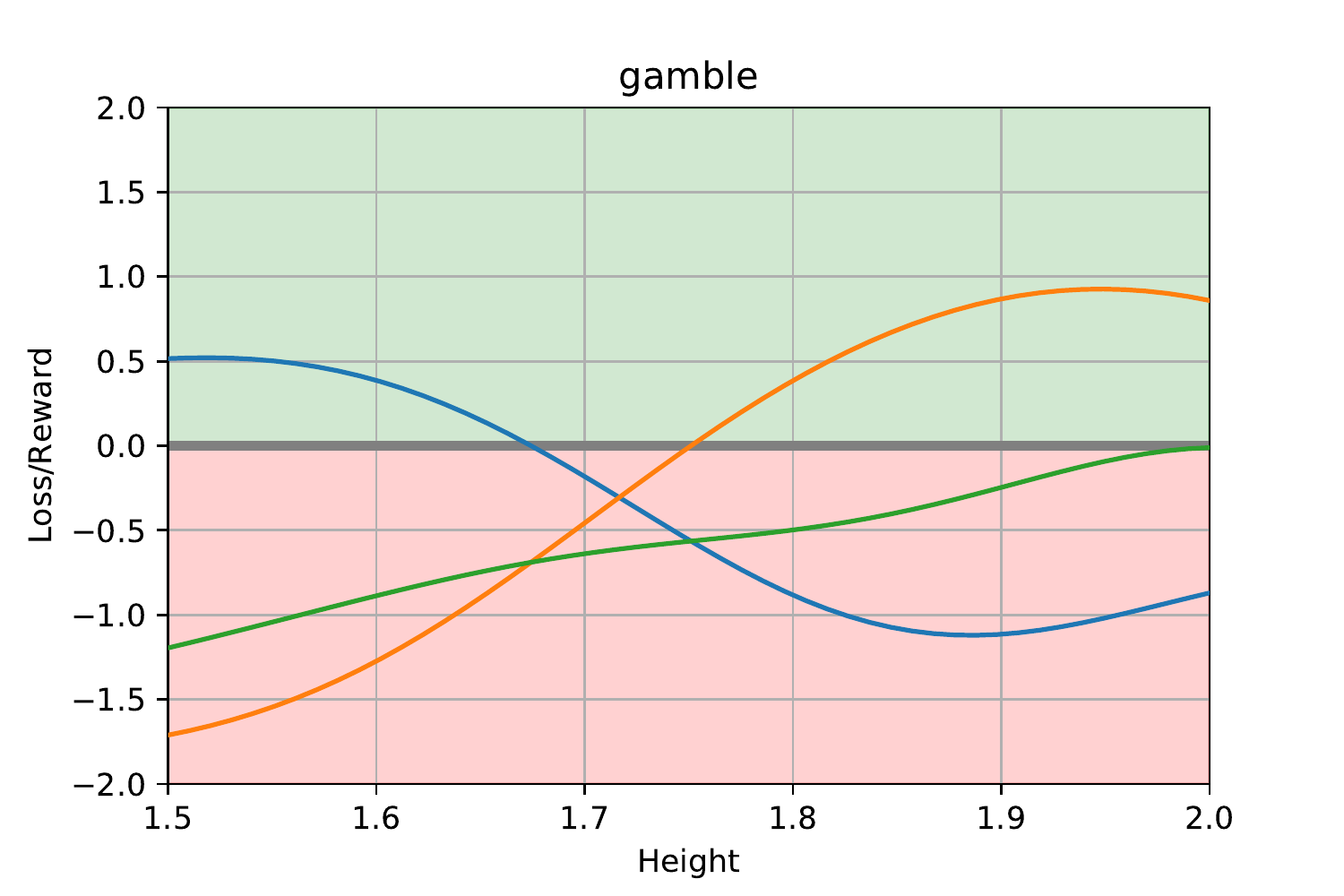}
\caption{}
\label{fig:sub1}
\end{subfigure}
\caption{How tall was Albert Einstein? Do you want to bet on it? Let us denote Einstein's height with $x$ and consider
the possibility space $\Omega=[1.5,2]m$. A gamble $g$ is a bounded function from $\Omega$ to the real numbers, an example in Plot~(a).
The meaning of $g$ is as follows: if you accept $g$ then  you for instance commit yourself to receive $0.75$ utiles if $x=1.8m$; to lose $0.75$ utiles if $x=1.63m$ etc., Plot~(b).\\
\ref{eq:taut} says that, if you are rational, you should accept any gamble in Plot~(c) because, no matter  Einstein's height, they may increase your wealth without ever decreasing it, that is all gambles with $g(x) \geq 0$. 
Similarly,  you should avoid  any gamble in Plot~(d) that will surely decrease your wealth, that is with  $g_i<0$ (this follows by \ref{eq:sl}), Plot~(d). 
Tautologies ($\gambles^{\geq}=\{g: g \geq 0\}$) and Falsum ($\gambles^{<}=\{g: g < 0\}$) are depicted in green and, respectively, red colours in Plot~(e).
Assume you have  accepted all gambles in the green area and the two gambles $g_1,g_2$, Plot~(f). Note that, the acceptance of $g_1,g_2$ depends on your beliefs about Einstein's height. If you   accept $g_1,g_2$  then you  should also accept: (i) $g_1(x)+g_2(x)$ Plot~(g); (ii) all the gambles $\lambda_1 g_1(x)+ \lambda_2 g_2(x)$  for any  $\lambda_i\geq0 $ Plot~(h);
(iii) all gambles  $\lambda_1 g_1(x)+ \lambda_2 g_2(x)+h(x)$ with $\lambda_i\geq0 $ and $h\in \gambles^{\geq}$, Plot~(i).
Plots~(g)--(i) follow  by \ref{eq:NE}.
Assume that instead of Plot~(f), you have  accepted the green area and $g_1,g_2$ in Plot~(j).
Then, you must also accept $g_1+g_2$ (because of \ref{eq:NE}).
However, $g_1+g_2$  is always negative, green function in Plot~(k). You always lose utiles in this case. In other words, by accepting  $g_1,g_2$ you incur a sure loss---\ref{eq:sl} is violated and so you are irrational.}
\label{fig:main}
\end{figure}

Postulate \ref{eq:sl}, which presupposes postulates \ref{eq:taut} and \ref{eq:NE},  provides the normative definition of TDG, referred to by $\theory$.  \rednew{Moreover, as simple as it looks, alone it captures the coherence postulate as formulated in the Introduction in case of classical probability theory.
To see this, assume $\domain$ is coherent. We give $\domain$ a probabilistic interpretation by observing that the mathematical dual of $\domain$ is the closed convex set:
\begin{equation}
\label{eq:dual}
\begin{aligned}
 \mathcal{P}=\left\{\mu  \in \mathsf{S} \Big| \int_{\pspace} g(\omega) d\mu(\omega)\geq0, ~\forall  g\in \mathcal{G}\right\},\\
\end{aligned}
\end{equation}
where $\mathsf{S}=\{ \mu \in \mathcal{M} \mid \inf \mu \geq0,~\int_{ \omega}  d\mu(\omega)=1\}$  is the set of all probability charges in $\pspace$, and $ \mathcal{M}$ the set of all charges (a charge is a finitely additive signed-measure) on $\pspace$. 
\bluenew{The derivations are given in Appendix~\ref{subsec:dual}. Observe that the term ``charge'' is used in Analysis to denote a finitely additive set function \citep[Ch.11]{aliprantisborder}. Conversely a measure is a countably additive set function.\footnote{De Finetti considered only finite additive probabilities,
while Kolmogorov considered sigma additive probabilities.} In this paper we use charges to be more general, but this does not really matter
for the results about QT that we are going to present later on.}
Hence, whenever an agent is coherent, Equation~\eqref{eq:dual} states that desirability corresponds to non-negative expectation (that is $\int_{\pspace} g(\omega) d\mu(\omega)\geq0$ for all probabilities in $\mathcal{P}$). When $\domain$ is incoherent,  $\mathcal{P}$ turns out to be empty---there is no probability compatible with the assessments in $\domain$. Stated otherwise, satisfying the axioms of classical probability---that is being a non-negative function that integrates to one--- is tantamount of being in the dual of a set $\domain$ satisfying the coherence postulate \ref{eq:sl}. }

\subsection{Computation postulate}\label{sec:comp}

The problem of checking whether $\domain$ is coherent or not can be formulated as the following decision problem:
\begin{equation}
\label{eq:dec}
\begin{aligned}
 \exists\lambda_i\geq0:-1-\sum\limits_{i=1}^{|\mathcal{G}|} \lambda_i g_i \in \gambles^{\geq}.
\end{aligned}
\end{equation}
If the answer is ``yes'', then the gamble $-1$ belongs to $\domain$, proving $\domain$'s incoherence. Actually any inference task can ultimately be reduced to a problem of the form~\eqref{eq:dec}, as discussed in Appendix~\ref{sec:inferTDG}. 
Hence, the above decision problem  unveils  a crucial fact: the hardness of inference in classical probability corresponds to the hardness of evaluating the non-negativity of a function in the considered space (let us call this the ``non-negativity decision problem'').

When $\pspace$ is infinite (in this paper we consider the case  $\pspace \subset\complex^n$) and for generic functions, the non-negativity decision problem is undecidable. To avoid such an issue, we may impose  restrictions on the class of allowed gambles and thus define $\theory$  on a appropriate subspace $\gambles_R$ of $\gambles$ (see Appendix~\ref{app:comp}). For instance, instead of $\gambles$, we may consider $\gambles_R$: the class of multivariate polynomials of degree at most $d$ (we denote by $\gambles_R^{\geq}\subset\gambles_R$ the subset of non-negative polynomials and by $\gambles_R^{<}\subset\gambles_R$ the negative ones). In  doing so, by Tarski-Seidenberg quantifier elimination theory \citep{tarski1951decision,seidenberg1954new}, the decision problem becomes decidable, but still intractable, being in general NP-hard.\footnote{\bluenew{By Stone-Weierstrass theorem, multivariate-polynomials are dense in the space of  continuous function in $[0,1]^n$. Therefore, multivariate-polynomials are a good choice for $\gambles_R$ because they guarantee both expressiveness and decidability.}} If we  accept that P$\neq$NP and we require that inference should be tractable  (in P), we are stuck. 
What to do?
\rednew{The solution, as advocated in this paper,}   
is to change the meaning of ``being non-negative'' for a  function by considering a subset $\bnonnegative \subsetneq \nonnegative_R$ for which the membership problem in \eqref{eq:dec} is in P. 

In other words, a computationally efficient TDG, which we denote by $\btheory$, 
should be based on a logical redefinition of the tautologies, i.e., by stating that
\begin{enumerate}[label=\upshape B0.,ref=\upshape B0]
\item\label{eq:btaut} $\bnonnegative$ should always be desirable,
\end{enumerate}
in the place of~\ref{eq:taut}. The rest of the theory can develop following the footprints of $\theory$. In particular, the deductive closure for $\btheory$ is defined by:
\begin{enumerate}[label=\upshape B1.,ref=\upshape B1]
\item\label{eq:b1} $\bdomain\coloneqq\posi(\bnonnegative \cup \mathcal{G})$.
\end{enumerate}
And the coherence postulate, which now naturally encompasses the computation postulate, states that:
\begin{definition}[P-coherence]
\label{def:bavs}
A set $\bdomain$ of desirable gambles is \emph{P-coherent}  if and only if\
\begin{enumerate}[label=\upshape B2.,ref=\upshape B2]
\item\label{eq:b2} $\bnegative \cap  \bdomain=\emptyset$,
\end{enumerate}
\bluenew{where $\bnegative=\text{Interior}(-\bnonnegative)$.}
\end{definition}
P-coherence owes its name to the fact that, whenever $\bnonnegative$ contains all positive constant gambles,~\ref{eq:b2} can be checked in polynomial time by solving:
\begin{equation}
\label{eq:bdec}
\begin{aligned}
\exists\lambda_i\geq0 ~~\text{ such that }~~ -1-\sum\limits_{i=1}^{|\mathcal{G}|} \lambda_i g_i \in \bnonnegative,
\end{aligned}
\end{equation}
where $-1$  denotes the constant function $f \in \gambles_R$ such that  $f(\omega)=-1$ for all $\omega \in  \pspace$.

\begin{figure}
\begin{subfigure}{.33\linewidth}
\centering
\includegraphics[scale=.33]{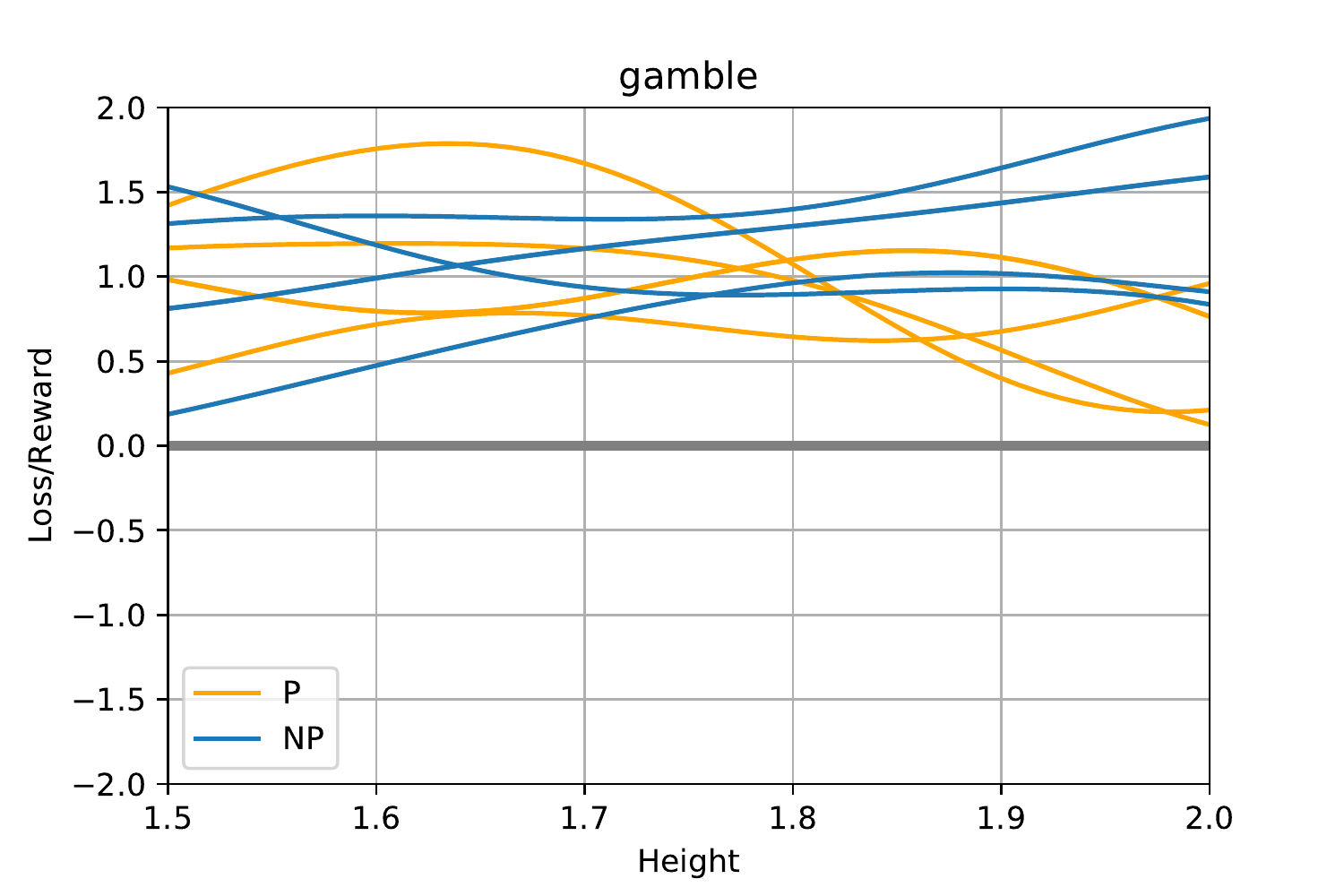}
\caption{}
\label{fig:sub1}
\end{subfigure}
\begin{subfigure}{.33\linewidth}
\centering
\includegraphics[scale=.33]{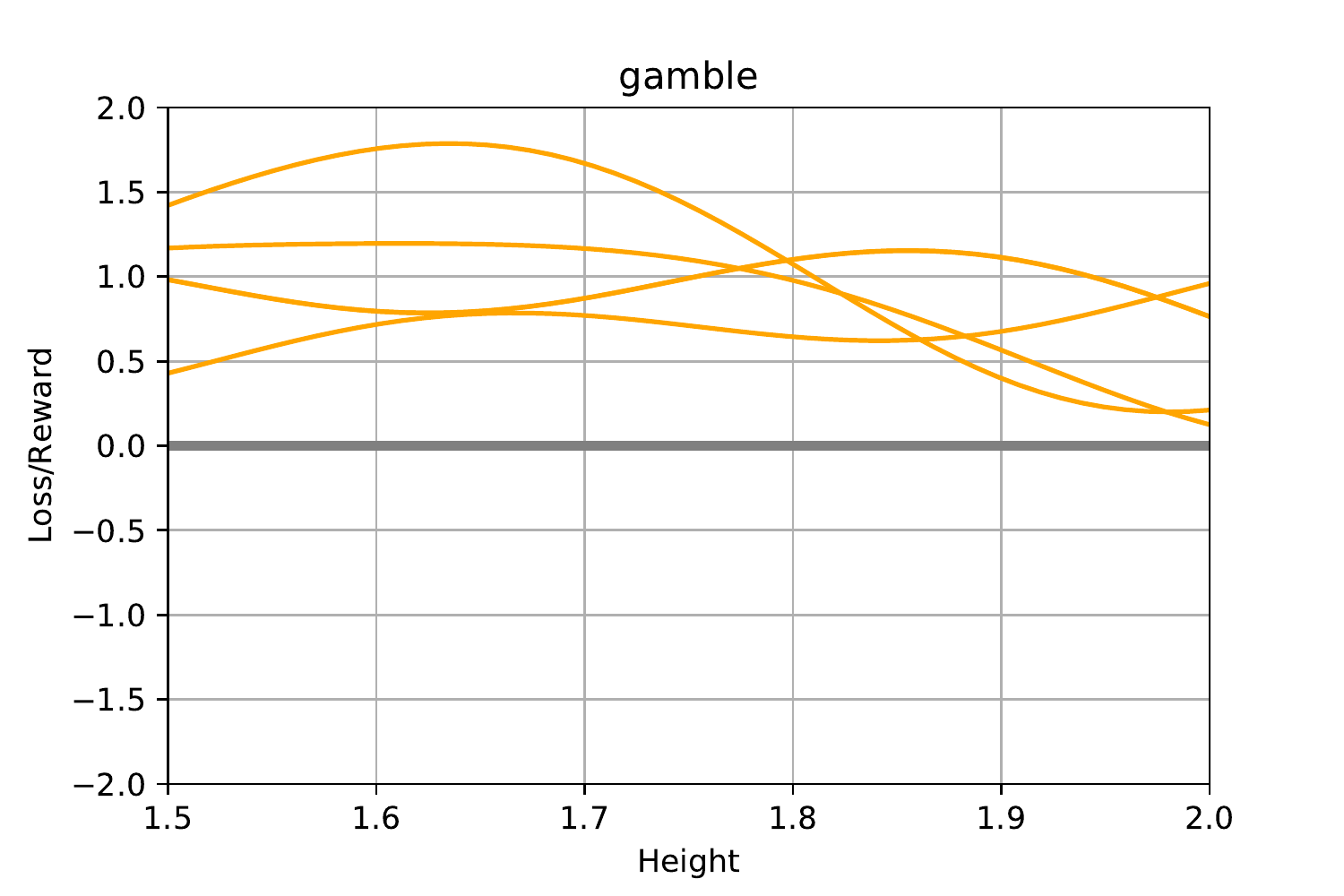}
\caption{}
\label{fig:sub1}
\end{subfigure}
\begin{subfigure}{.33\linewidth}
\centering
\includegraphics[scale=.33]{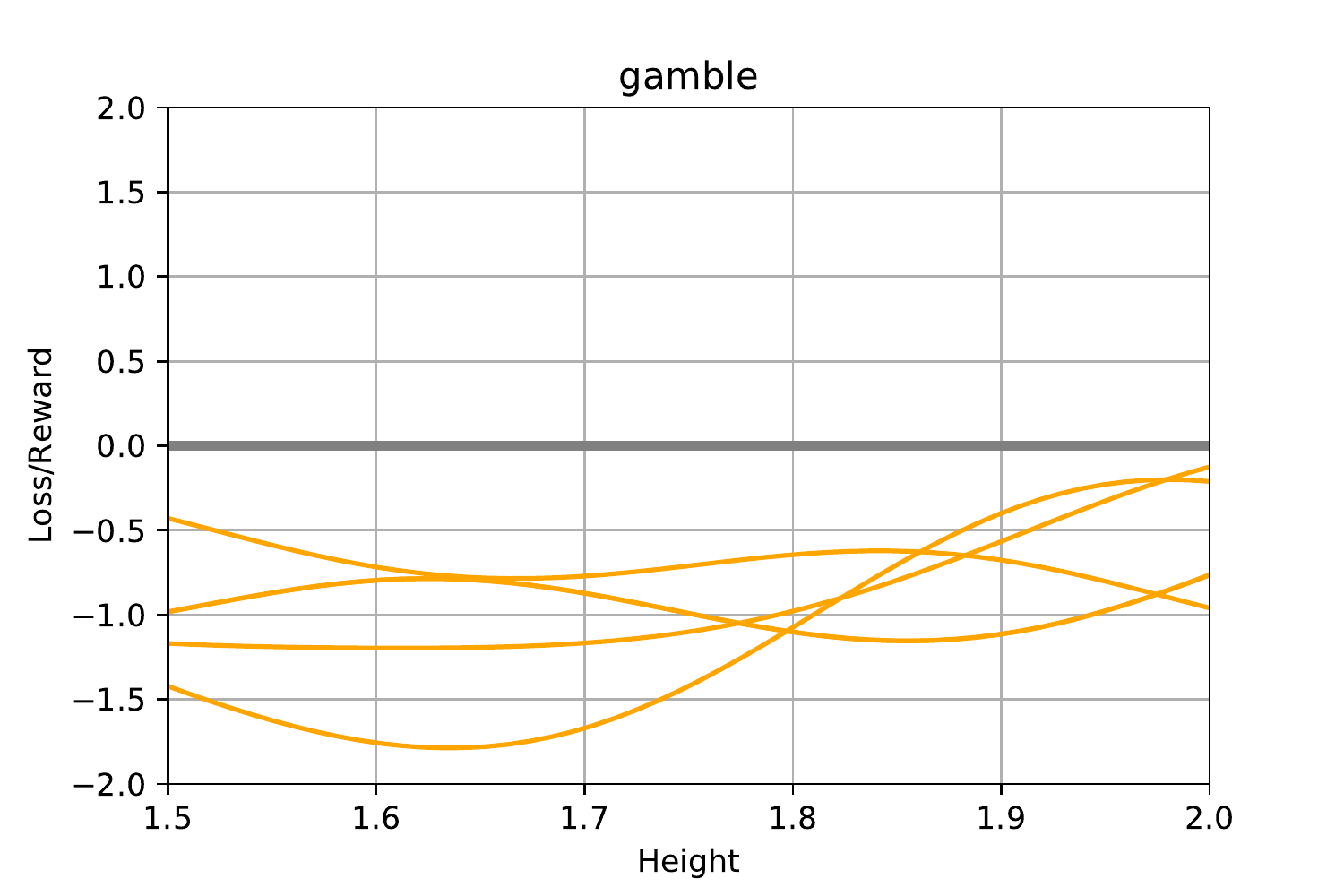}
\caption{}
\label{fig:sub1}
\end{subfigure}
\hspace{2cm}\begin{subfigure}{.5\linewidth}
\centering
\includegraphics[scale=.33]{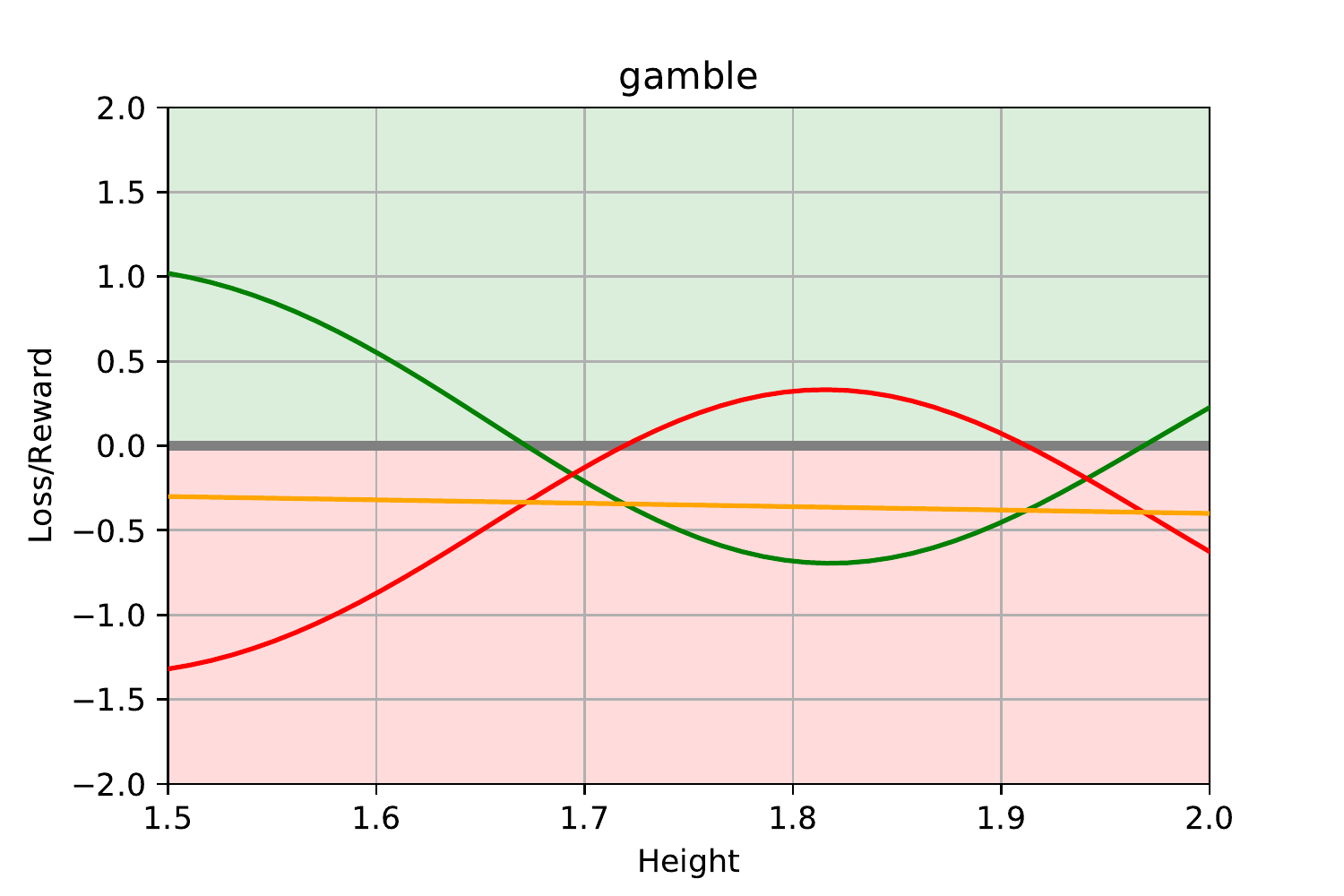}
\caption{}
\label{fig:sub1}
\end{subfigure}
\begin{subfigure}{.5\linewidth}
\centering
\includegraphics[scale=.33]{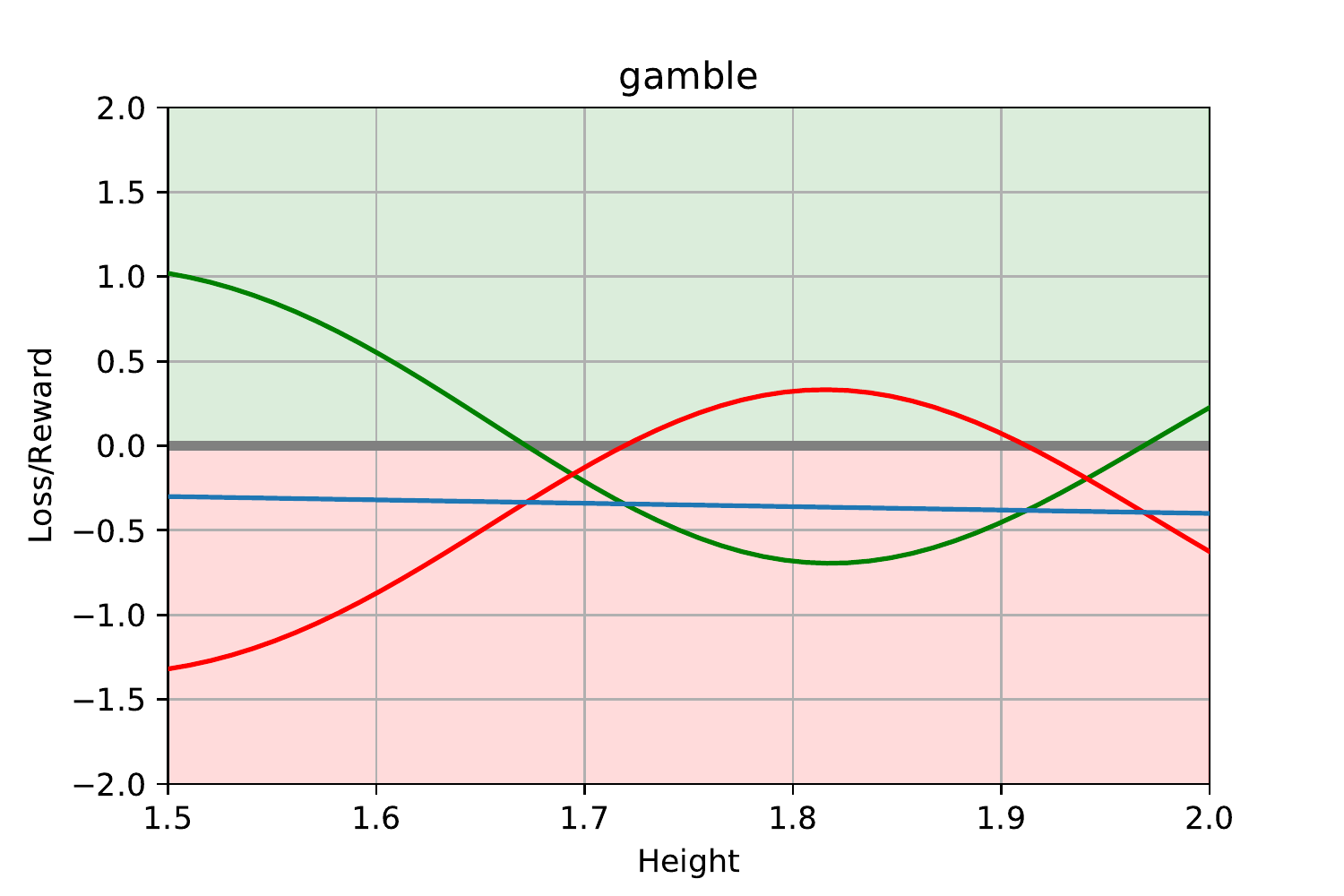}
\caption{}
\label{fig:sub1}
\end{subfigure}
\caption{Let us consider again Einstein's height example. Assume we can split the nonnegative gambles (tautologies) in two groups:
(i) the ones whose nonnegativity \textbf{can} be assessed in polynomial time (orange colour in Plot~(a));
(ii) the ones whose nonnegativity \textbf{cannot} be assessed in polynomial time (blue colour in Plot~(a)).
If you are a computationally rational agent, then you should surely accept all the orange gambles: they are nonnegative and you can evaluate their nonnegativity
in polynomial time.
P-coherence demands that you should accept all nonnegative orange gambles, Plot~(b), and avoid all negative orange gambles, Plot~(c).
Assume that  you have  accepted all nonnegative orange gambles and $g_1$ (red), $g_2$ (green) in Plot~(d).
Then, you must also accept $g_1+g_2$ (because of \ref{eq:b1}).
Note that, $g_1+g_2$  is always negative. However, according to P-coherence, you are (computationally) irrational only if  $g_1+g_2$  is of type orange, that is if you can evaluate in polynomial time that  $g_1+g_2$ is negative,  Plot~(d). In case $g_1+g_2$ is of type blue  Plot~(e), although you are accepting a gamble
that is negative, you are not computationally irrational. The reason is that you may not be able to evaluate its negativity. We will see that all ``quantum paradoxes'' arise when we are in a situation like Plot~(e).}
\label{fig:main1}
\end{figure}

Hence, $\btheory$ and $\theory$ have  the same deductive apparatus; they just differ in the considered set of tautologies, and thus in their (in)consistencies.
An example that gives an intuition of the postulates is given in Figure \ref{fig:main1}.

Interestingly, as we did previously, we can associate a ``probabilistic'' interpretation to the calculus defined by \ref{eq:btaut}--\ref{eq:b2} by computing the dual of a P-coherent set. Since $\gambles_R$ is a topological vector space, we can consider its dual space $\gambles_R^*$ of all bounded linear functionals $L: \gambles_R \rightarrow \reals$. Hence, with the additional condition that linear functionals preserve the unitary gamble, the dual cone of a P-coherent $\bdomain\subset \gambles_R$  is given by

\begin{equation}
\label{eq:dualL}
\bdomain^\circ=\left\{L \in \gambles_R^* \mid L(g)\geq0, ~~L(1)=1,~\forall g \in \bdomain\right\}.
\end{equation}
To $\bdomain^\circ$ we can then associate its extension  $ \bdomain^\bullet$ in $\mathcal{M}$, that is, the set of all  charges
 on $\pspace$ extending an element in $\bdomain^\circ$. 
 In other words, we can attempt to write $L(g)$ as an ``expectation'', that is an integral with respect to a  charge, i.e., $L(g)=\int_{\pspace} g(\omega)d\mu(\omega)$.
 In general however this set does not yield a classical probabilistic interpretation to $\btheory$. This is because, whenever
$\bnegative \subsetneq \negative_R$, there are negative gambles that cannot be proved to be negative in polynomial time:
  \begin{theorem}\label{th:fundamental}
  Assume that $\bnonnegative$ includes all positive constant gambles and that it is closed (in $\gambles_R$). 
Let $\bdomain \subseteq \gambles_R$ be a P-coherent set of  desirable gambles. The following statements are equivalent:
\begin{enumerate}
   \item $\bdomain$ includes a negative gamble that is not in $\bnegative$.
\item $\posi(\nonnegative\cup \mathcal{G})$ is incoherent, and thus $\mathcal{P}$ is empty.
    \item $\bdomain^{\circ}$ is not (the restriction to  $\gambles_R$ of) a closed convex set of mixtures of classical evaluation functionals.\footnote{\bluenew{
    Here ``closed'' is with respect to the weak$^*$-topology,  that is the coarsest topology on the dual space making the evaluation functions  continuous.
    Note also that evaluation functionals or, more in general the elements of the dual space, can be informally interpreted as simply states.}}
    \item The extension   $ \bdomain^\bullet$ of $\bdomain^{\circ}$ in the space $\mathcal{M}$ of all charges in $\pspace$ includes only  charges that are not probabilities (they have some negative value).
\end{enumerate}
 \end{theorem}

Theorem~\ref{th:fundamental} is the central result of this paper (the proof is in Appendix~\ref{app:complex}). It states that whenever $\bdomain$ includes a negative gamble (item~1), there is no classical probabilistic interpretation for it (item~2). The other points suggest alternative solutions to overcome this deadlock: either to change the notion of evaluation functional (item~3) or to use quasi-probabilities (probability distributions that admit negative values) as a means for interpreting $\btheory$ (item~4).
The latter case means that, when we write $L(g)=\int_{\pspace} g(\omega)d\mu(\omega)$, then $\mu(\omega)$ satisfies
$1=L(1)=\int_{\pspace} d\mu(\omega)=1$ but it is not a probability  charge.

In what follows, we are going to show that QT can be deduced from a particular instance of the theory $\btheory$. As a consequence, we get that the computation postulate, and in particular \ref{eq:btaut}, 
is 
the unique reason for all its paradoxes, which all boil down to a rephrasing of the various statements of Theorem~\ref{th:fundamental} in the considered quantum context.

\section{QT as computational rationality}\label{sec:qt}
We briefly recall the framework discussed in the Introduction. Consider first a single particle system with $n$-degrees of freedom and  
$$
\overline{\complex}^{n}\coloneqq\{ x\in \complex^{n}: ~x^{\dagger}x=1\}.
$$
We can interpret an element $\tilde{x} \in \overline{\complex}^{n}$ as ``input data'' for some classical preparation procedure.
For instance,  in the case of the spin-$1/2$ particle ($n = 2$),  if $\theta = [\theta_1 , \theta_2 , \theta_3 ]$ is the direction of a filter
in the Stern-Gerlach experiment, then $\tilde{x}$ is its one-to-one mapping into $\overline{\complex}^{2}$ (apart from a phase term).
For spin greater than $1/2$, the variable $\tilde{x} \in \overline{\complex}^{n}$ 
 cannot directly be interpreted in terms only of ``filter direction''.
Nevertheless,  at least on the formal level, $\tilde{x}$ plays the role of  a ``hidden variable'' in our model
and $\overline{\complex}^{n}$ of the possibility space $\pspace$.
This hidden-variable model for QT is also discussed in \cite[Sec.~1.7]{holevo2011probabilistic}, where the author  explains
why this model does not contradict the existing ``no-go'' theorems for hidden-variables, see also Section \ref{sec:hidden1d} and \ref{sec:tensorproduct}.

In QT any real-valued observable  is described by a Hermitian operator.
This naturally imposes restrictions on the type of functions $g$ in \eqref{eq:dec}:
$$
g(x)=x^\dagger G x,
$$
where  $x \in \pspace$ and $G\in \He^{n\times n}$, with $\He^{n\times n}$ being the set of Hermitian matrices of dimension $n \times n$.
Since $G$ is Hermitian and  $x$ is bounded ($x^{\dagger}x=1$), $g$ is a real-valued bounded function ($g(x)=\expval{G}{x}$ in bra-ket notation).

More generally speaking, we can consider composite systems of $m$ particles, each one with $n_j$ degrees of freedom. The  possibility space is the Cartesian product $\pspace=\times_{j=1}^m \overline{\complex}^{n_j}$ and the functions are $m$-quadratic forms:
\begin{equation}
 \label{eq:gamble_many}
 g(x_1,\dots,x_m)=(\otimes_{j=1}^m x_j)^\dagger G (\otimes_{j=1}^m x_j),
\end{equation}
with $G \in \He^{n \times n}$, $n=\prod_{j=1}^m n_j$ and where $\otimes$ denotes the tensor product between vectors regarded as column matrices.
Notice that in our setting the tensor product is ultimately a derived notion, not a primitive one (see  Section \ref{sec:tensorproduct}), as it follows by the properties of $m$-quadratic forms.

For $m=1$ (a single particle), evaluating the non-negativity of the quadratic form $x^\dagger G x$ boils down to checking whether the matrix $G$ is Positive Semi-Definite (PSD) and therefore can be performed in polynomial time.
This is no longer true for $m\geq 2$: indeed,  in this case there exist polynomials of type \eqref{eq:gamble_many} that are non-negative, but whose matrix $G$ is indefinite (it has at least one negative eigenvalue).
Moreover, it turns out that problem \eqref{eq:dec} is not \emph{tractable}:

\begin{proposition}[\cite{gurvits2003classical}]
\label{prop:gurvitis}
The problem of checking the non-negativity of functions of type \eqref{eq:gamble_many} is NP-hard for $m\geq2$. 
\end{proposition}

What to do? As discussed previously,
the solution  is to change the meaning of ``being non-negative''   by considering a subset $\bnonnegative \subsetneq \nonnegative$ for which the membership problem, and thus \eqref{eq:dec}, is in P.
For  functions of type  \eqref{eq:gamble_many}, we can extend the notion of non-negativity that holds for a single particle 
to $m>1$ particles:
\begin{equation}
\label{eq:b0} \Sigma^{\geq}\coloneqq\{g(x_1,\dots,x_m)=(\otimes_{j=1}^m x_j)^\dagger G (\otimes_{j=1}^m x_j): G\geq0\}.
\end{equation}
That is, the function is ``non-negative'' whenever $G$ is PSD (note that $\Sigma^{\geq}$ is the so-called cone of \textit{Hermitian sum-of-squares} polynomials, see  Section~\ref{sec:sos}, and that, in $\Sigma^{\geq}$ the  non-negative constant functions have the form $g(x_1,\dots,x_m)=c (\otimes_{j=1}^m x_j)^\dagger I (\otimes_{j=1}^m x_j)$ 
with $c\geq0$).
Now, consider any set of desirable gambles $\bdomain$ satisfying \ref{eq:btaut}--\ref{eq:b2} with the given definition of \eqref{eq:b0}: Eureka! We have just derived the first postulate of QT (see Postulate~1 in \cite[p.~110]{nielsen2010quantum}):
 \begin{quote}
 \bluenew{ Associated to any isolated physical system is a complex vector space
with inner product (that is, a Hilbert space) known as the state space of the
system. The system is completely described by its density operator, which is a
positive operator $\rho$ with trace one, acting on the state space of the system.}
 \end{quote}

Indeed, let $\mathcal{G}$ be a finite set of assessments, and $\domain$ the deductive closure as defined by~\ref{eq:b1}; it is not difficult to prove that the dual of $\domain$  is
\begin{align}
\label{eq:credaldef}
\mathscr{Q}&=\{ \rho \in \mathscr{S} \mid Tr(G \rho) \geq  0,~ ~\forall g \in \mathcal{G}\},
\end{align} 
where  $\mathscr{S}=\{ \rho \in \He^{n\times n} \mid \rho\geq0,~~Tr(\rho)=1\}$ is the set of all density matrices. As before, whenever the set $\bdomain$ representing Alice's beliefs about the experiment is coherent, Equation \eqref{eq:credaldef} means that desirability implies non-negative ``expected value'' for all models in $\mathscr{Q}$.
Note that in QT the expectation of $g$ is $Tr(G \rho)$. This follows by Born's rule, a law giving the probability that a measurement on a quantum system will yield a given result.

The agreement with Born's rule is an important constraint in any alternative axiomatisation of QT. \rednew{This is also the case of our theory, but in the sense that Born's rule can be derived from it}.  In fact, in the  view of a density matrix as a  dual operator, $\rho$ is formally equal to
\begin{equation}
  \label{eq:momepr}
\rho=L\left((\otimes_{j=1}^m x_j)(\otimes_{j=1}^m x_j)^\dagger\right),
\end{equation}
with $L$ defined in \eqref{eq:bdec}. 
\rednew{
\begin{example}
 Consider the case $n=m=2$, then
 \begin{equation}
\label{eq:exlinearoperator}
L\left((\otimes_{j=1}^2 x_j)(\otimes_{j=1}^2 x_j)^\dagger\right)= L\left(\left[\begin{smallmatrix}x_{11} x_{11}^{\dagger} x_{21} x_{21}^{\dagger} & x_{11}^{\dagger} x_{12} x_{21} x_{21}^{\dagger} & x_{11} x_{11}^{\dagger} x_{21}^{\dagger} x_{22} & x_{11}^{\dagger} x_{12} x_{21}^{\dagger} x_{22}\\x_{11} x_{12}^{\dagger} x_{21} x_{21}^{\dagger} & x_{12} x_{12}^{\dagger} x_{21} x_{21}^{\dagger} & x_{11} x_{12}^{\dagger} x_{21}^{\dagger} x_{22} & x_{12} x_{12}^{\dagger} x_{21}^{\dagger} x_{22}\\x_{11} x_{11}^{\dagger} x_{21} x_{22}^{\dagger} & x_{11}^{\dagger} x_{12} x_{21} x_{22}^{\dagger} & x_{11} x_{11}^{\dagger} x_{22} x_{22}^{\dagger} & x_{11}^{\dagger} x_{12} x_{22} x_{22}^{\dagger}\\x_{11} x_{12}^{\dagger} x_{21} x_{22}^{\dagger} & x_{12} x_{12}^{\dagger} x_{21} x_{22}^{\dagger} & x_{11} x_{12}^{\dagger} x_{22} x_{22}^{\dagger} & x_{12} x_{12}^{\dagger} x_{22} x_{22}^{\dagger}\end{smallmatrix}\right]\right),
\end{equation}
and so
$\rho_{11}=L(x_{11} x_{11}^{\dagger} x_{21} x_{21}^{\dagger})$, $\rho_{12}=L(x_{11}^{\dagger} x_{12} x_{21} x_{21}^{\dagger} )$ etc..
\end{example}}

Hence, when a projection-valued measurement characterised by the projectors $\Pi_1,\dots,\Pi_n$ is considered, it holds that
$$
L( (\otimes_{j=1}^m x_j)^\dagger \Pi_i (\otimes_{j=1}^m x_j))=Tr(\Pi_i  L((\otimes_{j=1}^m x_j)(\otimes_{j=1}^m x_j)^\dagger))=Tr(\Pi_i \rho ).
$$
Since $\Pi_i\geq0$ and the polynomials $(\otimes_{j=1}^m x_j)^\dagger \Pi_i (\otimes_{j=1}^m x_j)$ for $i=1,\dots,n$ form a partition of unity, i.e.:
$$
\sum_{i=1}^n (\otimes_{j=1}^m x_j)^\dagger \Pi_i (\otimes_{j=1}^m x_j)=  (\otimes_{j=1}^m x_j)^\dagger I (\otimes_{j=1}^m x_j)=1,
$$
we have that
$$
Tr(\Pi_i \rho )\in[0,1] \text{ and } \sum_{i=1}^n Tr(\Pi_i \rho )=1.
$$
For this reason,  $Tr(\Pi_i \rho )$ is usually interpreted as a probability. But \emph{the projectors $\Pi_i$'s are not indicator functions}, whence, strictly speaking, the traditional interpretation is incorrect. This can be seen clearly in the special case where postulates~\ref{eq:taut} and~\ref{eq:btaut} coincide, as in the case of a single particle, that is, where the theory can be given a classical probabilistic interpretation, see Section \ref{sec:momentmat}. In such a case, the corresponding $\rho$ is just a (\emph{truncated}) \emph{moment matrix}, i.e., one for which there is
at least one probability  such that $E[xx^\dagger]=\rho$. In summary, our standpoint here is that $Tr(\Pi_i \rho )$ should rather be interpreted as the expectation of the $m$-quadratic form $x^\dagger \Pi_i x$. This makes quite a difference with the traditional interpretation since in our case there can be (and usually there will be) more than one charge compatible with such an expectation, as we will point out more precisely later on.

 \bluenew{Quantum measurements are discrete: when we perform a measurement, we observe a detection (along one of the directions $ \Pi_i$).
This phenomenon of \emph{quantisation} is one of the major differences between quantum and classical physics.
We took it into account in the choice of the framework,  the possibility space being (only) the ``directions''  of the particle's spins
 and the measurement apparatus sensing only certain fixed ``directions''   ($x^\dagger \Pi_i x=x^\dagger v_iv_i^{\dagger} x$ is a function of two ``directions'' $x$ and $v_i$).  
Despite its centrality, we want however to point out that quantisation is \emph{not} the source of Bell-like inequalities and entanglement.  As said before, this is because ``quantum weirdness'' is intrinsic to any theory of computational rationality, and  is thence not confined to QT only.}


\begin{remark}
 \bluenew{It is often claimed that QT  includes classical probability theory (CPT) as a special case, or better 
  that QT includes \emph{discrete} CPT.
 In the present framework, this is due to properties of quadratic forms. Indeed $x^\dagger \Pi_1 x,\dots,x^\dagger \Pi_n x$
 form a partition of unity, and therefore $E[x^\dagger \Pi_i x]=Tr(\Pi_i\rho)=p_i$, whenever 
 $\rho=E[xx^\dagger]=\sum_{i=1}^{n}p_i \Pi_i$ (with $p_i\geq0$ and $\sum_{i=1}^{n} p_i=1$).
\\
On the contrary, as the possibility space $\Omega$ is infinite (e.g., the ``directions''  of the particle's spins), in this paper when we speak about CPT (and compare it with QT), we mean \emph{continuous}  classical probability theory (in the complex numbers).
Hence again, since both~\ref{eq:b1},\ref{eq:b2} and~\ref{eq:NE},\ref{eq:sl} are the same logical postulates parametrised by the appropriate meaning of ``being negative/non-negative'', the only axiom truly separating (continuous) classical probability theory from the quantum one is~\ref{eq:btaut} (with the specific form of~\eqref{eq:b0}), thus implementing the requirement of computational efficiency.\\
In other words,  we claim that QT is ``easier'' than CPT because, as the possibility space $\Omega$ is infinite, inference in \underline{continuous} CPT is NP-hard.
In QT, we realize that, when we try to address the question whether or not an experimentally generated  state  is  entangled
(many quantum  technologies  rely  on  the  presence  of entanglement). We will discuss in  Section \ref{sec:entang} that  determining entanglement of a general state
is equivalent to prove the nonnegativity of a polynomial that, as we  discussed in Proposition \ref{prop:gurvitis}, is NP-hard.  In fact, we can reformulate the Entanglement Witness Theorem
as the clash between the  classical notion of coherence and P-coherence, see Theorem \ref{th:ewt}.
}
 \end{remark}

\rednew{By providing a connection between the present work and \cite{benavoli2016quantum}, in Section \ref{subsec:otherax}  we  discuss how to get  the other postulates and rules of QT:  L\"{u}ders rule (measurement updating) and Schr\"odinger rule (time evolution).}

\subsection{Truncated moment matrices vs.\ density matrices}
\label{sec:momentmat}
In a single particle system of dimension $n$, $\rho=L(x x^{\dagger})$.
 In such case, $ \rho$ can be interpreted as a truncated moment matrix, i.e., there exists a probability distribution on the complex vector variable $x\in\Omega$ such that
 \begin{equation}
  \label{eq:ccc}
   \rho=\int_{x\in \Omega} xx^{\dagger} d\mu(x).
 \end{equation}
In fact, consider the eigenvalue-eigenvector decomposition of the density matrix:
$$
\rho=\sum\limits_{i=1}^n \lambda_i v_iv_i^{\dagger},
$$
with $\lambda_i\geq0$ and $v_i \in \complex^{n}$  being orthonormal.
We can define the probability distribution
$$
\mu(x)= \sum\limits_{i=1}^n \lambda_i \delta_{v_i}(x),
$$
where $\delta_{v_i}$ is an atomic charge (Dirac's delta) on $v_i$.
Then it is immediate to verify that 
$$
\int_{x\in \Omega} x x^{\dagger}d\mu(x)=\sum\limits_{i=1}^n \lambda_i v_iv_i^{\dagger}=\rho.
$$
Note also that,  a truncated moment matrix does not uniquely define a probability distribution, i.e., for a given $\rho$ there may exist two probability distributions $\mu_1(x)\neq\mu_2(x)$ such that
$$
 \rho=\int_{x\in \Omega} xx^{\dagger} d\mu_1(x)=\int_{x\in \Omega} xx^{\dagger} d\mu_2(x).
$$
This means that, if we interpret  $\rho$ as a truncated moment matrix  \rednew{and thus defining via \eqref{eq:ccc} a closed convex set of probabilities (more precisely charges)}, QT is a theory of imprecise probability \cite{walley1991}. We will discuss more on this topic in Appendix \ref{sec:hidden1d}.
In fact, Karr \cite{karr1983extreme}
has  proved that the set of probabilities, which are feasible for the truncated moment constraint, e.g., $\rho=L(x x^{\dagger})$, is convex and
compact with respect to the weak$^*$-topology. Moreover, the extreme points of this set are probabilities that have at finite number of distinct points of support (e.g., they are finite mixtures of Dirac's deltas).
A similar characterisation for  POVM measurements is discussed in the QT context in \cite{chiribella2007continuous}.

However, we have clarified that for $m>1$ particles, $\rho$ can be interpreted as a truncated moment matrix only when $\bdomain=\{(\otimes_{j=1}^m x_j)^{\dagger} G (\otimes_{j=1}^m x_j): Tr(G\rho)\geq0\}$, the P-coherent
set of desirable gambles associated to $\rho$, does not satisfy the first condition of Theorem \ref{th:fundamental}. More discussions on this point are presented in the next sections. 

%


\section{Entanglement}\label{sec:entang}
Entanglement is usually presented as a  characteristic of QT. In this section we are going to show that it is actually an immediate consequence of computational tractability, meaning that entanglement phenomena are not confined to QT but can be observed in other contexts too. An example of a non-QT entanglement is provided  in Section~\ref{sec:ent_not_only}.

To illustrate the emergence of entanglement from P-coherence, we verify  that  the set of desirable gambles whose dual is an entangled density matrix $\rho_{e}$ includes a negative gamble that is not in $\bnegative$, and thus, although
being logically coherent, it cannot be given a classical probabilistic interpretation.

In what follows we focus only on bipartite systems $\pspace_A \times \pspace_B$, with $n=m=2$. The results are nevertheless general.

Let $(x,y) \in \pspace_A \times \pspace_B$, where
$x=[x_1,x_2]^T$ and $y=[y_1,y_2]^T$. 
We aim at showing that there exists a gamble $h(x,y)=(x \otimes y)^{\dagger} H (x \otimes y) $ satisfying:
\begin{equation}
\label{eq:condepr0}
\begin{aligned}
 Tr(H \rho_{e})&\geq0 \text{ and }\\
h(x,y)=(x \otimes y)^{\dagger} H (x \otimes y) &< 0 \text{ for all } (x,y)\in \pspace_A\times \pspace_B.\\
\end{aligned}
\end{equation}
The first inequality says that $h$ is desirable in $\btheory$. That is, $h$ is a gamble desirable to Alice whose beliefs are represented by $\rho_{e}$. The second inequality says that $h$ is negative and, therefore, leads to a sure loss in $\theory$.
By~\ref{eq:btaut}--\ref{eq:b2}, the inequalities in~\eqref{eq:condepr0} imply that $H$ must be an indefinite Hermitian matrix.

 Assume that $n=m=2$ and consider the entangled density matrix: 
$$
\rho_{e}=\frac{1}{2}\begin{bmatrix}
      1 & 0 & 0 &1\\
      0 & 0 & 0 &0\\
      0 & 0 & 0 &0\\
      1 & 0 & 0 &1\\
     \end{bmatrix},
$$
and the Hermitian matrix:
$$
H=\left[\begin{matrix}0.0 & 0.0 & 0.0 & 1.0\\0.0 & -2.0 & 1.0 & 0.0\\0.0 & 1.0 & -2.0 & 0.0\\1.0 & 0.0 & 0.0 & 0.0\end{matrix}\right].
$$
This matrix is indefinite (its eigenvalues are $\{1, -1, -1, -3\}$) and is such that $Tr(H\rho_{e})=1$.
Since $Tr(H\rho_{e})\geq0$, the gamble
\begin{equation}
\label{eq:sosqt}
\begin{aligned}
(x \otimes y)^{\dagger} H (x \otimes y)=&
- 2 x_{1} x_1^{\dagger} y_{2} y_2^{\dagger} +  x_{1} x_2^{\dagger} y_{1} y_2^{\dagger} +  x_{1} x_2^{\dagger} y_1^{\dagger} y_{2} +  x_1^{\dagger} x_{2} y_{1} y_2^{\dagger} +  x_1^{\dagger} x_{2} y_1^{\dagger} y_{2} - 2 x_{2} x_2^{\dagger} y_{1} y_1^{\dagger},
 \end{aligned}
\end{equation}
is desirable for Alice in $\btheory$. 

Let  $x_i=x_{ia}+\iota x_{ib}$ and $y_i=y_{ia}+\iota y_{ib}$ with $x_{ia},x_{ib},y_{ia},y_{ib}\in \reals$, for $i=1,2$, denote the real and imaginary
components of $x,y$. Then 
\begin{equation}
\label{eq:realim}
\begin{aligned}
(x \otimes y)^{\dagger} H (x \otimes y)=&- 2 x_{1a}^{2} y_{2a}^{2} - 2 x_{1a}^{2} y_{2b}^{2} + 4 x_{1a} x_{2a} y_{1a} y_{2a} + 4 x_{1a} x_{2a} y_{1b} y_{2b} \\
&- 2 x_{1b}^{2} y_{2a}^{2} - 2 x_{1b}^{2} y_{2b}^{2} + 4 x_{1b} x_{2b} y_{1a} y_{2a} + 4 x_{1b} x_{2b} y_{1b} y_{2b}\\
&- 2 x_{2a}^{2} y_{1a}^{2} - 2 x_{2a}^{2} y_{1b}^{2} - 2 x_{2b}^{2} y_{1a}^{2} - 2 x_{2b}^{2} y_{1b}^{2}\\
=&-(\sqrt{2}x_{1a}y_{2a}-\sqrt{2}x_{2a}y_{1a})^2-(\sqrt{2}x_{1a}y_{2b}-\sqrt{2}x_{2a}y_{1b})^2\\
&-(\sqrt{2}x_{1b}y_{2b}-\sqrt{2}x_{2b}y_{1b})^2
-(\sqrt{2}x_{2b}y_{1a}-\sqrt{2}x_{2a}y_{1b})^2<0.
 \end{aligned}
\end{equation}

This is the essence of the quantum puzzle:  $\bdomain$ is P-coherent but (Theorem~\ref{th:fundamental}) there is no $\mathcal{P}$ associated to it and therefore, from the point of view of a classical probabilistic interpretation, it is not coherent (in any classical description of the composite quantum system, the variables $x,y$ appear to be entangled in a way unusual for classical subsystems). 

As previously mentioned, there are two possible ways out from this impasse:
  to claim the existence of either non-classical evaluation functionals or of negative probabilities. Let us examine them in turn.

\begin{description}
\item[(1) Existence of non-classical evaluation functionals:]
From an informal betting perspective, the effect  of a quantum experiment on $h(x,y)$ is to evaluate this polynomial to return the payoff for Alice. By Theorem~\ref{th:fundamental}, there is no compatible classical evaluation functional, and thus in particular no 
value of the variables $x,y\in \pspace_A \times  \pspace_B$, such that 
$ h(x,y)=1$. Hence, if we adopt this point of view, we have to find another, non-classical, explanation for $ h(x,y)=1$. The following evaluation functional, denoted as $ev(\cdot)$, may do the job:
$$
\text{ev}\hspace{-1.1mm}\left(\begin{bmatrix}
x_1y_1\\
x_2y_1\\
x_1y_2\\
x_2y_2\\
\end{bmatrix}\right)=\begin{bmatrix}
\tfrac{\sqrt{2}}{2}\\
0\\
0\\
\tfrac{\sqrt{2}}{2}\\
\end{bmatrix},~\text{which implies}~ \text{ev}\hspace{-1mm}\left((x \otimes y)^{\dagger} H (x \otimes y)\right)=1.
$$

Note that, $x_1y_1=\tfrac{\sqrt{2}}{2}$ and $x_2y_1=0$ together imply that
$x_2=0$, which contradicts $x_2y_2=\tfrac{\sqrt{2}}{2}$. Similarly,
 $x_2y_2=\tfrac{\sqrt{2}}{2}$ and $x_1y_2=0$ together imply that
$x_1=0$, which contradicts $x_1y_1=\tfrac{\sqrt{2}}{2}$. Hence, as expected, the above evaluation functional is non-classical. It amounts to  assigning a value to the products $x_iy_j$ but not to the single components of $x$ and $y$ separately. Quoting \cite[Supplement~3.4]{holevo2011probabilistic}, ``entangled states are holistic entities in which the single components only exist virtually''. 



\item[(2) Existence of negative probabilities:] Negative probabilities  are not an intrinsic characteristic of QT. They
appear whenever one attempts to explain QT ``classically'' by looking at the space of charges on $\pspace$.
To see this, consider $\rho_e$, and assume that, based on \eqref{eq:momepr}, one calculates: 
  \begin{equation}
  \label{eq:momepr}
\int  \begin{bmatrix}x_{1} x_1^{\dagger} y_{1} y_1^{\dagger} & x_1^{\dagger} x_{2} y_{1} y_1^{\dagger} & x_{1} x_1^{\dagger} y_1^{\dagger} y_{2} & x_1^{\dagger} x_{2} y_1^{\dagger} y_{2}\\x_{1} x_2^{\dagger} y_{1} y_1^{\dagger} & x_{2} x_2^{\dagger} y_{1} y_1^{\dagger} & x_{1} x_2^{\dagger} y_1^{\dagger} y_{2} & x_{2} x_2^{\dagger} y_1^{\dagger} y_{2}\\x_{1} x_1^{\dagger} y_{1} y_2^{\dagger} & x_1^{\dagger} x_{2} y_{1} y_2^{\dagger} & x_{1} x_1^{\dagger} y_{2} y_2^{\dagger} & x_1^{\dagger} x_{2} y_{2} y_2^{\dagger}\\x_{1} x_2^{\dagger} y_{1} y_2^{\dagger} & x_{2} x_2^{\dagger} y_{1} y_2^{\dagger} & x_{1} x_2^{\dagger} y_{2} y_2^{\dagger} & x_{2} x_2^{\dagger} y_{2} y_2^{\dagger}\end{bmatrix} d \mu(x,y)=\frac{1}{2}\begin{bmatrix}
      1 & 0 & 0 &1\\
      0 & 0 & 0 &0\\
      0 & 0 & 0 &0\\
      1 & 0 & 0 &1\\
     \end{bmatrix}.
\end{equation}
Because of Theorem~\ref{th:fundamental},  there is no probability charge $\mu$ satisfying these moment constraints, the only compatible being signed ones.
Box~1 reports the nine components and corresponding weights of one of them:
\begin{equation}
\label{eq:momsigned}
 \mu(x,y)=\sum\limits_{i=1}^{9}w_i\delta_{\{(x^{(i)},y^{(i)})\}}(x,y) ~~\text{ with }~~ (w_i,x^{(i)},y^{(i)}) ~~\text{ as in Box~1. } 
\end{equation}
Note that some of the weights are negative but $\sum_{i=1}^{9}w_i=1$, meaning that we have an affine combination of atomic charges (Dirac's deltas).

{

\begin{framed}

\vspace{-0.5ex}

\noindent {{\bf Box 1: charge compatible with \eqref{eq:momepr}}}

\vspace{-0.9ex}

{\footnotesize

\noindent 
\begin{center}
 \tiny{
\begin{tabular}{|c|c|c|c|c|c|}
\multicolumn{1}{c}{} & \multicolumn{1}{c}{1} & \multicolumn{1}{c}{2} & \multicolumn{1}{c}{3} & \multicolumn{1}{c}{4} & \multicolumn{1}{c}{5}\\
\hline
 \multirow{ 2}{*}{x} & -0.0963 - 0.6352$\iota$ & 0.251 - 0.9665$\iota$ & 0.7884 + 0.2274$\iota$ & 0.5702 - 0.4006$\iota$ & 0.3452 - 0.4539$\iota$\\
 & -0.0065 - 0.7663$\iota$ & 0.0387 + 0.0381$\iota$ & -0.1263 +  0.5574$\iota$ & 0.0027 + 0.7172$\iota$ & 0.4872 + 0.6613$\iota$\\\hline
 \multirow{ 2}{*}{y} &  -0.3727 - 0.3899$\iota$ & 0.6359 - 0.5716$\iota$ & 0.1553 - 0.4591$\iota$ & -0.3515 + 0.2848$\iota$ & 0.2129 - 0.2004$\iota$\\
 & -0.4385 - 0.7189$\iota$ & 0.3725 + 0.3608$\iota$ & 0.4039 +  0.7759$\iota$ & 0.5911 - 0.6678$\iota$ & 0.2032 - 0.9345$\iota$\\ \hline
 w &     0.4805&  0.7459 & -0.892 &  0.7421 &  0.4724 \\ \hline  
 \end{tabular}\\
 \vspace{0.5cm}
  \begin{tabular}{|c|c|c|c|c|}
\multicolumn{1}{c}{} & \multicolumn{1}{c}{6} & \multicolumn{1}{c}{7} & \multicolumn{1}{c}{8} & \multicolumn{1}{c}{9} \\
  \hline
  \multirow{ 2}{*}{x} & 0.818 + 0.2654$\iota$ & -0.0541 - 0.8574$\iota$ & -0.3179 - 0.1021$\iota$ & -0.1255 - 0.3078$\iota$  \\
 & -0.486 + 0.1556$\iota$ & 0.4995 + 0.1112$\iota$ & 0.5198 + 0.7864$\iota$ & 0.2943 - 0.8961$\iota$  \\
 \hline
 \multirow{ 2}{*}{y} &  0.446 + 0.6996$\iota$ & -0.1628 + 0.561$\iota$ & 0.6285 - 0.4852$\iota$ & 0.0933 - 0.4588$\iota$  \\
 & -0.5474 - 0.1096$\iota$ & -0.8105 + 0.0419$\iota$ & -0.0035 - 0.6079$\iota$ & 0.8455 + 0.2568$\iota$  \\
 \hline
 w &      0.3297 &  -0.7999 &        -0.2544 &  0.1755\\
 \hline
\end{tabular}}\end{center}
The table reports the components of a  charge $\sum_{i=1}^{9}w_i\delta_{\{(x^{(i)},y^{(i)})\}}(x,y)$ that satisfies \eqref{eq:momepr}. The $i$-th column
of the row denoted as $x$ (resp. $y$) corresponds to the element $x^{(i)}$ (resp. $y^{(i)}$). The $i$-th column of the vector
$w$ corresponds to $w_i$. 
Consider for instance the first monomial $x_{1} x_1^{\dagger} y_{1} y_1^{\dagger}$ in \eqref{eq:momepr}, its expectation w.r.t.\ the above charge is
{\scriptsize
$$
\begin{aligned}
&\int x_{1} x_1^{\dagger} y_{1} y_1^{\dagger} \left(\sum_{i=1}^{9}w_i\delta_{\{(x^{(i)},y^{(i)})\}}(x,y)\right) dxdy=\sum_{i=1}^{9}w_i x^{(i)}_{1} {x^{(i)}_1}^{\dagger} y^{(i)}_{1} {y^{(i)}_1}^{\dagger}\\
&= 0.4805 (-0.0963 - 0.6352\iota)(-0.0963 + 0.6352\iota)(-0.3727 - 0.3899\iota)(-0.3727 + 0.3899\iota)\\
&+ 0.7459 (0.251 - 0.9665\iota)(0.251 + 0.9665\iota)(-0.1628 + 0.561\iota)(-0.1628 - 0.561\iota)\\
&+\dots\\
&+ 0.1755(-0.1255 - 0.3078\iota)(-0.1255 + 0.3078\iota)(0.0933 - 0.4588\iota)(0.0933 + 0.4588\iota)=\frac{1}{2}.
\end{aligned}
$$}
 }

\vspace{0ex}

\end{framed}
}

The  charge described in Box~1  is one among the many that satisfy~\eqref{eq:momepr} and has been derived numerically.
Explicit procedure for constructing such negative-probability representations have been developed in \citep{schack2000explicit,sperling2009necessary,
gerke2016multipartite,gerke2018numerical}.
%
%
\end{description}

Again, we want to stress that the two above paradoxical interpretations are a consequence of Theorem~\ref{th:fundamental}, and therefore can emerge when considering any instance of a theory of P-coherence in which
the hypotheses of  this result hold.

\subsection{Local realism}\label{sec:local}
The issue of \emph{local realism} in QT arises when one performs measurements on a pair of separated but entangled particles.
This again shows the impossibility of a peaceful agreement between the internal logical consistency of a P-coherent theory and the attempt to provide an external coherent (classical) interpretation.
Let us discuss it from the latter perspective. Firstly, notice that, since $(x_{1} x_1^{\dagger}+x_{2} x_2^{\dagger})=1$, the linear operator
\begin{align}
\label{eq:matL}
L\left(\left[\begin{matrix}x_{1} x_1^{\dagger} y_{1} y_1^{\dagger} & x_1^{\dagger} x_{2} y_{1} y_1^{\dagger} & x_{1} x_1^{\dagger} y_1^{\dagger} y_{2} & x_1^{\dagger} x_{2} y_1^{\dagger} y_{2}\\x_{1} x_2^{\dagger} y_{1} y_1^{\dagger} & x_{2} x_2^{\dagger} y_{1} y_1^{\dagger} & x_{1} x_2^{\dagger} y_1^{\dagger} y_{2} & x_{2} x_2^{\dagger} y_1^{\dagger} y_{2}\\x_{1} x_1^{\dagger} y_{1} y_2^{\dagger} & x_1^{\dagger} x_{2} y_{1} y_2^{\dagger} & x_{1} x_1^{\dagger} y_{2} y_2^{\dagger} & x_1^{\dagger} x_{2} y_{2} y_2^{\dagger}\\x_{1} x_2^{\dagger} y_{1} y_2^{\dagger} & x_{2} x_2^{\dagger} y_{1} y_2^{\dagger} & x_{1} x_2^{\dagger} y_{2} y_2^{\dagger} & x_{2} x_2^{\dagger} y_{2} y_2^{\dagger}\end{matrix}\right]\right)
\end{align}
satisfies the properties:
\begin{align}
\nonumber
L(x_{1} x_1^{\dagger} y_{1} y_1^{\dagger})+L(x_{2} x_2^{\dagger} y_{1} y_1^{\dagger})&=L((x_{1} x_1^{\dagger}+x_{2} x_2^{\dagger})y_{1} y_1^{\dagger})=L(y_{1} y_1^{\dagger}),\\
\nonumber
L(x_{1} x_1^{\dagger} y_{2} y_2^{\dagger})+L(x_{2} x_2^{\dagger} y_{2} y_2^{\dagger})&=L((x_{1} x_1^{\dagger}+x_{2} x_2^{\dagger})y_{2} y_2^{\dagger})=L(y_{2} y_2^{\dagger}),\\
\label{eq:PTrace}
L(x_{1} x_1^{\dagger} y_{1} y_2^{\dagger})+L(x_{2} x_2^{\dagger} y_{1} y_2^{\dagger})&=L((x_{1} x_1^{\dagger}+x_{2} x_2^{\dagger})y_{1} y_2^{\dagger})=L(y_{1} y_2^{\dagger}).
\end{align}
Hence, by summing up some of the components of the matrix \eqref{eq:matL}, we can recover  the marginal linear operator
$$
M_y=L\left(\left[\begin{matrix}y_{1} y_1^{\dagger} & y_{1} y_2^{\dagger} \\
 y_{2} y_1^{\dagger} & y_{2} y_2^{\dagger}\end{matrix}\right]\right)=\left[\begin{matrix}\frac{1}{2} & 0\\ 
0 & \frac{1}{2}\end{matrix}\right],
$$
where the last equality holds when $L((x\otimes y)(x\otimes y)^{\dagger})=\rho_e$, i.e., $M_y=\frac{1}{2}I_2$ is the reduced density matrix of $\rho_e$ on system $B$. The operation we have just described, when applied to a density matrix, is known in QT as \emph{partial trace}. Given the interpretation of $\rho$ as a  dual operator, the 
operation of partial trace simply  follows by Equation \eqref{eq:PTrace}.

Similarly,  we can obtain  
$$
M_x=L\left(\left[\begin{matrix}x_{1} x_1^{\dagger} & x_{1} x_2^{\dagger} \\
 x_{2} x_1^{\dagger} & x_{2} x_2^{\dagger}\end{matrix}\right]\right)=\left[\begin{matrix}\frac{1}{2} & 0\\ 
0 & \frac{1}{2}\end{matrix}\right].
$$
Matrix $M_x$ (analogously to $M_y$) is compatible with probability: there are marginal probabilities whose $M_x$ is the moment matrix, an example being
$$
\frac{1}{2}\delta_{\left\{\begin{bmatrix}
            1\\
            0
           \end{bmatrix}\right\}}(x)+\frac{1}{2}\delta_{\left\{\begin{bmatrix}
            0\\
            1
           \end{bmatrix}\right\}}(x).
$$ 
In other words, we are brought to believe that \emph{marginally} the physical properties $x,y$ of the two particles
have a meaning, i.e., they can be explained through probabilistic mixtures of classical evaluation functionals. We can now ask Nature, by means of a real experiment, to decide between our common-sense notions of how the world works, and Alice's one.
Experimental verification of this phenomenon can be obtained by a CHSH-like experiment, which aims at experimentally reproducing
a situation where~\eqref{eq:condepr0} holds, as explained in Box~2. In this interpretation, the CHSH experiment is an \textit{entanglement witness}, we 
discuss the connection between \eqref{eq:condepr0} and the entanglement witness theorem  in Section~\ref{sec:witness}.

{

\begin{framed}

\vspace{-0.5ex}

\noindent {{\bf Box 2: CHSH experiment}}

\vspace{-0.9ex}

{\footnotesize

\noindent 
\begin{center}
\includegraphics[width=12cm]{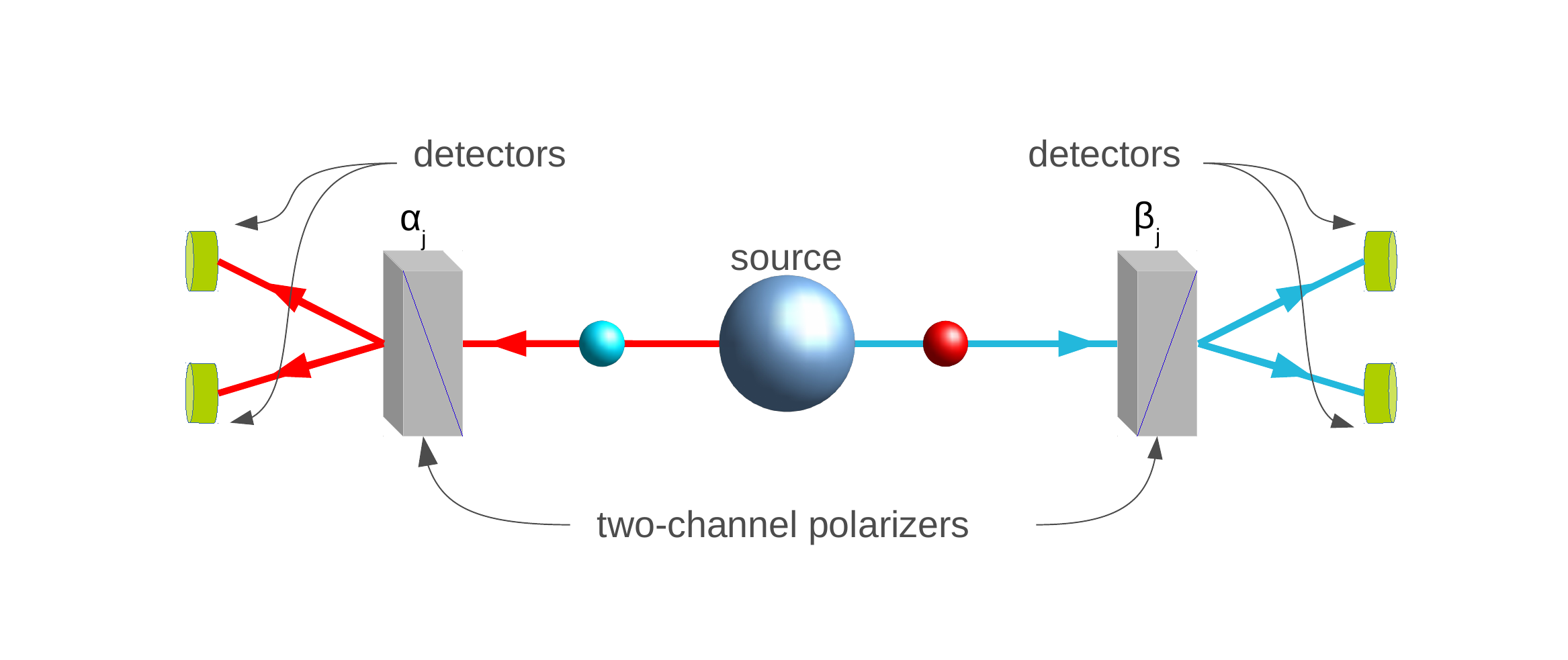}
\end{center}
The source produces pairs of entangled photons, sent in opposite directions. 
Each photon encounters a two-channel polariser whose orientations  can be set by the experimenter. 
Emerging signals from each channel are captured by detectors. 
Four possible orientations $\alpha_i,\beta_j$ for $i,j=1,2$  of the  polarisers are tested.  Consider the Hermitian matrices
$G_{\alpha_i}=\sin(\alpha_i)\sigma_{x}+\cos(\alpha_i)\sigma_{z}$, $G_{\beta_j}=\sin(\beta_j)\sigma_{x}+\cos(\beta_j)\sigma_{z}$,
where $\sigma_{x},\sigma_{z}$ are the 2D Pauli's matrices, and define the gamble 
$$
(x^\dagger G_{\alpha_i} x)(y^\dagger G_{\beta_j} y )=(x \otimes y)^\dagger G_{\alpha_i\beta_j} (x \otimes y)
$$
on the result of experiment with $G_{\alpha_i\beta_j}=G_{\alpha_i} \otimes G_{\beta_j}$. Consider then the sum gamble
$$
h(x,y)=(x \otimes y)^\dagger (G_{\alpha_1\beta_2}-G_{\alpha_1\beta_2}+G_{\alpha_2\beta_1}+G_{\alpha_2\beta_2})(x \otimes y)
$$
and observe that
$$
\begin{aligned}
h(x,y)&=(x \otimes y)^\dagger (G_{\alpha_1}\otimes (G_{\beta_2}-G_{\beta_2}))(x \otimes y) +(x \otimes y)^\dagger (G_{\alpha_2}\otimes (G_{\beta_1}+G_{\beta_2}))(x \otimes y)\\
&=(x^\dagger G_{\alpha_2} x)(y^\dagger (G_{\beta_1}+G_{\beta_2}) y)\\
&\leq y^\dagger (G_{\beta_1}+G_{\beta_2}) y\\
&\leq 2,\\
\end{aligned}
$$
this is the CHSH inequality.  For a small $\epsilon>0$, we have that 
$$
h(x,y)-2-\epsilon=(x \otimes y)^\dagger (-(2+\epsilon) I_4+G_{\alpha_1\beta_2}-G_{\alpha_1\beta_2}+G_{\alpha_2\beta_1}+G_{\alpha_2\beta_2})(x \otimes y)<0
$$ 
\textbf{but} 
$$
~~~~~~~~~~~~~~Tr((-(2+\epsilon)I_4+G_{\alpha_1\beta_2}-G_{\alpha_1\beta_2}+G_{\alpha_2\beta_1}+G_{\alpha_2\beta_2}) \rho_e)=-2+\epsilon+2\sqrt{2}\geq0
$$
for $\alpha_1=\pi/2,\alpha_2=0,\beta_1=\pi/4,\beta_2=-\pi/4$. We are again in a situation like \eqref{eq:condepr0}.
The experiment in the figure certifies the entanglement by measuring the QT expectation of the four components of $h(x,y)$.
 }

\vspace{0ex}

\end{framed}
}


The situation we have just described is the playground of Bell's theorem, stating the impossibility of Einstein's postulate of local realism.

The argument goes as follows. 
If we assume that the physical properties $x,y$ of the two particles (the polarisations of the photons) have definite values
$\tilde{x},\tilde{y}$ that exist independently of observation (\emph{realism}), then the measurement on the first qubit
must influence the result of the measurement on the second qubit. 
 Vice versa if we assume \emph{locality}, then $x,y$ cannot exist independently of the observations.
To sum up, a local hidden variable theory that is compatible with QT results cannot exist \citep{bell1964einstein}.

But is there really anything contradictory here? The message we want to convey is that this is not the case.
Indeed, since Theorem~\ref{th:fundamental} applies, there is no probability compatible with the moment matrix $\rho_e$.
Ergo, although they may seem to be compatible
with probabilities, the marginal  matrices $M_x,M_y$ are not  moment matrices of any probability. 

The conceptual mistake in the situation we are considering is to forget that $M_x,M_y$ come from $\rho_e$. A joint linear operator uniquely defines its marginals but not the other way round. 
There are \emph{infinitely many} joint probability charges whose $M_x,M_y$  are the marginals, e.g.,
\[
d\mu(x,y)=\left(\frac{1}{2}\delta_{\left\{\begin{bmatrix}
            1\\
            0
           \end{bmatrix}\right\}}(x)+\frac{1}{2}\delta_{\left\{\begin{bmatrix}
            0\\
            1
           \end{bmatrix}\right\}}(x)\right)\left(\frac{1}{2}\delta_{\left\{\begin{bmatrix}
            1\\
            0
           \end{bmatrix}\right\}}(y)+\frac{1}{2}\delta_{\left\{\begin{bmatrix}
            0\\
            1
           \end{bmatrix}\right\}}(y)\right),
\]
but none of them satisfy  Equation~\eqref{eq:momepr}. In fact, such an intrinsic non-uniqueness of the compatible joints is another amazing characteristic of QT: from this perspective, it 
 is not only a  theory of probability, but a theory  of \emph{imprecise} probability, see Section~\ref{sec:momentmat}. 
 
Instead, the reader can verify that the  charge in Equation~\eqref{eq:momsigned} satisfies
both Equation~\eqref{eq:momepr} and:

\[
\displaystyle{
\int \begin{bmatrix}x_{1} x_1^{\dagger} & x_{1} x_2^{\dagger} \\
 x_{2} x_1^{\dagger} & x_{2} x_2^{\dagger}\end{bmatrix} d\mu(x,y)=\displaystyle{\int \begin{bmatrix}y_{1} y_1^{\dagger} & y_{1} y_2^{\dagger} \\
 y_{2} y_1^{\dagger} & y_{2} y_2^{\dagger}\end{bmatrix} d\mu(x,y)=\left[\begin{matrix}\frac{1}{2} & 0\\ 
0 & \frac{1}{2}\end{matrix}\right].
}
}
\]

The take-away message of this subsection is that we should only interpret $M_x,M_y$  as marginal operators  and keep in mind that QT is a logical theory of P-coherence.
We see paradoxes when we try to force a physical interpretation upon QT, whose nature is instead computational.
In other words, if we accept that computation is more primitive than our classical interpretation of physics, all paradoxes disappear.


\subsection{Entanglement witness theorem}
\label{sec:witness}

In the previous Subsections, we have seen that all paradoxes of QT emerge because of disagreement between its internal coherence and the attempt to force on it a classical coherent interpretation.

Do quantum and classical probability sometimes agree?
Yes they do, but when at play there are density matrices $\rho$ such that Equation \eqref{eq:condepr0} does not hold, and thus in particular for \emph{separable density matrices}.
We make this claim precise by providing a link between Equation \eqref{eq:condepr0} and the \emph{entanglement witness theorem}
\citep{horodecki1999reduction,horodecki2009quantum}.

We first report the definition of entanglement witness \citep[Sec.~6.3.1]{heinosaari2011mathematical}:
\begin{definition}[Entanglement witness]
\label{def:entangle}
 A Hermitian operator $W \in \He^{n_1 \times n_2} $ is an \emph{entanglement
witness} if and only if $W$ is not a positive operator but $(x_1 \otimes x_2)^{\dagger} W (x_1 \otimes x_2) \geq 0$ for all 
vectors $(x_1,x_2)\in \pspace_1\times \pspace_2$.\footnote{In  \citep[Sec.~6.3.1]{heinosaari2011mathematical}, the last part of this definition says
``for all factorized vectors $x_1 \otimes x_2$''. This is equivalent to considering the pair $(x_1,x_2)$.} 
\end{definition}

The next well-known result (see, e.g.,  \citep[Theorem~6.39, Corollary~6.40]{heinosaari2011mathematical}) provides a characterisation of entanglement 
and separable states in terms of entanglement
witness.
\begin{proposition}
\label{prop:witn}
 A state $\rho_e$ is entangled if and only if there exists an entanglement
witness $W$ such that $Tr(\rho_e W ) < 0$. A state is separable if and only if $Tr(\rho_e W ) \geq 0$ for all entanglement
witnesses $W$.
\end{proposition}


Assume that  $W$ is an entanglement witness for the entangled density matrix  $\rho_e$ and consider $W'=-W$. By Definition \ref{def:entangle} and Proposition \ref{prop:witn}, it follows
that 
\begin{equation}
\label{eq:witn}
Tr(\rho_e W' ) > 0 \text{ and } (x_1 \otimes x_2)^{\dagger} W' (x_1 \otimes x_2) \leq 0. 
\end{equation}

The first inequality states that the  gamble $(x_1 \otimes x_2)^{\dagger} W' (x_1 \otimes x_2)$ is \emph{strictly} desirable for Alice (in theory $\btheory$) given her belief $\rho_e$.
Since the set of desirable gambles (\ref{eq:b1}) associated to $\rho_e$ is closed, there exists $\epsilon>0$ such that $W''=W'-\epsilon I$ is still desirable, i.e,
$Tr(\rho_e W'' )\geq0$ and 
$$
(x_1 \otimes x_2)^{\dagger} W'' (x_1 \otimes x_2) = (x_1 \otimes x_2)^{\dagger} W' (x_1 \otimes x_2) - \epsilon<0,
$$
where we have exploited that $(x_1 \otimes x_2)^{\dagger} \epsilon I (x_1 \otimes x_2)= \epsilon$.
Therefore, \eqref{eq:witn} is equivalent to
\begin{equation}
\label{eq:witn1}
Tr(\rho_e W'' ) \geq 0 \text{ and } (x_1 \otimes x_2)^{\dagger} W'' (x_1 \otimes x_2) <0,
\end{equation}
which is the same as \eqref{eq:condepr0}.

Hence, by Theorem~\ref{th:fundamental}, we can equivalently formulate the entanglement witness theorem as an arbitrage/Dutch book:
\begin{theorem}
\label{th:ewt}
 Let $\bdomain=\{g(x_1,\dots,x_m)=(\otimes_{j=1}^m x_j)^\dagger G (\otimes_{j=1}^m x_j)\mid Tr(G\tilde{\rho})\geq 0\}$ be the set of desirable gambles corresponding to some density matrix $\tilde{\rho}$. The following claims are equivalent:
 \begin{enumerate}
\item $\tilde{\rho}$ is entangled;
\item $\posi(\bdomain \cup \gambles^{\geq})$ is not coherent in $\theory$.
\end{enumerate}
\end{theorem}

This result provides another view of the entanglement witness theorem in light of P-coherence. In particular, it tells us that the existence of a witness satisfying Equation~\eqref{eq:witn} boils down to the disagreement between the classical probabilistic interpretation and the theory $\btheory$ on the rationality (coherence) of Alice, and therefore that whenever they agree on her rationality it means that $\rho_e$ is separable. 

This connection explains why the problem of characterising entanglement  is hard in QT: it amounts to proving the negativity of a function, which is NP-hard.

\section{Entangled states do not only exist in QT}\label{sec:ent_not_only}
In this Section we are going to present an example of entanglement in a P-coherence theory of probability that is different from QT.

Consider two  real variables $x=[x_1,x_2]^T$ with $x_i \in \reals$,  the possibility space $\Omega=\{x\in\reals^2\}$ and the vector space
of gambles
$$
\gambles_R=\{v^T(x)Gv(x): G \text{ is a symmetric real matrix}\},
$$
where $v(x)$ is the column vector of monomials:
$$
v(x)=[1,x_1,x_2,x_1^2,x_1 x_2,x_2^2, x_1^3,x_1^2x_2,x_1x_2^2,x_2^3]^T,
$$
whose dimension is $d=10$. $\gambles_R$ is therefore the space of all polynomials of degree $6$ of the real variables $x_1,x_2$.
It can be observed that in $\gambles_R$ the constant functions can be represented as
$$
v^T(x)Gv(x)=v^T(x)(ce_1e_1^T)v(x)=c,
$$
for any constant $c\in \reals$, $e_1$ being the first element of the canonical basis of $\reals^{d}$, i.e., $e_1=[1,0,\dots,0]$.
Compare this with the  non-negative constant functions in QT that have the form $g(x_1,\dots,x_m)=c (\otimes_{j=1}^m x_j)^\dagger I (\otimes_{j=1}^m x_j)$,
see Section \ref{sec:qt}.

We say that $\bdomain$  is a P-coherent set of desirable gambles whenever it satisfies \eqref{eq:btaut}--\eqref{eq:b2} in Section \ref{sec:comp} with 
\begin{align}
 \Sigma^{\geq}&=\{g(x_1,x_2)=v^T(x)Gv(x): G\geq0\},\\
 \Sigma^{<}&=\{g(x_1,x_2)=v^T(x)Gv(x): G<0\}.
\end{align}
The polynomial $v^T(x)Gv(x)$ with $G\geq0$ are called \textit{sum-of-square polynomials}, see also Section \ref{sec:sos}.

In this case too we can define the dual operator
\begin{equation}
\label{eq:linearoperatorx1x2}
Z:=L\left(v(x)v^T(x)\right).
\end{equation}
If we define $z_{\alpha\beta}=L(x_1^\alpha x_2^\beta)\in \reals$, then $Z$  has the following structure
\begin{align}
\label{eq:dualMatZ}
Z=\begin{bsmallmatrix}
 z_{00} &z_{10} &z_{01} &z_{20} &z_{11} &z_{02} &z_{30} &z_{21} &z_{12} &z_{03}\\
z_{10} &z_{20} &z_{11} &z_{30} &z_{21} &z_{12} &z_{40} &z_{31} &z_{22} &z_{13}\\
z_{01} &z_{11} &z_{02} &z_{21} &z_{12} &z_{03} &z_{31} &z_{22} &z_{13} &z_{04}\\
z_{20} &z_{30} &z_{21} &z_{40} &z_{31} &z_{22} &z_{50} &z_{41} &z_{32} &z_{23}\\
z_{11} &z_{21} &z_{12} &z_{31} &z_{22} &z_{13} &z_{41} &z_{32} &z_{23} &z_{14}\\
z_{02} &z_{12} &z_{03} &z_{22} &z_{13} &z_{04} &z_{32} &z_{23} &z_{14} &z_{05}\\
z_{30} &z_{40} &z_{31} &z_{50} &z_{41} &z_{32} &z_{60} &z_{51} &z_{42} &z_{33}\\
z_{21} &z_{31} &z_{22} &z_{41} &z_{32} &z_{23} &z_{51} &z_{42} &z_{33} &z_{24}\\
z_{12} &z_{22} &z_{13} &z_{32} &z_{23} &z_{14} &z_{42} &z_{33} &z_{24} &z_{15}\\
z_{03} &z_{13} &z_{04} &z_{23} &z_{14} &z_{05} &z_{33} &z_{24} &z_{15} &z_{06}\\
\end{bsmallmatrix}
\end{align}
where, for instance, $z_{00}=L(1)$, $z_{10}=L(x_1)$, $z_{20}=L(x_1^2)$, $z_{01}=L(x_2)$, $z_{02}=L(x_2^2)$ etc..

Now, the dual of $\bdomain$ can be calculated as in Section \ref{sec:comp}:
\begin{equation}
\label{eq:dualM1sos}
\begin{aligned}
\mathscr{Q}&=\left\{{z} \in \reals^{{d}}:  L(g)=Tr({G} Z)\geq0, ~Z\geq 0,~ z_{00}=1, ~~\forall g \in \bdomain\right\}.
\end{aligned}
\end{equation} 
Note that, because the way constants are represented in this context, we have $z_{00}=1$ instead of $Tr(Z)=1$.

Since sum-of-squares (SOS) implies nonnegativity, a natural question is to known whether any nonnegative polynomial
of degree $6$ of the real variables $x_1,x_2$ can be expressed
as a sum of squares. 
It turns out that this is not the case. Indeed, David Hilbert showed that equality between the
set of nonnegative polynomials of $n$ variables of degree $2d$ and SOS polynomials of $n$ variables of degree $2d$ holds only in the
following three cases: univariate polynomials (i.e., $n = 1$); quadratic polynomials ($2d = 2$); bivariate quartics ($n = 2$, $2d = 4$).
For all other cases, there always exist nonnegative polynomials that are not sums
of squares. 
The most famous counter-example is a polynomial due to Motzkin:
$$
m(x)=x_1^4x_2^2+x_1^2x_2^4-x_1^2x_2^2+1.
$$

Motzkin polynomial is nonnegative but it is not a sum-of-squares.
In what follow, it will be used to prove that a theory $\theory^*$ in the space of polynomials
of two real variables of degree $6$  has ``entangled density matrices'' in its dual space.

Consider the following PSD matrix \cite{Benavoli2017b} of type \eqref{eq:dualMatZ}
\begin{align}
\label{eq:darioMat}
Z_e=\begin{footnotesize}\begin{bsmallmatrix} 1 & 0 & 0 & 353 & 0 & 353 & 0 & 0 & 0 & 0\\ 0 & 353 & 0 & 0 & 0 & 0 & 249572 & 0 & 66 & 0\\ 0 & 0 & 353 & 0 & 0 & 0 & 0 & 66 & 0 & 249572\\ 353 & 0 & 0 & 249572 & 0 & 66 & 0 & 0 & 0 & 0\\ 0 & 0 & 0 & 0 & 66 & 0 & 0 & 0 & 0 & 0\\ 353 & 0 & 0 & 66 & 0 & 249572 & 0 & 0 & 0 & 0\\ 0 & 249572 & 0 & 0 & 0 & 0 & 706955894 & 0 & 17 & 0\\ 0 & 0 & 66 & 0 & 0 & 0 & 0 & 17 & 0 & 17\\ 0 & 66 & 0 & 0 & 0 & 0 & 17 & 0 & 17 & 0\\ 0 & 0 & 249572 & 0 & 0 & 0 & 0 & 17 & 0 & 706955894 \end{bsmallmatrix}\end{footnotesize}
\end{align}
and the negative  polynomial $-m(x)$.
Since  $L(f)=-1-z_{42}-z_{24}+z_{22}$ and   $z_{22}=66, z_{24}=z_{42}=17$ in \eqref{eq:darioMat},  we have that $L(f)=31>0$.
Therefore,  the conditions in \eqref{eq:condepr0} are met and Theorem \ref{th:fundamental} holds.
In particular, this means that the polynomial $-m(x)$ is desirable by a subject, Alice, whose beliefs are expressed by $Z_e$ and who is therefore a rational agent in 
$\theory^*$. However,  the polynomial is negative: Alice is irrational in 
$\theory$.

Notice that, if we consider the definition of entanglement in Section \ref{sec:witness},  
$Z_e$ is an entanglement matrix, since
$$
-m(x)<0 ~~\text{ but }~~ L(-m(x))>0.
$$
Moreover, from \eqref{eq:darioMat} and the definition of $Z$ we can extract the marginal operator for the variables $x_1,x_2$.
\begin{align}
 M_{x_1} &=L\left(\begin{bsmallmatrix}
                1 & x_1 & x_1^2 & x_1^3\\
                x_1 & x_1^2 & x_1^3 & x_1^4\\
                x_1^2 & x_1^3 & x_1^4 & x_1^5\\
                x_1^3 & x_1^4 & x_1^5  & x_1^6\\
               \end{bsmallmatrix}
\right)=
\begin{bsmallmatrix}
 1 & 0 & 353 & 0\\
 0 & 353 & 0& 249572\\
 353 & 0 & 249572 & 0\\
 0 & 249572 & 0 & 706955894
\end{bsmallmatrix},\\
 M_{x_2}&=L\left(\begin{bsmallmatrix}
                1 & x_2 & x_2^2 & x_2^3\\
                x_2 & x_2^2 & x_2^3 & x_2^4\\
                x_2^2 & x_2^3 & x_2^4 & x_2^5\\
                x_2^3 & x_2^4 & x_2^5  & x_2^6\\
               \end{bsmallmatrix}
\right)=
\begin{bsmallmatrix}
 1 & 0 & 353 & 0\\
 0 & 353 & 0& 249572\\
 353 & 0 & 249572\\ & 0\\
 0 & 249572 & 0 & 706955894
\end{bsmallmatrix}.
\end{align}
This time the operation is not equal to partial trace because, again, the structure of the polynomials $v^T(x)Gv(x)$ 
is different from those in QT.  

Since $ M_{x_1}, M_{x_2}$ are both PSD they are valid moment matrices.
Moreover, for single real variables, Hilbert has shown that every nonnegative polynomial is SOS and, therefore,
marginally, the above truncated moment matrices are compatible with classical probability. A subject whose marginal beliefs are  expressed by $ M_{x_1}\geq0, M_{x_2}\geq0$ is always rational in $\theory$.
This is exactly the same situation encountered when showing the impossibility of a local realistic interpretation of quantum mechanics (see Section \ref{sec:local}): the marginals seem compatible with  classical probability but the joint is not, since we can derive a Bell-type inequalities that are violated by  $Z_e$.

Again the conceptual mistake of this reading  is  to forget that $M_x$ and $M_y$ come from $Z_e$ and thus they should only be interpreted  as marginal operators .

\section{Discussions}

\subsection{The class of P-nonnegative gambles}
\label{sec:HermSOS}
The class of P-nonnegative gambles, defined in Section \ref{sec:qt}, is the closed convex cone of all \textit{Hermitian sum-of-squares} in $\gambles_R=\{(\otimes_{j=1}^m x_j)^\dagger G (\otimes_{j=1}^m x_j) \mid G\in \He^{r \times r}\}$, that is of all gambles $ g(x_1,\dots,x_m) \in \gambles_R$ for which $G$ is PSD. In particular this means that 
 Alice can efficiently determine whether a gamble is P-nonnegative or not.  
But is this class the only closed convex cone of nonnegative polynomials in $\gambles_R$ for which the membership problem can be solved efficiently (in polynomial-time)?
It turns out that the answer is negative (see for instance \cite{d2009polynomial,josz2018lasserre}): in addition of \textit{Hermitian sum-of-squares} (the one that Nature has chosen for QT)
one could also consider \textit{real sum-of-squares} in $\gambles_R$, that is polynomials of the form $(\otimes_{j=1}^m x_j)^\dagger G (\otimes_{j=1}^m x_j)$ that are sum-of-squares of polynomials of the real and imaginary part of the variables $ x_j$.

A separating example is the polynomial in \eqref{eq:sosqt}, which is not a \textit{Hermitian sum-of-squares} but it is a \textit{real sum-of-square}, as it can be seen from  \eqref{eq:realim}. This polynomial was used in our example because it can be constructed by inspection and its nonnegativity
follows immediately by \eqref{eq:realim}.  Clearly, there exist nonnegative polynomials in $\gambles_R$ that are neither Hermitian sum-of-squares nor real sum-of-squares.

Why has Nature chosen Hermitian sum-of-squares? 
This is an open question that we will investigate in future work. A possible explanation may reside
in the different size of the corresponding optimisation problems \cite{josz2018lasserre}.
Another possible explanation is that  the class of \textit{Hermitian sum-of-squares} is always strictly included in the class of \textit{real sum-of-squares} polynomials.
Therefore, the former may be the smallest class of gambles that allows one to  efficiently determine whether a gamble is P-nonnegative, but that is
still expressive enough \cite[Proposition 6]{Benavoli2019d}.

%
%
%
%


\subsection{A hidden variable model for a single quantum system}
\label{sec:hidden1d}
In \cite{kochen1968problem},  Kochen and Specker gave a  hidden variable model agreeing with the Born's rule but not preserving the structure of functional dependencies in QT.
Their idea amounts to introducing a ``hidden variable'' for each observable $H$ producing stochasticity in outcomes
of measurement of $H$. The totality of all such hidden variables is then
the phase space variable $\omega$ of the model.

It turns out that, for a single quantum system, our  model based on the phase space
$$
\Omega=\{x\in \complex^n: x^{\dagger}x=1\},
$$
is also a hidden variable model. The reason being that, for a single quantum system, $\Sigma^{\geq}=\gambles_R^{\geq}$
and $\Sigma^{<}=\gambles_R^{<}$, and therefore Alice will never accepts negative gambles. Notice that in this case, the matrix $\rho=L(xx^{\dagger})$ can be interpreted as a moment matrix as discussed in Section \ref{sec:momentmat}.

This hidden-variable model for QT is also discussed by Holevo in \cite[Sec.~1.7]{holevo2011probabilistic}, \bluenew{where the author notices that such model does  not fulfil the key requirement imposed in many existing ``no-go'' theorems stating that a state is uniquely represented.
To appreciate this point, remember that in Section \ref{sec:momentmat} we noticed that a truncated moment matrix does not define a unique probability. Based on this nonuniqueness property, we have thence shown that QT is compatible with CPT in the case of a $n$-level single particle system. 
Now, notice that a  hidden variable theory would necessarily treat as distinct two probabilities that define the same density matrix.
However,  since they  are  underdetermined by the observations, these two probabilities are unidentifiable models.
In fact, since any real-valued observable  is described by a Hermitian operator
and the expectation of a Hermitian operator  w.r.t.\ a given density matrix (truncated moment matrix) is unique ($Tr(G\rho)$), the density matrix (truncated moment matrix) is sufficient to provide an adequate characterisation of these two probabilities.
Therefore, if we accept that a hidden-variable model does not require a one-to-one correspondence (because of unidentifiability),  a
hidden-variable model may be defined as an equivalence class (the set of all probabilities associated to a given truncated moment matrix). This nonuniqueness argument is also discussed in \cite{srinivas1982hidden} and \cite[Supplementary 3.2]{holevo2011probabilistic}}.

Finally, we notice that our hidden variable model differs from Holevo's one when we consider $m>1$ particles.
This is the topic of the next Section.


\subsection{On the use of tensor product}
\label{sec:tensorproduct}

In Section \ref{sec:qt} we saw that the possibility space of composite systems of $m$ particles, each one with $n_j$ degrees of freedom, is given by 
$\pspace=\prod_{j=1}^m \overline{\complex}^{n_j}$, and that gambles (real-valued observable) on such space  are actually bounded real functions $g(x_1,\dots,x_m)=(\otimes_{j=1}^m x_j)^\dagger G (\otimes_{j=1}^m x_j)$,  where $\otimes$ denotes the tensor product between vectors, seen as column matrices.

In what follows, we justify the use of tensor product, and more specifically the type of gambles on the possibility space of composite systems, as a consequence of the way a multivariate theory of probability is usually formulated. 

As a start, let us consider the case of 
classical probability. In CPT,  structural judgements of independence/dependence between random variables are expressed through product
between random variables: 
given  factorised gambles $g(x_1,\dots,x_m)=\prod_{j=1}^m g_j(x_j)$, the variables $x_1,\dots,x_m$ are said to be independent
if $E[\prod_{j=1}^m g_j(x_j)]=\prod_{j=1}^m E[g_j(x_j)]$ for all $g_j$, where $E[\cdot]$ denote the expectation operator.
With this in mind, let us go back to our setting. Marginal gambles are of type $g_j(x_j)=x_j^{\dagger}G_j x_j$. This means that structural judgements are performed by considering factorised gambles of the form $\prod_{j=1}^m x_j^{\dagger}G_j x_j$. It is thence not difficult to verify that
$$
\prod_{j=1}^m x_j^{\dagger}G_j x_j=(\otimes_{j=1}^m x_j)^\dagger (\otimes_{j=1}^m G_j) (\otimes_{j=1}^m x_j).
$$
By closing the set of factorised gambles  under the operations of addition and scalar (real number) multiplication, 
one finally gets a vector space whose domain domain coincides with the collection of all gambles of the form $(\otimes_{j=1}^m x_j)^\dagger G (\otimes_{j=1}^m x_j)$.
Hence, structural judgements of independence/dependence are stated by Alice considering the desirability
of gambles belonging to $\gambles_R$.


In QT, there is   a confusion over  the  role of the tensor product, used to define the state space, \rednew{that can actually be elucidated by looking at Theorem~\ref{th:fundamental}. In order to see this, let us consider}
 two quantum systems $A$ and $B$, with corresponding Hilbert spaces $\mathcal{H}_A$
and $\mathcal{H}_B$.  By duality, the density matrix (state) of the joint system  lives in the tensor product space $\mathcal{H}_A\otimes \mathcal{H}_B$.
\rednew{Indeed, we have that}
$$
\begin{aligned}
L((\otimes_{j=1}^2 x_j)^\dagger G (\otimes_{j=1}^2 x_j))&=L(Tr(G(\otimes_{j=1}^2 x_j)(\otimes_{j=1}^2 x_j)^\dagger))\\
&=Tr(G\; L((\otimes_{j=1}^2 x_j)(\otimes_{j=1}^2 x_j)^\dagger)),
\end{aligned}
$$
and $L((\otimes_{j=1}^2 x_j)(\otimes_{j=1}^2 x_j)^\dagger)$ belongs to $\mathcal{H}_A\otimes \mathcal{H}_B$.
However, when \eqref{eq:condepr0} holds, we may justify  entanglement hypothesising  the existence of non classical evaluation functions
or, equivalently, a larger possibility space (Theorem \ref{th:fundamental}). This is clearly discussed in \cite[Supplement 3.4]{holevo2011probabilistic}:
\begin{quote}
 Since the set of pure states of the composite system $\Omega$ is larger
than Cartesian product $\Omega_A\times \Omega_B$, the phase space of the classical description of the
composite system will be larger than the product of phase spaces for the
components: $\Omega_A\times \Omega_B \varsubsetneq \Omega$. Therefore this classical description is not a
correspondence between the categories of classical and quantum system
preserving the operation of forming the composite systems.
Moreover, it appears that there is no way to establish such a correspondence. In any classical description of a composite quantum system
the variables corresponding to observables of the components are necessarily entangled in the way unusual for classical subsystems.
\end{quote}
\rednew{To sum up, we} argue that $\Omega_A\times \Omega_B \varsubsetneq \Omega$ is a manifestation of computational rationality.



\subsection{Sum-of-squares optimisation}
\label{sec:sos}
The theory of moments (and its dual theory of positive polynomials) are used to develop efficient numerical
schemes for polynomial optimization, i.e., global optimization problems with polynomial function. 
Such problems arise  in the analysis and control of nonlinear dynamical systems, and also in other areas such as combinatorial optimization.
This scheme consists of a hierarchy of semidefinite
programs (SDP) of increasing size which define tighter and tighter relaxations of the original
problem. Under some assumptions, it can be showed that the associated sequence of optimal values converges to the global minimum, see for instance \cite{parrilo2003semidefinite,lasserre2009moments}.
Note that, every polynomial in 
$$
\Sigma^{\geq}\coloneqq\{g(x_1,\dots,x_m)=(\otimes_{j=1}^m x_j)^\dagger G (\otimes_{j=1}^m x_j): G\geq0\},
$$
is (Hermitian) sum-of-squares because it can be rewritten as:
$$
(\otimes_{j=1}^m x_j)^\dagger H H^{\dagger} (\otimes_{j=1}^m x_j)=\sum_{i=1}^k |(H^{\dagger} (\otimes_{j=1}^m x_j))_i|^2,
$$
with $G= HH^{\dagger}$.

We have recently discussed the connection between SOS optimisation (for polynomials of real variables) and computational rationality (also called bounded rationality) in \cite{Benavoli2019b}.

In QT, SDP has been used to numerically prove that a certain state is entangled  \cite{landau1988empirical,doherty2002distinguishing,doherty2004complete, wehner2006tsirelson,doherty2008quantum,navascues2008convergent,pironio2010convergent,bamps2015sum,barak2017quantum}. 
The work \cite{doherty2002distinguishing,doherty2004complete}  realized that the set of separable
quantum  states  can  be  approximated  by  sum-of-squares  hierarchies. 
This  leads  to  the  SDP hierarchy  of  Doherty-Parrilo-Spedalieri,  which  is  extensively  employed  in  quantum  information.

The present, purely foundational, work differs from these approaches by stating that the Universe (microscopic world)  is nothing but a big ``device'' that solves SOS
optimisation problems. In fact, as discussed in Section \ref{sec:discussions},  the postulate of computational efficiency embodied by~\ref{eq:btaut} (through the above definition of $\Sigma^{\geq}$, i.e., the cone of Hermitian sum-of-squares polynomials) may indeed be the fundamental law in QT.
This totally different perspective may have a relevant impact in quantum computing, such as the development of new algorithms for quantum computing that exploit
the connection between QT and SDP highlighted in this paper.

\subsection{Comparisons }\label{sec:comparisons}
\rednew{Before comparing our approach with attempts to reconstruct QT, we make the following general remark.
\\
 Although being model-based, and thus assuming a model of the input space, what is really only used by our  approach are the
minimal assumptions---linearity, tautology (non-negative gambles) and falsum (negative gambles)---
needed to capture the common structural properties of all rational choice theories \cite{nau1991arbitrage,zaffalon2017a,zaffalon2018a}: decision theory, game theory, social choice theory and market theory. Given those assumptions, it is then possible to derive CPT and QT, that is the definition of probability or state and all the corresponding rules.
On the contrary, in many of the approaches to quantum reconstruction, either the notion of probability or the notion of state
is already assumed. 
Our framework therefore provides in several respects a complementary path to quantum reconstruction, which reveals 
the essence of the ``quantum weirdness'':  the computational postulate. }

\subsubsection{Our previous approach}\label{subsec:otherax}

Partially inspired by \cite{pitowsky2003betting}, in \cite{benavoli2016quantum} we introduced a syntactically different framework from which we were able to derive the postulates of QT. It is roughly defined as follows. 
Assume the space of outcomes of experiment on a $n$-dimensional quantum systems is represented by the set $\Omega=\{\omega_1,\dots,\omega_n\}$,  with $\omega_i$ denoting the elementary event ``detection along $i$''. 
A gamble on an experiment 
is a Hermitian matrix $G \in \He^{n\times n}$. 
By accepting  a gamble $G$, Alice commits herself to receive  $\gamma_{i}\in \reals$ utiles  if the outcome of the experiment eventually happens to be 
$\omega_i$, where $\gamma_{i}$ is defined from $G$ and the projection-valued measurement $\Pi^{*}=\{\Pi_i^*\}_{i=1}^n$, representing the $n$ orthogonal directions of the quantum state, as follows:
 \begin{equation}
  \Pi^{*}_{i} G \Pi^{*}_{i} = \gamma_{i}\Pi^{*}_{i} \text{ for } i=1,\dots,n.
 \end{equation} 

A subset $ \He^{n\times n}$ is thus said to be coherent if it is a  closed convex cone $\mathcal{C}$ containing the set  $\{G\in\He^{n\times n}:G\gneq0\}$ of all PSD matrices in $\He^{n\times n}$ and disjoint from the interior of $\{G\in\He^{n\times n}:G \leq 0\}$.
We then proved that the dual of a coherent  convex cone of matrix gambles $\mathcal{C}$  is a closed convex set of density matrices
\cite[Prop.IV.3]{benavoli2016quantum}:
\begin{equation}
\label{eq:dualrhoo}
\begin{aligned}
\mathscr{Q}
&=\left\{ \rho \in  \He^{n\times n}: \rho\geq0, Tr(\rho)=1, Tr({G} \rho)\geq0,  ~~\forall G  \in \mathcal{C}\right\}.
\end{aligned}
\end{equation} 
When  $\mathcal{C}$ is a maximal\footnote{For a  definition, see Definition \ref{def:maximal} in the Appendix.} cone,  its dual $\mathscr{Q}:=\mathcal{C}^{\bullet}$ includes a single density matrix.
This allowed us to show that QT is a generalised theory of probability: its axiomatic foundation can be derived from a logical consistency requirement in the way a subject accepts gambles on the results of a quantum experiment (similar to the axiomatic foundation of classical probability).
However, at the time of writing \cite{benavoli2016quantum}, it was not clear to us why the probability is generalised in such a way in QT, what this could possibly mean, why does entanglement exist, etc.
That is, it was not clear to us that everything follows by the computation postulate \ref{eq:btaut}.

Below, by providing a bijection between coherent  convex cones of matrices and  P-coherent systems of polynomial gambles,
we show the connection between \cite{benavoli2016quantum}  and the present work.
In order to that, we first need to replace \cite[Definition III.1,(S3)]{benavoli2016quantum}, that is the \textit{openness} property, with (S3'): 
if $G+\epsilon I \in \mathcal{C}$ for all $\epsilon>0$ then  $G$ is in $ \mathcal{C}$ (\textit{closedness}). 
Hence, the bijection $\mathfrak{f}$ between $\mathcal{C}$ and $\bdomain$ is obtained  by first noticing the correspondence between  the duals  in Section \ref{sec:qt} and \eqref{eq:dualrhoo}, and thus simply composing
the duality maps from $\bdomain$ to $\mathscr{Q}$ and from $\mathscr{Q}$ to  $\mathcal{C}$: 

{\footnotesize
\[
\begin{tikzcd}[column sep=large, row sep=large]
\mathscr{Q}   \arrow{rd}{} 
  &  \arrow{l}{} \overbrace{\left\{\sum_{i=1}^{|\mathcal{G}|}\lambda_i (\otimes_{j=1}^m x_j)^{\dagger}G_i(\otimes_{j=1}^m x_j)+ (\otimes_{j=1}^m x_j)^{\dagger}M (\otimes_{j=1}^m x_j)~: ~~\lambda_i\geq0, M\geq0 \right\}  }^{\bdomain=posi(\mathcal{G}\cup\Sigma^{\geq})}\arrow{d}{\mathfrak{f}} \\
    &  \underbrace{\left\{\sum_{i=1}^{|\mathcal{G}|}\lambda_i G_i + M~:~ ~\lambda_i\geq0, M\geq0\right\}}_{\mathcal{C}}
\end{tikzcd}
\]}


Based on $\mathfrak{f}$ and the results in \cite{benavoli2016quantum}, it is therefore almost\footnote{We need to take into account that (S3) has been replaced by (S3').} immediate to derive from \ref{eq:btaut}--\ref{eq:b2} the remaining other axioms and rules of QT (such as L\"{u}der's rule (measurement updating), Schr\"odinger's rule (time evolution)), see Appendix~\ref{app:otherax}.

\subsubsection{QBism}\label{sec:qbysm}
The foundation of QT through desirability presented in this paper shares a goal similar to that of QBism, which is to justify quantum mechanics by using coherence arguments \cite{Caves02,Appleby05a,Appleby05b,Timpson08,longPaper, Fuchs&SchackII,mermin2014physics}.

QBism aims at rewriting QT in a way that 
does not involve the density operator $\rho$ and use classical probabilistic Bayesian arguments to prove coherence of QT.
For instance, in \cite{Fuchs&SchackII} QBists have shown that it is possible 
to rewrite Born's rule in a way that does not include the density operator $\rho$ anymore, so as to give it a classical (Bayesian) interpretation.
Starting from a \textit{symmetric informationally complete} (SIC) POVM  \cite{renes2003symmetric,appleby2007symmetric}, 
they rewrite the Born's rule  as:
\begin{equation} \label{eq:bornruleBayes0}
\begin{array}{rcl}
   p_j = \sum\limits_{i=1}^{n^2} \left((n+1)q_i - \dfrac{1}{n}\right) r(j|i) .\\
\end{array} 
\end{equation}
Here, $q_i$ represent the probability  of the outcomes of an experiment (a SIC measure with $n^2$ outcomes), which
is only imagined and never performed. The terms $r(j|i)$ are instead the probabilities of seeing the outcome $j$ in the 
actually performed experiment, given  the  outcome  $i$ in the imagined one. Finally, $ p_j$ is the probability of the 
outcomes in the actual experiment.  Fuchs and Shack  call~\eqref{eq:bornruleBayes0} the ``quantum law of total probability'', because they interpret it  as a modification  of the standard law of total probability:
\begin{equation} \label{eq:lawprob}
   p_j=\sum_{i=1}^{n^2} r(j|i) q_i,
\end{equation}
which is well known to follow from coherence, i.e., Dutch-book arguments  \cite{savage1954,finetti1970}.
Thus, according to Fuchs and Shack, the rule~\eqref{eq:bornruleBayes0} can be regarded as an addition to the standard Dutch-book coherence, a sort of ``coherence plus'', which defines the valid quantum 
states: 
\begin{quote}``It expresses a kind of empirically extended coherence---not implied by Dutch-book coherence alone, but formally similar to the kind of relation one gets from Dutch-book coherence'' \cite{Fuchs&SchackII}.
\end{quote} 
Under this view, Equation~\eqref{eq:bornruleBayes0} can be interpreted as a consistency requirement, i.e., a law that 
must hold to produce a consistent quantum world. 
%
We regard it as portraying a partial view: in fact, the law~\eqref{eq:bornruleBayes0} is not total probability, and hence, strictly speaking, is incoherent according to de Finetti's (classical) definition of coherence. Note in fact that, $(n+1)q_i - \tfrac{1}{n}$ for $i=1,\dots,n^2$
is not a probability distribution, \rednew{since it is  negative for some $i$.
Therefore, for QBism an important  question is, can the rule \eqref{eq:bornruleBayes0} still be derived from a Dutch-book argument?
If not, what type of coherence does this rule represent?}  

By applying our formalism, we can answer such questions.
Assume we are considering a single particle system of dimension $n$ and, therefore, the possibility space is 
$\overline{\complex}^{n}\coloneqq\{ x\in \complex^{n}: ~x^{\dagger}x=1\}$.
Then, note that  $q_i=Tr(Q_i L(xx^{\dagger}))$ and, we have shown in Section \ref{sec:qt}, that
in the single particle case $L(xx^{\dagger})$ can always be interpreted as a moment matrix (in the single particle case, a polynomial $x^{\dagger}Gx$ is non-negative iff 
$G\geq0$).
\rednew{
Thus,  $(n+1)q_i - \tfrac{1}{n}$ is an expectation and, therefore,  is coherent in de Finetti's sense!
According to our interpretation, we should not compare  \eqref{eq:bornruleBayes0} and \eqref{eq:lawprob}.
They cannot be compared, because the first is an expectation (of a quadratic form), while the second
is the law of total probability.}

\rednew{In multiple-particle systems, the coherence of \eqref{eq:bornruleBayes0} follows by \ref{eq:btaut}--\ref{eq:b2}. 
Hence, in this case, the ``modified coherence'' implied by the rule   \eqref{eq:bornruleBayes0}  is expressed by the computation postulate.}


\subsubsection{Operational approach and information-theoretic reconstructions}\label{sec:hardy}
The  operational framework starts with a description of a typical laboratory situation where there is a preparation procedure, followed by 
some transformation procedures and finally a measurement.
The first approach to derive QT within the operational framework is due to Lucien Hardy \cite{hardy2001quantum}. Hardy starts with a  set of unnormalized probabilities (States)
describing the preparation procedure, a set of measurements, and a set of transformations, and then imposes five ``physically reasonable'' postulates on top of
that framework. He thence shows that QT is the unique theory in the framework that satisfies these postulates.

This way of proceeding is common to all reconstructions based on the operational approach: the main differences
are to be found in  the choice of postulates, being  either physically reasonable
or based on  information-theoretic principles.

The initial set of states are in general unnormalized probabilities. This definition is so generic that there are no
constraints \cite{koberinski2018quantum} other than the stipulation that it ought to describe the  probability of the measurement outcomes.
As clarified in  \cite[Assumption 7]{barrett2007information}, the initial definition of allowed states (and transformations)
is (sometimes only implicitly) an axiom in the operational approach to quantum reconstructions.
Thus, one can think about the operational approaches   \cite{barrett2007information} as first defining  a \textit{generalized probabilistic theory} and then defining postulates to derive either CPT or QT (or other theories like Generalized Non-Signalling Theory or Generalized Local Theory).

It is worth to point out the following \cite{koberinski2018quantum}:

\begin{quote}
[These] reconstructions do not provide explanatory power in an even stronger sense: they do
not typically tell us what quantum states are, or what is really going on in the world when we perform a Bell experiment, for
example. 
\end{quote}


\rednew{
The approach we followed considers an explanatory model as starting, not from the outcomes of the experiment but, from 
 (hidden) variables describing the phenomenon of interest (e.g., the ``directions''  of the particle's spins), their domains (possibility space)
 and the functions of the variables we can observe (observables) and are interested in (queries).
\\
Given this view, we have thence shown that QT is an instance of a theory satisfying both  the coherence and computation postulates.
In doing so, we shed light on both what are quantum states and on the nature of Bell-like inequalities and entanglement.
Moreover, we clarify why there are cases where QT is in agreement
with CPT---that is precisely when inferences in CPT can be computed in polynomial time---, and cases
where QT cannot be reconciled with CPT---that is precisely when inferences in CPT cannot be computed
in polynomial time.}

\rednew{We want to point out again that one of the differences between our approach and the operational reconstruction, is the choice
of the possibility space. In the operational reconstruction, CPT is defined over a finite possibility space $\widetilde{\pspace}$
(the finite outcomes of a POVM measurement). Therefore, the set of all probabilities on $\widetilde{\pspace}$ is the simplex of probability
of dimension  $|\widetilde{\pspace}|-1$. The simplex has finite extremes.
To rule out  CPT and single out QT,  one (or more than one) axiom (for instance Hardy's axiom 5)
is introduced to imply that the convex set of states must have infinitely many extremes. 
In this paper, we instead follow a different route. The possibility space  $\pspace$ we consider is infinite  and, in such a space,
 CPT has also infinite extremes: these extremes are atomic probability charges on the elements of the possibility space $\pspace$.
  Therefore, from the perspective adopted in this work, the difference between CPT and QT is not in the finite versus infinite number  of extremes
  (they both have infinite extremes).
Instead, we have shown that the classical-quantum divide  originates from the computational postulate.} 

We have also made precise the term ``generalised probability theory''.
Generalised probability theories have a long tradition outside (quantum) physics:
they have been developed to model robustness (robust Bayesian model), to model the different facets of uncertainty,
to model incomplete preferences in market and game-theory, to model conditioning with zero probability events etc..

Probability theory is at the foundation of rational decision theory, but most importantly
the laws of probability can be derived from rational requirements in the way an agent make decisions.
One of the major perils in generalisations of probability theory  is to build theories that are  so generic that there are no
constraints. In these cases, reasoning with the theory can lead to logical contradictions.
 Avoiding that should be the premise of any generalised operational probabilistic theory.

To prevent that, a generalisation of probability must begin with a notion of logical consistency (coherence
or, in other words, a definition of rational choice).
In this paper, we have shown that the one that Nature uses in QT \rednew{follows by the coherence and computation postulate.}

%
%

\section{Conclusions}\label{sec:discussions}
Since its foundation, there have been two main ways to explain the differences between QT and classical probability. The first one, that goes back to Birkhoff and von Neumann \cite{birkhoff1936logic}, explains this differences  with the premise that,
in QT, the Boolean algebra of events is taken over by the ``quantum logic'' of projection operators on a Hilbert space.
The second one is based on the view that the quantum-classical clash is due to the appearance of negative probabilities \cite{dirac1942bakerian, feynman1987negative}.

Recently, there has been a research effort, the so-called ``quantum reconstruction'', which amounts to trying to rebuild the theory from more primitive postulates.

A common trait of all these approaches is that of regarding QT as a generalised theory of probability. But why is probability generalised in such a way, and what are quantum states? Why sometimes is QT compatible with classical probability theory?  We have shown that the answer to this question rests in the computational intractability of classical probability theory contrasted to the polynomial-time complexity of QT.

Note that there have been previous investigations into the computational nature of QT but they have mostly focused on topics of undecidability (these results are usually obtained via a limiting argument, as the number of particles goes to infinity, see, e.g., \cite{cubitt2015undecidability}; this does not apply to our setting as we rather take the stance that the Universe is a finite physical system) and of potential computational advantages of non-standard theories involving modifications of quantum theory
\cite{bacon2004quantum,aaronson2004quantum,aaronson2005quantum,chiribella2013quantum}.

The key postulate that separates classical probability and QT is~\ref{eq:btaut}: the computation postulate. Because of~\ref{eq:btaut}, Theorem~\ref{th:fundamental} applies and thus the ``weirdness'' of QT follows: negative probabilities, existence of non-classical evaluation functionals and, therefore, irreconcilability with the classical probabilistic view.
The formulation of Theorem~\ref{th:fundamental} points to the fact that there are
 three possible ways out to provide a theoretical foundation of QT: (1) redefining the notion of evaluation functionals (algebra of the events), which
 is the approach adopted within the \emph{Quantum Logic} \cite[Axiom VII]{mackey2013mathematical};
 (2) the algebra of the events is classical but probabilities are replaced  by quasi-probabilities (allowing negative values), see for instance \cite{schack2000explicit,ferrie2011quasi};
 (3) the quantum-classical contrast has a purely computational character.
 The last approach starts by accepting P$\neq$NP to justify the
 separation between the microscopic quantum system and the macroscopic world.
 We quote  Aaronson \cite{aaronson2005guest}:
 \begin{quote}
... while experiment will always be the last appeal, the
 presumed intractability of NP-complete problems might be taken as a useful constraint in the search
 for new physical theories.
 \end{quote}
The postulate of computational efficiency embodied by~\ref{eq:btaut} (through~\eqref{eq:b0}) may indeed be the fundamental law in QT, similar to 
the second law of thermodynamics, or the impossibility  of superluminal signalling.

 
 \rednew{As future work we aim to address the following question: can we make the choice of nonnegative gambles (Hermitian sum-of-squares in QT) to follow from an additional postulate? This will turn our approach into a full quantum reconstruction. In particular, we would like to investigate the possibility of characterising 
 the choice of nonnegative gambles via symmetry arguments (for instance by exploiting  ideas from \cite{skilling2017symmetrical}).
 We also plan to compare our approach with that of Goyal and Knuth \cite{skilling2017symmetrical}. Here the authors prove that CPT and QT can be 
 derived from  the same principles, once  the appropriate symmetries (that are operative in each domain) are identified.
  In particular, we would like to verify if  the the compatibility (incompatibility) between 
 CPT and QT we proved for  single (respectively multiple) particle systems is also   somehow present in their approach.
 We finally aim  to extend our approach to the case of indistinguishable particles.}

\bibliographystyle{ieeetr}
\bibliography{biblio}

\clearpage
 
\appendix
 \setcounter{section}{18}

 \section{Further results and derivations }
 \label{sec:appendix0}
 We introduce the following notations and definitions that will be used in the rest of the Appendix.
 Recall that $\posi$ denotes the conic hull operator, whereas by $\cl$ we denote the closure operator.
 \begin{itemize}
  \item $\mathbbm{1}$ denotes the unitary gamble in $\gambles$.  We have introduced the symbol $\mathbbm{1}$ to distinguish the unitary function in $\gambles$, i.e.,  $\mathbbm{1}(\omega)=1$ for all $\omega \in  \Omega$, from the scalar (real number) $1$.
  \item $\mathbbm{1}_R$ denotes the unitary gamble in $\gambles_R$.
  \item $N(\assess):=\cl\posi(\nonnegative \cup \assess)$.
  \item $N_R(\assess):=\cl\posi(\nonnegative_R \cup \assess)$.
  \item $N_B(\assess):=\cl\posi(\bnonnegative \cup \assess)$.
 \end{itemize}

 \begin{definition}
 \label{def:maximal}
  A coherent (equiv. P-coherent) set of desirable gambles $\domain$ is said to be \textbf{maximal} when
  for every $g\in \gambles \backslash \domain$ (respectively  $g\in \gambles_R \backslash \domain$) $-1 \in \cl \posi(g \cup \domain)$.
 \end{definition}
 It means that, if $\domain$ is maximal, we cannot enlarge it while preserving coherence (equiv. P-coherence).

\subsection{The unitary constant plays the role of falsum}
\label{sec:onefalsum}
The following result, in addition to providing a necessary and sufficient condition for coherence,  states that $-\mathbbm{1}$ can be regarded as playing the role of the Falsum and \ref{eq:sl} can be reformulated as $-\mathbbm{1} \notin \domain$. 
It is an immediate consequence of  Theorem 3.8.5 and Claim 3.7.4 in \cite{Walley91}.
\begin{proposition}\label{prop:coco1}
Let $\assess$ be a set of gambles. The following claims are equivalent
\begin{enumerate}
\item $N(\assess)$ is coherent, 
\item $\posi (\assess) \cap \negative = \emptyset$,
\item $-\mathbbm{1} \notin \posi (\assess \cup \nonnegative) $,
\item $g \notin N(\assess)$, for some gamble $g\in \gambles$.
\end{enumerate}
\end{proposition}

Postulate \ref{eq:sl}, which presupposes postulates \ref{eq:taut} and \ref{eq:NE},  provides the normative definition of TDG, referred to by $\theory$. Based on it, in Subsection \ref{subsec:dual} we derive the axioms of classical, Bayesian, probability theory. 
This is simply based on the fact that, geometrically,  $\domain$ is a closed convex cone.  It is thence clear from the above definition that $\nonnegative$ is the minimal coherent set of desirable gambles. It characterises a state of full ignorance -- a subject
without knowledge about $\omega$ should only  accept nonnegative gambles.
Conversely, a coherent set of desirable gambles is called \emph{maximal} if there is no other coherent set of desirable gambles including it
(see Definition \ref{def:maximal}). 
In terms of rationality, a maximal coherent set of desirable gambles is a set of gambles that Alice cannot extend  by accepting
other gambles while keeping at the same time rationality. It also represents a situation in
which Alice is sure about the state of the system. 

\subsection{Inference in $\theory$}
\label{sec:inferTDG}
In the operational interpretation of $\theory$, agents can buy/sell gambles from/to each other.
Therefore, an agent must be able to determine the selling/buying prices for gambles.
This can be formulated as an inference procedure on  $\domain$. For simplicity, we consider finite sets of assessments, and denote by $| A|$  the cardinality of a set $A$.

\begin{definition}
 Let $\assess$ be a finite set of assessments of desirability, and $N(\assess)$ be a coherent set of
desirable gambles. Given  $f \in \gambles$, we denote with
\begin{equation}
\label{eq:lp}
\begin{aligned}
 \underline{E}(f):=&\sup_{\gamma_0\in \reals, \lambda_i \in \reals^+} \gamma_0\\
 &s.t:\\
 &f(\omega) - \gamma_0  -\sum\limits_{i=1}^{|\mathcal{G}|} \lambda_i g_i(\omega) \geq0~~~~\forall \omega \in \Omega,
\end{aligned}
\end{equation}
the lower prevision of $f$. The upper prevision of $f$ is equal to $\overline{E}(f)=-\underline{E}(-f)$.
\end{definition}

The lower prevision of a gamble is Alice's supremum buying price for $f$, i.e., how much she should pay to buy
the gamble $f$. The upper prevision is Alice's infimum selling price for $f$, i.e., how much she should ask to sell 
the gamble $f$. We will show in Appendix \ref{sec:inferenceCPT} that the lower and upper prevision are just 
 the lower and upper expectation for the gamble $f$.
By exploiting \ref{eq:NE}, we can equivalently rewrite \eqref{eq:lp} as:
\begin{equation}
\label{eq:lpne0}
\begin{aligned}
 \underline{E}(f)=&\sup_{\gamma_0\in \reals, \lambda_i \in \reals^+} \gamma_0\\
 &s.t:\\
 &f - \gamma_0\mathbbm{1}  -\sum\limits_{i=1}^{|\mathcal{G}|} \lambda_i g_i \in \nonnegative,
\end{aligned}
\end{equation}
or equivalently,
\begin{equation}
\label{eq:lpne1}
\begin{aligned}
 \underline{E}(f)=&\sup_{\gamma_0\in \reals} \gamma_0 ~~~s.t.~~~ f- \gamma_0\mathbbm{1}\in N(\assess).
\end{aligned}
\end{equation}

In other words, we have expressed the constraint in the above optimisation problems as a membership.

\subsection{Probabilistic interpretation through duality}\label{subsec:dual}
We aim to show that the \textit{dual} of a coherent set of desirable gambles is a closed convex set
of probability charges:
\begin{equation}
\mathcal{P}=\left\{\mu\in \mathcal{M}^{\geq}: \int \mathbbm{1} d\mu=1,~~\int gd\mu \geq0,~ ~\forall g \in \domain \right\},
\end{equation}
where $\mathcal{M}^{\geq}$ is the set of nonnegative charges.

The key point in the duality proof  is that $\gambles^{\geq}$ (the set of all nonegative gambles (real-valued bounded function) on $\pspace$)
includes indicator functions.\footnote{An indicator function defined on a subset $A\subseteq \pspace$ is a function that is equal to one 
for all elements in $A$ and $0$ for all elements outside $A$.}
This is crucial to prove that the dual of $\domain$ is always included in $\mathcal{M}^{\geq}$. When this is not the case,
the dual of a coherent set of desirable gambles is not anymore a convex set of probabilities.

Note that, equipped with the supremum norm, $\gambles$  constitutes a Banach space, and its topological dual $\mathcal{L}^*$ is the space  of all bounded  functionals on it. We assume the weak${}^*$ topology on $\mathcal{L}^*$.

Let $\mathcal{A}$ be the algebra of subsets of $\Omega$  and $\mu:\mathcal{A} \rightarrow [-\infty,\infty]$ 
denotes a charge: that is $\mu$ is a  finitely additive set function of $\mathcal{A}$  \citep[Ch.11]{aliprantisborder}, \citep{bhaskara1983},
that can take positive and negative values.
We have that 
 every gamble on $\mathcal{L}$ 
 is integrable with respect to any finite charge \citep[Th.11.8]{aliprantisborder}. 
Therefore, for any gamble $g$ and finite charge $\mu$ we can define  $\int gd\mu$, which we can interpret as
a linear functional  $L(\cdot):=\langle \cdot, \mu\rangle$ on  
$g$. 
We denote by $\mathcal{M}$ the set of all finite charges on $\mathcal{L} $ and  by $\mathcal{M}^{\geq}$ the set of nonnegative charges. 
 $\mathcal{M}$ is isometrically isomorphic to $\mathcal{L}^*$. The duality bracket between  $\mathcal{L} $ and $\mathcal{M}$ is given by $\langle f, \mu \rangle := \int f d\mu$, with $f \in \mathcal{L} $ and $\mu \in \mathcal{M}$. 

A linear functional $L$ of gambles is said to be \emph{nonnegative} whenever it  satisfies : $L(g) \geq 0$, for $g \in \nonnegative$. 
A nonnegative linear functional is called a \emph{state} if moreover it preserves the unitary constant gamble. 
In our context, this means $L(\mathbbm{1}) = \langle \mathbbm{1}, \mu\rangle=\int \mathbbm{1} d\mu=1$, i.e., the linear functional is scale preserving. Hence,  the set of states $\mathscr{S}$ corresponds to the closed convex set of  all probability charges.

We define the \emph{dual} of a subset $\domain$ of $\mathcal{L}$ as:\\
\begin{equation}
\domain^\bullet=\left\{\mu\in \mathcal{M}: \int gd\mu \geq0, ~\forall g \in \domain \right\}.
\end{equation}


Similarly, the dual of a subset $\mathcal{R}$ of $\mathcal{M}$ is the set:\\
\begin{equation}
\mathcal{R}^\bullet=\left\{g\in \mathcal{L}: \int gd\mu \geq0, ~\forall \mu \in \mathcal{R}\right\}.
\end{equation}

Note that in both cases $(\cdot)^\bullet$ is always a closed convex cone \citep[Lem.5.102(4)]{aliprantisborder}. Furthermore,  one has that $(\cdot)^\bullet{^\bullet}=(\cdot)$, whenever $(\cdot)$ is a closed convex cone  \citep[Th.5.103]{aliprantisborder}, and   $(\cdot)_1 \subseteq (\cdot)_2$ if and only if $(\cdot)_2^\bullet \subseteq (\cdot)_1^\bullet$ \citep[Lem.5.102(1)]{aliprantisborder}. In particular, whenever $(\cdot)_1$ and $(\cdot)_2$ are closed convex cones, $(\cdot)_1 \subsetneq (\cdot)_2$ if and only if $(\cdot)_2^\bullet \subsetneq (\cdot)_1^\bullet$.

Based on those facts, it is thus possible to verify that the dual of a coherent set of desirable gambles can actually be completely described in terms of a (closed convex) set of states (probability charges). 
In this aim, we start by the following observations.

\begin{proposition}\label{prop:czesc}
 It holds that 
 \begin{enumerate}
   \item $(\mathcal{L})^\bullet=\{0\}$ and $\mathcal{L}=(\{0\})^\bullet$;
 \item $(\nonnegative)^\bullet=\mathcal{M}^{\geq}$ and $\nonnegative=(\mathcal{M}^{\geq})^\bullet$;
  \end{enumerate}
\end{proposition}
\begin{proof} Since $(\cdot)^\bullet{^\bullet}=(\cdot)$, whenever $(\cdot)$ is a closed convex cone, in both cases it is enough to verify only one of the claims.
For the first item, the second claim is immediate.
For the second item, we verify the first claim.
The inclusion from right to left being clear, for the other direction  observe that: (i) $ g=I_{\{B\}}$ (with $I_B$ being the indicator function on $B \in \mathcal{A}$), is a nonnegative gamble and, therefore, is in $\nonnegative$; (ii) if $\mu$ is negative in $B \in \mathcal{A}$, i.e.,  then $\int I_{B}d\mu$ is negative too and, thus, $\mu$ cannot be in $(\nonnegative)^\bullet$.  
\end{proof}
\begin{proposition}\label{prop:czesc2}
Let $\domain $ be a closed convex cone. The following claims are equivalent
 \begin{enumerate}
   \item $\domain$ is coherent;
     \item $\domain\supseteq \nonnegative$ and $\domain \neq \mathcal{L}$;
   \item $(\domain)^\bullet\subseteq \mathcal{M}^{\geq}$ and $(\domain)^\bullet \neq \{0\}$.
  \end{enumerate}
\end{proposition}
\begin{proof} (2) $\Leftrightarrow$ (3): From Proposition \ref{prop:czesc}, $(\domain)^\bullet = \{0\}$ if and only if $\domain = \mathcal{L}$, and $\domain\supseteq \nonnegative$ if and only if $(\domain)^\bullet\subseteq \mathcal{M}^{\geq}$.
\\ (1) $\Rightarrow$ (2): Assume $\domain$ is coherent. By A.1 $\domain \supseteq \nonnegative$ and by A.2 there is $g \in \mathcal{L}$ such that $\sup g < 0$ and $g \notin \domain$. 
\\ (2) $\Rightarrow$ (1): Let $\nonnegative \subseteq \domain \subsetneq \mathcal{L}$. First of all, notice that, by Proposition \ref{prop:czesc}, $(\domain)^\bullet\subseteq (\nonnegative)^\bullet = \mathcal{M}^{\geq}$. 
Now, assume that $\domain$ is not coherent. This means A.2 fails, i.e. there is $g \in \mathcal{L}$ such that $\sup g < 0$ and $g \in \domain$. 
Hence, consider $\mu \in \mathcal{M}^{\geq}$ and pick $g \in \domain$ such that $\sup g < 0$. It holds that $\langle g, \mu\rangle \geq 0$ if and only if $\mu=0$, meaning that $(\domain)^\bullet=\{0\}$ and therefore, by Proposition \ref{prop:czesc} again, $\domain=\mathcal{L}$, a contradiction.
\end{proof}

Finally, we can prove 

\begin{theorem}
\label{prop:dualcharges0}
The map
\[\domain \mapsto  \mathcal{P}:=\domain^\bullet \cap \mathscr{S} \]
establishes a bijection between coherent sets of desirable gambles and non-empty closed convex sets of states.
\end{theorem}
\begin{proof}
We want to verify that the map
\[\domain \mapsto  \mathcal{P}:=\domain^\bullet \cap \mathscr{S} \]
establishes a bijection between coherent sets of desirable gambles and non-empty closed convex sets of states.
The proof is analogous to that by \citep[Th.4]{pmlr-v62-benavoli17b}.
 Let $\domain$ be a coherent set of desirable gambles. By Proposition \ref{prop:czesc2}, we get that $\domain^\bullet$ is a closed convex cone included in $\mathcal{M}^{\geq}$ that does not reduce to the origin.  
Thus, after normalisation, $\mathcal{P}$ is nonempty. Preservation of closedness and convexity by finite intersections yields that $\mathcal{P}$ is closed and convex. Furthermore $\mathbb{R}_+\mathcal{P}=\domain^\bullet$, and therefore $\domain= (\mathbb{R}_+\mathcal{P})^\bullet$, where $\mathbb{R}_+\mathcal{P}:=\{ \lambda\mu : \lambda \geq 0, \mu \in \mathcal{P}\}$, meaning that the map is an injection.
We finally verify that the map is also a surjection. To do this, let $\mathcal{P}$ be a non empty closed convex set of probability charges. It holds that $\mathbb{R}_+\mathcal{P}$ is a closed convex cone included in $\mathcal{M}^{\geq}$ different from $\{0\}$. Again by Proposition \ref{prop:czesc2},  we  conclude that the dual $(\mathbb{R}_+\mathcal{P})^\bullet$ of $\mathbb{R}_+\mathcal{P}$ is a coherent set of desirable gambles and $\mathcal{P}=\mathbb{R}_+\mathcal{P}\cap \mathscr{S}= (\mathbb{R}_+\mathcal{P}){{}^\bullet{^\bullet}}\cap \mathscr{S}$. 
\end{proof}

This means that we can write the dual of $\domain$ as the set
\begin{equation}
\label{eq:ddual}
 \mathcal{P}=\left\{\mu\in \states: L_\mu(g) \geq0,~\forall g \in \domain\right\},
\end{equation} 
which is a closed convex-set of probability charges.  We have derived the axioms of probability: a non-negative function that integrates to one.
Moreover, it can be easily proven that $ \mathcal{P}$ is a singleton whenever $\domain$ is a \textit{maximal} coherent set of desirable gambles.
In other words, for \textit{maximal} $\domain$, there is only one probability charge $\mu$ that is compatible with the constraints 
$L_\mu(g) \geq 0$.

As an immediate corollary of the previous results, we finally obtain Theorem \ref{cor:noncoherent}.
It provides us with a characterisation of the dual of a closed convex cone which is not coherent.
\begin{theorem}\label{cor:noncoherent}
Let $\domain $ be a non empty closed convex cone. Then the following are equivalent
\begin{enumerate}
 \item $\domain  \neq \mathcal{L}$  and $\domain$ is not coherent;
 \item 
 $\domain^\bullet\not\subseteq \mathcal{M}^{\geq}$,
 \item $\domain^\bullet \cap \{ \mu \in \mathcal{M} \mid \langle \mathbbm{1}, \mu \rangle =1 \} \not \subseteq \mathscr{S}$.
  \item $\{0\}\subsetneq\domain^\bullet$ and $\domain^\bullet \cap   \mathscr{S}= \emptyset$.
 \end{enumerate} 
\end{theorem}
Essentially, Theorem \ref{cor:noncoherent} is telling us that, from the dual point of view, non degenerated closed convex cones of gambles that are not coherent are characterised by signed (non positive) charges. 

\begin{proofof}{Theorem~\ref{cor:noncoherent}}
Given  a non empty closed convex cone $\domain $, we need to check  that the following claims are equivalent
\begin{enumerate}
 \item $\domain  \neq \mathcal{L}$  and $\domain$ is not coherent;
 \item 
 $(\domain)^\bullet\not\subseteq \mathcal{M}^{\geq}$,
 \item $\domain^\bullet \cap \{ \mu \in \mathcal{M} \mid \langle 1, \mu \rangle =1 \} \not \subseteq \mathscr{S}$.
  \item $\{0\}\subsetneq\domain^\bullet$ and $\domain^\bullet \cap   \mathscr{S}= \emptyset$.
 \end{enumerate} 
Notice that $\domain^\bullet \supseteq \{0\}$, and that, by Proposition  \ref{prop:czesc}, all claims imply $\domain^\bullet \neq \{0\}$ and $\domain \neq \gambles$. The equivalence between (2) and (3) being obvious, the one between (1) and (2) is an immediate consequence of Proposition \ref{prop:czesc2}.  For the equivalence between (1) and (4), first of all notice that, since $\domain $ is a non empty closed convex cone, 
$\{0\}\subseteq\domain^\bullet$ if and only if $\domain  \neq \mathcal{L}$. Now, assume  $\domain$ is not coherent.  This means there is a negative gamble $g$ in  $\domain$. Hence, $\domain^\bullet \cap   \mathscr{S}= \emptyset$ since probability charges do not preserves negative gambles. Finally, if $\domain$ is coherent, by Theorem \ref{prop:dualcharges0},   $\domain^\bullet \cap   \mathscr{S} $ is a non-empty closed convex cone of probability charges. 
\end{proofof}

%
 
 \subsection{Inference in CPT}
 \label{sec:inferenceCPT}
As already mentioned, once we have defined the duality between TDG and probability theory, we can immediately reformulate the lower and upper previsions by means of probabilities.
Indeed, let $\assess$ be a finite set of assessments, and $\domain:= N(\assess)$ be a coherent set of
desirable gambles. Given  $f \in \gambles$, the lower prevision of $f$ defined in Equation \eqref{eq:lp} can also be computed as:
\begin{equation}
\label{eq:lpdual}
\begin{aligned}
 \underline{E}(f):=&\inf_{\mu \in \mathscr{S}} \int_{\Omega} f(\omega) d\mu\\
 &s.t:\\
  &\int_{\Omega} g_i(\omega) d\mu\geq0, \forall ~ g_i\in \mathcal{G}.\\
\end{aligned}
\end{equation}
which is equivalent to
\begin{equation}
\label{eq:lpdual1}
\begin{aligned}
 \underline{E}(f):=&\inf_{\mu \in \mathcal{P}} \int_{\Omega} f(\omega) d\mu.
\end{aligned}
\end{equation}
 The upper prevision of $f$ is defined $\overline{E}(f)=-\underline{E}(-f)$.

Hence, the lower and upper prevision of $f$ w.r.t.\ $\domain$ are just the \textit{lower and upper expectation} of $f$
w.r.t.\ $\mathcal{P}$.
In case $\domain$ is maximal, then $\mathcal{P}$ includes only a single probability and, therefore, in this case:
$$
\underline{E}(f)=\overline{E}(f)=E(f),
$$
i.e., the expectation of $f$.
That is, the solution of \eqref{eq:lpdual} coincides with the expectation of $f$.
We have considered the more general case $\underline{E}(f)\leq \overline{E}(f)$  because, as we have discussed in Section \ref{sec:momentmat}, QT is a theory of ``imprecise'' probability \cite{walley1991}.\footnote{The term ``imprecise'' refers to the fact that the closed convex set $\mathcal{P}$ may not be a singleton, that is the probability may not be ``precisely'' specified. Imprecise probability theory  is also referred as robust Bayesian.}

  \subsection{Relative coherence}\label{app:comp}

   Why do we need $\theory$ if it is dual to classical probability theory?
   The point is that, $\theory$, having the structure of an abstract logic, is independent of the specific contingent properties of the underlying space of gambles. From this perspective, it is thence  more general than probability theory and can be used to model any circumstance in which an agent has to make a rational choice.
   
For instance, we can define the desirability postulates on any vector subspace of $\gambles$ that includes
constant gambles. This is useful in applications where the gambles Alice can examine are a subset of $\gambles$.\footnote{
For example, in Finance the tradable gambles are usually piecewise polynomials.
An European call option on the future value $x$ of an underlying security  with strike $k$ is mathematically expressed by the gamble $\max(x - k, 0)$.}

In what follows we verify under which circumstances the  dual is still the classical theory of probability. Conversely, when we change the notion of being nonnegative,  we can get  weaker theory of probabilities. 

Let  $\gambles_R$ denote a closed linear subspace of $\gambles$ that includes all constant gambles. By  $\nonnegative_R = \nonnegative \cap \gambles_R$ we denote the subset of nonnegative gambles, and by  $\negative_R = \negative \cap \gambles_R$ we denote the subset of negative gambles. 
We thus relativise postulates \ref{eq:taut}--\ref{eq:sl} as follows. 
Firstly, we restrict tautologies, and thus the extension of the predicate being nonnegative, to the considered subspace.
\begin{enumerate}[label=\upshape A$_R$0.,ref=\upshape A$_R$0]
\item\label{eq:tautR} $\nonnegative_R$ should always be desirable.
\end{enumerate}

Secondly, we modify the deductive closure accordingly.
\begin{enumerate}[label=\upshape A$_R$1.,ref=\upshape A$_R$1]
\item\label{eq:NER} $N_R(\assess_R)\coloneqq\cl \posi(\nonnegative_R\cup \mathcal{G}_R)$.
\end{enumerate}
Notice that $N_R(\assess_R)\subseteq \gambles_R$, for $\assess_R \subseteq \gambles_R$. In such case we sometimes denote $N_R(\assess_R)$ by $\domain_R$. 

Hence, finally, we state that: 
\begin{definition}[Relative coherence postulate]
\label{def:avsR}
 A set $\domain_R \subseteq \gambles_R$  is \emph{coherent relative to } $\gambles_R$  if and only if
 \begin{enumerate}[label=\upshape A$_R$2.,ref=\upshape A$_R$2]
\item\label{eq:slR} $ \negative_R \cap \domain_R=\emptyset$.
\end{enumerate}
\end{definition}

It is easy to check that $\theory$ and its relativisation to $\gambles_R$ defined by Postulates \ref{eq:tautR}--\ref{eq:slR}, denoted by $\theory_R$, are fully compatible. That is: 

\begin{theorem}\label{thm:cocore}
Let $\mathcal{G}_R$ be a set of assessments in $\gambles_R$.
$N_R(\mathcal{G}_R)$ is coherent in $\theory_R$ if and only if $N(\mathcal{G}_R)$ is coherent in $\theory$. Moreover $N_R(\mathcal{G}_R)= \gambles_R \cap N(\mathcal{G}_R)$.
\end{theorem}

 
 Finally, we compare inference in theory $\theory_R$ with inference in the classical theory $\theory$.
 
 Let $\assess$ be a finite set of assessment in $\gambles_R$, and $N_R(\assess)$ be a coherent set of
desirable gambles in $\theory_R$. The lower prevision of a gamble $f \in \gambles_R$  is defined  as
\begin{equation}
\label{eq:lpner}
\begin{aligned}
 \underline{E}_R(f):=&\sup_{\gamma_0\in \reals, \lambda_i \in \reals^+} \gamma_0\\
 &s.t:\\
 &f - \gamma_0\mathbbm{1}_R  -\sum\limits_{i=1}^{|\mathcal{G}|} \lambda_i g_i \in \nonnegative_R,
\end{aligned}
\end{equation}
where $\mathbbm{1}_R$ denotes the unitary gamble in $\gambles_R$, i.e.,  $\mathbbm{1}_R(\omega)=1$ for all $\omega \in  \Omega$.
The upper prevision is denoted as $ \overline{E}_R(f)=-\underline{E}_R(-f)$. Comparing \eqref{eq:lpne0} and \eqref{eq:lpner}, the reader can notice that
 $f - \gamma_0\mathbbm{1}  -\sum_{i=1}^{|\mathcal{G}|} \lambda_i g_i \in \nonnegative$ in  \eqref{eq:lpne0}
becomes $f - \gamma_0\mathbbm{1}_R  -\sum_{i=1}^{|\mathcal{G}|} \lambda_i g_i \in \nonnegative_R$ in \eqref{eq:lpner}.

Since the inclusion in $\nonnegative$ is a weaker constraint then the inclusion in $\nonnegative_R$,
we have that $\underline{E}_R(f)\leq \underline{E}(f)$.
However, since $\gambles_R$ is a linear subspace of $\gambles$ that includes all constant gambles, with $\nonnegative_R = \nonnegative \cap \gambles_R$, and $N_R(\assess)$ is coherent (in $\theory$), we therefore have that, for every $f \in \gambles_R$
$$
\stackrel{\substack{\theory_R\\~}}{\underline{E}_R(f)} ~~=~~ \stackrel{\substack{\theory\\~}}{\underline{E}(f)} ~~\leq ~~  \stackrel{\substack{\theory\\~}}{\overline{E}(f)}  ~~=~~ \stackrel{\substack{\theory_R\\~}}{\overline{E}_R(f)},
$$
meaning that $\underline{E}_R(f)$ can be interpreted as a lower expectation. In particular, this tell us that, from an operational point of view, $\theory_R$ works exactly as $\theory$.

The same can also be observed by looking at characteristics of the dual of a coherent set in $\theory_R$.
Indeed,   assume $\gambles_R$ is a closed subspace of $\gambles$. Then, endowed with the relative topology, the continuous linear functions on $\gambles_R$ are exactly the restriction to $\gambles_R$ of the continuous linear functionals on $\gambles$  \citep[Theorem 5.87]{aliprantisborder}. 
As $\gambles_R$ includes all constant gambles, by the  Riesz-Kantorovich extension theorem, any positive functional on $\gambles_R$ can be extended (possibly non uniquely) to a positive functional on $\gambles$. This means that, given the correspondence in Theorem \ref{thm:cocore},  the dual of a set $\domain_R$ coherent in $\theory_R$ can be identified with the closed convex set of probability charges $(N(\domain_R))^\bullet$.

To conclude, whenever the subspace $\gambles_R$ is clear and Theorem \ref{thm:cocore} holds, we can identify $\theory$ and $\theory_R$.

 \subsection{P-coherence and Theorem \ref{th:fundamental}}\label{app:complex}
 

 In what follows we provide some characterisation of P-coherence.

  \begin{proposition}\label{prop:cohe}
  Assume that $\bnonnegative$ is closed (and $\bnegative$ is not empty).
Let $\assess \subseteq \gambles_R$ a set of assessments. The following are equivalent
\begin{enumerate}
 \item $-\mathbbm{1}_R \notin \posi ( \assess \cup \bnonnegative)$
    \item $\posi ( \assess \cup \bnonnegative) \cap \bnegative = \emptyset$
     \item $N_B(\assess)$ is P-coherent
\end{enumerate}
  \end{proposition}
  \begin{proof}
Next, for the equivalence between items (1)-(3), first of all, notice that $\posi ( \assess \cup \bnonnegative) \subseteq \cl \posi ( \assess \cup \bnonnegative) \subseteq N_B(\assess)$. Hence, (3) implies (2) implies (1). For the remaining implications we reason as follows. Assume that (2) holds, and assume $f + \delta \in \posi (\assess \cup \bnonnegative)$, for every $\delta > 0$. This means $f \in \cl(\posi (\assess \cup \bnonnegative))$.  Suppose $f \in \bnegative$, then since $\bnegative$ is open, $f + \delta \in \bnegative$ for some $\delta > 0$,
contradicting P-coherence of $\posi ( \assess \cup \bnonnegative)$. We therefore conclude that $f \notin \bnegative$, and that (3) holds. 
Now, assume (2) does not hold, i.e. $f \in \bnegative$ and $f \in  \posi ( \assess \cup \bnonnegative)$. Hence, $-f$ is in the interior of $\bnonnegative$, meaning that for some $\delta > 0$, $-f - \delta = g \in \bnonnegative$. From this we get that $-\mathbbm{1}_R = \frac{g + f}{\delta} \in  \posi ( \assess \cup \bnonnegative)$: (1) does not hold.
  \end{proof}
  
  Analogously to $\theory$, one can ask if and when $N_B$ is a closure operator whose class of non-trivial closed sets coincide with the P-coherent sets, or stated otherwise, if and when  $N_B$ associates to each $\bdomain \subseteq \gambles_R$ the intersection of all P-coherent sets that include $\bdomain$. It turns out that we need to add some conditions to the structural properties of $\btheory$ to obtain such property:
  
  \begin{proposition}\label{prop:eqq}
Assume that $\bnonnegative$ satisfies
\begin{description}
\item[(*)] for every $f \in \gambles_R$, there is $\epsilon > 0$ such that $f+\epsilon \in \bnonnegative$.
\end{description}
Let $\assess \subseteq \gambles_R$ a set of assessments. The following are equivalent
\begin{enumerate}
     \item $N_B(\assess)$ is P-coherent
    \item $N_B(\assess) \neq \gambles_R$
\end{enumerate}
  \end{proposition}
\begin{proof}
Since (1) implies (2), we need to verify the other direction. Assume $N_B(\assess)$ is not P-coherent. By Proposition \ref{prop:cohe}, $-1 \in N_B(\assess)$, and thus $- \epsilon \in N_B(\assess)$, for every $\epsilon \geq 0$. Let $f \in \gambles_R$. By (*) there is $\epsilon > 0$ such that $f+\epsilon \in \bnonnegative \subseteq N_B(\assess)$.  Hence, by closure under linear combinations,  $f + \epsilon + (- \epsilon) = f \in  N_B(\assess)$.
\end{proof}
  Notice that condition (*) is satisfied by the P-coherent model $\btheory$ of QM introduced in Section \ref{sec:qt}, as well as by the model of Section \ref{sec:ent_not_only}. On the other hand, (*) is not the only condition to force the equivalence between the two claims in Proposition \ref{prop:eqq}. 
  We plan for future work to study the natural condition on $\bnonnegative$ related to the property for $N_B$ to be the closure operator induced by P-coherent sets.

    In what follows, we prove the main theorem of the paper.
  
 \begin{proofof}{Theorems \ref{th:fundamental}}
  Assume that $\gambles_R$ includes all positive constant gambles and that $\bnonnegative$ is closed (in $\gambles_R$). 
Let $\bdomain \subseteq \gambles_R$ be a P-coherent set of  desirable gambles. We have to verify that the following statements are equivalent:
\begin{enumerate}
   \item $\bdomain$ includes a negative gamble that is not in $\bnegative$.
\item $\posi(\nonnegative\cup \mathcal{G})$ is incoherent, and thus $\mathcal{P}$ is empty
    \item $\bdomain^{\circ}$ is not (the restriction to  $\gambles_R$ of) a closed convex set of mixtures of classical evaluation functionals. 
    \item The extension   $ \bdomain^\bullet$ of $\bdomain^{\circ}$ in the space $\mathcal{M}$ of all charges in $\pspace$ includes only signed charges (quasi-probabilities).
\end{enumerate}
  First of all, notice that the restriction to $\gambles_R$ of the set of all normalised charges that correspond to a bounded linear functionals coincides with $\bdomain^{\bullet}$. Given this, the equivalence between (3) and (4) is immediate, whereas the equivalence between (2) and (4) is given by Theorem \ref{cor:noncoherent}.  We finally verify the equivalence between (1) and (3). In this case, the direction from left to right being obvious, the other direction is due to the fact that $g \leq f$, for every $g \in \bdomain$ and $f \in \posi(\nonnegative\cup \bdomain)\setminus \bdomain$.
  \end{proofof}
  
  The next result provides a necessary and  sufficient condition for the existence of a P-coherent set of  desirable gambles satisfying each claim of  Theorem \ref{th:fundamental}.
    \begin{proposition}\label{prop:funda}
  Assume that $\gambles_R$ includes all positive constant gambles and $\bnonnegative$ is closed (in $\gambles_R$). The following two claims are equivalent
  \begin{itemize}
  \item there is a P-coherent set of  desirable gambles $\bdomain \subseteq \gambles_R$ that includes a negative gamble that is not P-negative
  \item $\bnonnegative \subsetneq \nonnegative_R$
  \end{itemize}
  \end{proposition}

\subsection{Duality in QT}
\label{sec:dualQM}
Recall from Section \ref{sec:comp} that the set $\left\{ L \in \gambles_R^* \mid L(g)\geq0, ~~ L(\mathbbm{1}_R)=1,~\forall g \in \bdomain\right\}$ is the dual of $\bdomain\subset \gambles_{R}$. 

The monomials $\otimes_{j=1}^m x_j$ form a basis of the space $\gambles_{R}$. Define the Hermitian matrix of scalars
\begin{equation}
\label{eq:linearoperator}
Z:= L\left((\otimes_{j=1}^m x_j)(\otimes_{j=1}^m x_j)^\dagger\right),
\end{equation} 
and let $\{z_{ij}\}\in \complex^{d}$, with $d=\frac{n(n+1)}{2}$ and $n=\prod_{j=1}^m n_j$, be the vector of variables obtained  by taking  the elements 
of the upper triangular part of $Z$.
Given any gamble $g$, we can therefore rewrite $ L(g)$ as a function of the vector  $\{z_{ij}\}\in \complex^{d}$. This means that the dual space $\gambles_{R}^*$ is isomorphic to
$\complex^{d}$, and we can thence define the dual maps $(\cdot)^\circ$ between $\gambles_{R}$ and $\complex^{d}$ as follows.
\begin{definition}\label{def:dual0}
Let $\bdomain$ be a closed convex cone in $\gambles_{R}$. Its dual cone is defined as 
\begin{equation}
\label{eq:dualM0}
\bdomain^\circ=\left\{{z} \in \complex^{{d}}:  L(g)\geq0, ~\forall g \in \bdomain\right\},
\end{equation} 
where $ L(g)$ is completely determined by $\{z_{ij}\}$ via the definition \eqref{eq:linearoperator}.
\end{definition}


In discussing properties of the dual space, we need the following well-known result from linear algebra:
\begin{lemma}\label{lem:TR}
For any $M \in H^{d\times d}$ and $v \in \complex^{d}$, it holds that
\begin{equation}
\label{eq:matrixTR}
 Tr(M (v v^\dagger)) = Tr((v v^\dagger)M) =  v^\dagger M v.
\end{equation}
\end{lemma}

By Lemma \ref{lem:TR} and the definitions of $g$ and $Z$, we obtain the following result.
\begin{proposition}\label{prop:LisTR}
Let $g(x_1,\dots,x_m) = (\otimes_{j=1}^m x_j)^\dagger G (\otimes_{j=1}^m x_j)$ and $G$ Hermitian. Then 
for every $z \in  \complex^{{d}}$, it holds that $ L(g) = Tr({G} Z)$, 
where $Z$ is defined in \eqref{eq:linearoperator}. 
\end{proposition}

The next lemma states that the only symmetry in the matrix $(\otimes_{j=1}^m x_j)(\otimes_{j=1}^m x_j)^\dagger$ with $x_j \in \complex^{n_j}$ is that 
of being self-adjoint.  
\begin{lemma}
\label{lem:uniqueX}
Consider the matrix
 \begin{align}
  \label{eq:X}
X&=(\otimes_{j=1}^m x_j)(\otimes_{j=1}^m x_j)^\dagger.
 \end{align}
 Let $X_{st}$ denote the $st$-th element of $X$ then for all $t\geq s$ (upper triangular elements) we have that $X_{st}=X_{kl}$ 
 $\forall x_j \in \overline{C}^{n_j}$ iff $k=s$ and $l=t$.
\end{lemma}
We then verify that
\begin{proposition}\label{prop:igno}
Let $\bdomain$ be a P-coherent set of desirable gambles. The following holds:
\begin{equation}
\label{eq:dualM1}
\bdomain^\circ=\left\{{z} \in \complex^{{d}}:  L(g)=Tr({G} Z)\geq0, ~Z\geq 0,~\forall g \in \bdomain\right\}.
\end{equation} 
\end{proposition}
\begin{proof}
By P-coherence, $\bdomain$ includes $\Sigma^{\geq}$, which is isomorphic to the closed convex cone of PSD matrices.
We have that
$$
 L(g)=Tr({G} Z)\geq0  ~~\forall g\in \Sigma^{\geq} \subseteq \bdomain.
$$
From a standard result in linear algebra, see for instance \cite[Lemma 1.6.3]{holevo2011probabilistic}, this implies that $Z \geq0$, i.e., it must be a PSD matrix.
\end{proof}
In what follows, we verify that, analogously to Appendix \ref{subsec:dual}, the dual $\bdomain^\circ$ is completely characterised by a closed convex set of states. But before doing that, we have to clarify what is a state in this context.

 In a P-coherent theory, postulate \ref{eq:taut} is replaced with postulate \ref{eq:btaut}. Hence, to define what  a state is, one cannot anymore refer to nonnegative gambles but to gambles that are P-nonnegative. This means that states are linear operators that assign nonnegative real numbers to P-nonnegative, and that additionally preserve the unit gamble.
In the context of Hermitian gambles, the unitary gamble is  
\begin{equation}
\label{eq:constntQM}
\mathbbm{1}_R(x_1,\dots,x_m)=(\otimes_{j=1}^m x_j)^\dagger I (\otimes_{j=1}^m x_j)= \prod\limits_{i=1}^m {x_j}^{\dagger}x_j=1,
\end{equation} 
where $I$ is the identity matrix.
Therefore, we want that
\begin{equation}
\begin{aligned}
 L\Big((\otimes_{j=1}^m x_j)^\dagger I (\otimes_{j=1}^m x_j)\Big)&= L\Big(Tr\Big( I ~(\otimes_{j=1}^m x_j)(\otimes_{j=1}^m x_j)^\dagger\Big)\Big)\\
&=Tr\Big(I~ L\Big( (\otimes_{j=1}^m x_j)(\otimes_{j=1}^m x_j)^\dagger)\Big)\Big)\\
&=Tr(I ~Z)=Tr(Z)=1.
\end{aligned}
\end{equation}
Hence, the set of states is
\begin{equation}
\begin{aligned}
 \mathscr{S}_B=\{{z} \in \complex^{{d}}: \mid Z \geq0, ~~Tr(Z)=1\}.
\end{aligned}
\end{equation}

By reasoning exactly as for Theorem \ref{prop:dualcharges0}, we then have the following result.
\begin{theorem}
\label{th:dualSOS}
The map
\[ \bdomain \mapsto \mathscr{Q}:=\bdomain^\circ\cap \mathscr{S}_B\]
is a bijection between P-coherent set of desirable gambles in $\gambles_R$  and closed convex subsets of $\mathscr{S}_B$. 
\end{theorem}

Hence, we can identify the dual of a P-coherent set of desirable gambles $\bdomain$, with the closed convex set of states
\begin{equation}
\label{eq:dualM1}
\begin{aligned}
\mathscr{Q}
&=\left\{ z \in  \mathscr{S}_B:  L(g)=Tr({G} Z)\geq0,  ~~\forall g \in \bdomain\right\},
\end{aligned}
\end{equation} 
which is equivalent to .

Notice that the matrices corresponding to states are density matrices, in fact \eqref{eq:dualM1} is equivalent to
$$
\left\{ \rho \in  \He^{n\times n}: \rho\geq0, Tr(\rho)=1, Tr({G} \rho)\geq0,  ~~\forall G  \in \bdomain\right\}.
$$
Hence,  we can identify the set $\mathscr{S}_B$  with the set of density matrices and denote its elements as usual with the symbol $\rho$.



\subsection{Inference in QT}
\label{sec:QTinference}
 In this subsection, we extend to QT the results of Appendix \ref{sec:inferTDG} and \ref{sec:inferenceCPT} to derive lower 
 and upper previsions  of a gamble $f(x_1,\dots,x_m) = (\otimes_{j=1}^m x_j)^\dagger F (\otimes_{j=1}^m x_j)$ in QT. 
 By definition of lower prevision, we have that
\begin{equation}
\label{eq:lpberqt}
\begin{aligned}
 \underline{E}_B(f):=&\sup_{\gamma_0\in \reals, \lambda_i \in \reals^+} \gamma_0\\
 &s.t:\\
 &(\otimes_{j=1}^m x_j)^\dagger F (\otimes_{j=1}^m x_j) - \gamma_0\mathbbm{1}_R(x_1,\dots,x_m)  -\sum\limits_{i=1}^{|\mathcal{G}|} \lambda_i (\otimes_{j=1}^m x_j)^\dagger G_i (\otimes_{j=1}^m x_j) \in \bnonnegative,\\
\end{aligned}
\end{equation}
the upper prevision  $ \overline{E}_B(f)=-\underline{E}_B(-f)$. 
Note that the membership problem in \eqref{eq:lpberqt} reduces to find $\gamma_0\in \reals, \lambda_i \in \reals^+$ such that
$F- \gamma_0I -\sum\limits_{i=1}^{|\mathcal{G}|}G_i$ is PSD.
The dual of the above optimisation problem is
\begin{equation}
\label{eq:dualpberqt}
\begin{aligned}
 \underline{E}_B(f):=&\inf_{\rho\in \mathscr{S}_B} Tr(F\rho)\\
 &s.t:\\
 &Tr(G_i \rho) \geq0,~~\forall i=1,\dots,|\mathcal{G}|
\end{aligned}
\end{equation}
Whenever a P-coherent $\bdomain$ is maximal, its dual $\mathscr{Q}$ includes a single density matrix $\tilde{\rho}$\footnote{This does not require $|\mathcal{G}|=\infty$.
Since $\tilde{\rho}$ is a matrix it can be uniquely specified by finite assessments of desirability $\tilde{\rho}=\{\rho: Tr(G_i \rho) \geq0,~~\forall i=1,\dots,|\mathcal{G}|\}$.}.
$$
 \underline{E}_B(f) = \overline{E}_B(f)=  E_B(f)=Tr(F\tilde{\rho}).
 $$
Hence, in this case both the lower and upper prevision of the gamble $f(x_1,\dots,x_m) = (\otimes_{j=1}^m x_j)^\dagger F (\otimes_{j=1}^m x_j) $
 coincide with $Tr(F\tilde{\rho})$.  
 Note that in QT the expectation of $f$ is $Tr(F\tilde{\rho})$.

   \subsection{Technicalities for Subsection \ref{subsec:otherax}}\label{app:otherax}
   In this section we discuss how to exploit the correspondence with the system introduced in  \cite{benavoli2016quantum} in the aim of deriving the remaining three axioms of QT.
   
   Recall that a subset $ \He^{n\times n}$ is  said to be coherent if it is a  convex cone $\mathcal{C}$ containing the set  $\{G\in\He^{n\times n}:G\gneq0\}$ of all PSD matrices in $\He^{n\times n}$ and disjoint from the interior of $\{G\in\He^{n\times n}:G \leq 0\}$. We know that there is a bijective correspondence between closed convex sets of density matrices and coherent subsets of of $\He^{n\times n}$, but also between closed convex sets of density matrices and P-coherent sets of gambles. 
By looking at such correspondences, it is then immediate to verify that:
   
   \begin{proposition}
   The map $\mathfrak{f}: \bdomain \mapsto \{ G \in \He^{n\times n} \mid  x^\dagger G x \in \bdomain\}$ is a bijection between  coherent subsets of $\He^{n\times n}$ and P-coherent sets of gambles.
   \end{proposition}
   Based on this correspondence, we can thus exploit the results in  \cite{benavoli2016quantum} to derive  L\"{u}der's rule (measurement updating) and Schr\"odinger's rule (time evolution).
   
  \subsubsection{L\"uder's rule}
  \label{sec:luder}
It states the following:
   
\begin{itemize}
\item
Quantum projection measurements are described by a collection  
$\{\Pi_i\}_{i=1}^n$  of
projection operators that satisfy the completeness equation $\sum_{i=1}^n \Pi_i =I$. These are operators acting on the 
state space of the
system being measured. If the state of the quantum system is $\rho$ immediately
before the measurement then  the state after the measurement is
$$
\hat{\rho}=\dfrac{\Pi_i \rho \Pi_i}{Tr(\Pi_i \rho \Pi_i)},
$$
provided that $Tr(\Pi_i \rho \Pi_i)>0$ and the probability that result $i$ occurs is given by
$p_i=Tr(\Pi_i \rho \Pi_i)$.
\end{itemize}

A projection-valued measurement $\Pi^{*}=\{\Pi_i^*\}_{i=1}^n$ can be seen as a partition of unity $\{x^\dagger\Pi_i^*x \}_{i=1}^n$. Thus an event ``indicated'' by a certain projector $\Pi_i$
in $\Pi=\{\Pi_i\}_{i=1}^n$ can also be seen as the function $\pi_i(x):=x^\dagger\Pi_i^*x$. The information it represents is: an experiment $\Pi$ is performed and the event indicated by $\Pi_i$ happens.\footnote{
We assume  that the quantum measurement device is a ``perfect meter'' (an ideal common  assumption in QM), i.e., there are not observational errors -- Alice can trust the
received information.} Under this assumption, Alice can focus on gambles that are contingent on the event $\Pi_i$: these are the gambles such that ``outside'' $\Pi_i$ no utile is received or due -- status quo is maintained --; in other words, they represent gambles that are called off if the outcome of the experiment is not $\Pi_i$. 
Mathematically, we define Alice's conditional set of desirable gambles as follows.

\begin{definition}
 Let  $\bdomain$ be an P-coherent set of gambles, the set obtained as
\begin{equation}
\label{eq:condition}
\bdomain_{\Pi_i}=\left\{g(x) \in \gambles_R \mid  
x^\dagger \Pi_i G \Pi_i  x\in \bdomain \right\}
\end{equation} 
is  called the {\bf set of desirable gambles conditional} on  $\Pi_i$. \end{definition}

Notice that $x^\dagger \Pi_i G \Pi_i  x= \alpha \pi_i(x)$.

We can also compute the dual of  $\bdomain_{\Pi_i}$, i.e.,
$\mathscr{Q}_{\Pi_i}$, and thus obtain the
\begin{description}
\item[Subjective formulation of L\"uder's rule:]~\\
Given a closed convex set of states $\mathscr{Q}$, the corresponding conditional set on $\Pi_i$ is obtained as
\begin{equation}
\label{eq:rhobayes}
 \mathscr{Q}_{\Pi_i}=\left\{ \dfrac{\Pi_i Z \Pi_i}{Tr(\Pi_i Z \Pi_i)} \Big|  z \in 
\mathscr{Q}\right\},
\end{equation} 
 provided that $Tr(\Pi_i Z \Pi_i)>0$  for every $z \in \mathscr{Q}$. Note that the latter  condition  implies that $\pi_j(x) \notin \mathscr{Q}$
 for any $j\neq i$.
\end{description} 

The following diagram gives the relationships among  $\bdomain,\mathscr{Q}, \bdomain_{\Pi_i},\mathscr{Q}_{\Pi_i}$.
$$
\begin{tikzpicture}
  \matrix (m) [matrix of math nodes,row sep=3em,column sep=6em,minimum width=2em]
  {
     \bdomain & \,\bdomain_{\Pi_i} \\
     \mathscr{Q} & \mathscr{Q}_{\Pi_i} \\};
  \path[-stealth]
    (m-1-1) edge [<->] node [left] {dual} (m-2-1)
            edge [double] node [below] {conditioning} (m-1-2)
    (m-2-1.east|-m-2-2) edge [double]  node [below] {conditioning}
          (m-2-2)
    (m-2-2) edge [<->] node [right] {dual} (m-1-2);
\end{tikzpicture}
$$

\subsubsection{Time evolution postulate}
It states the following:
   
\begin{itemize}\item
The evolution of a closed quantum system is described by a unitary
transformation. That is, the state $\rho$ of the system at time $t_0$ is related 
to the state
$\rho'$ of the system at time $t_1>t_0$ by a unitary operator $U$ which depends only 
on the
times $t_0$ and $t_1$, $\rho' = U \rho U^\dagger$.
\end{itemize}

Let us consider the dynamics of sets of gambles at present time $t_0$ and future time $t_1$ under the assumption that no information at all is received during such an interval of time (i.e., we have a closed quantum system). The focus is on characterising the coherence of sets of gambles in this time. 

To this end, we add the following temporal postulate:
\begin{enumerate}[label=\upshape B3.,ref=\upshape B3]
\item\label{eq:b3}  A temporal P-coherent transformation is a map $\phi(\cdot,t_1,t_0)$ from $\gambles_R$ to itself that satisfies the following properties: 
 \begin{itemize}
 \item[(i)] $\phi(\cdot,t_0,t_0)$ is the identity map; 
 \item[(ii)]  $\phi(\cdot,t_1,t_0)$ is onto; 
 \item[(iii)] $\phi(\gambles^+,t_1,t_0)=\gambles^+$;
\item[(iv)] $\phi(\cdot,t_1,t_0)$ is linear and constant preserving.
\end{itemize}
 \end{enumerate}
The rationale behind these conditions is the following.

Condition~(i) is obvious. 

Condition~(ii) is a way of stating that sets of desirable gambles are only established at present time $t_0$, since
any gamble $g_1$ at time $t_1$ corresponds to an element $g_0$  at time $t_0$. 

Condition~(iii)  means that no further information is received from time  $t_0$ to time  $t_1$.

Finally, condition~(iv) states in particular once again that the utility scale is linear.

By \cite{benavoli2016quantum}[Theorem A.9] and the fact that $\mathfrak{f}$ is a bijection preserving coherence, the temporal P-coherence postulate \eqref{eq:b3} leads to the following:
   
\begin{description}
\item[Subjective formulation of the time evolution postulate of QT:]~\\
(1) All the transformations $\phi(\cdot,t_1,t_0)$ defined above
are of the following form
$$
g(x)  \xhookrightarrow{\phi} h(x)=x^\dagger (U^\dagger G U) x ,
$$
for some unitary or anti-unitary matrix $U \in \He^{n\times n}$, which only depends on the times $t_1,t_0$ and is equal to the identity
for $t_1=t_0$.\\
(2)  The transformation $\phi(\cdot,t_1,t_0)$  preserves P-coherence:\vspace{0.3cm}\\
\centerline{if $\bdomain_0 $ is P-coherent, then $\bdomain_1=\{g(x) \in \gambles_R \mid x^\dagger (U^\dagger G U) x \in \bdomain_0\}$ is  also P-coherent.}
\end{description}

By exploiting duality, we can also reformulate the above results  in terms of sets of states and derive the time evolution postulate as a direct consequence of temporal P-coherence.

\end{document}

%% file: macros.tex
\usepackage{multirow}
\usepackage{bm}
\usepackage{amsmath, amssymb, amsfonts,amsthm}
\usepackage{array}
\usepackage{epsfig}
\usepackage{enumerate}
\usepackage{bbm}
\usepackage{color}
\usepackage{float}
\usepackage[mathscr]{euscript}
\usepackage{mathtools}
\usepackage{graphicx,subcaption}
 \usepackage{physics}
\usepackage{framed}
\usepackage{setspace}
\usepackage{mdframed}
\usepackage{tikz-cd}

  \usepackage[numbers,sort&compress]{natbib} 

\usepackage{thmtools}
\usepackage{environ}
\usepackage{mathptmx}
  \usepackage{times}
\usepackage{pgfplots}
\usetikzlibrary{matrix}
\usepackage[toc,page]{appendix}

  \usetikzlibrary{positioning,calc,fit,shapes.geometric,patterns}
  \pgfdeclarelayer{background}
  \pgfdeclarelayer{foreground}
  \pgfsetlayers{background,main,foreground}
  \tikzstyle{vec}=[circle,inner sep=1pt,outer sep=-1pt,fill]
  \tikzstyle{border}=[thick]
  \tikzstyle{favborder}=[border,dotted]
  \tikzstyle{exclborder}=[border,dashed]
\usepackage{etex,etoolbox}
\usepackage{hyperref}

\hypersetup{colorlinks=true,linkcolor=blue,citecolor=magenta}

\newenvironment{proofof}[1]{\begin{trivlist}\item[\hskip\labelsep{\textsc{Proof
  of {#1}.\ }}]}{\hspace*{\fill} {\qed}\end{trivlist}}

\newenvironment{tbs}{%
   \small\tt
   \begin{enumerate}[$\blacktriangleright$]}{\end{enumerate}}
\newcommand{\btbs}{\begin{tbs}}                                                                      
\newcommand{\etbs}{\end{tbs}}

\newcommand{\rednew}[1]{\textcolor{black}{#1}}
\newcommand{\bluenew}[1]{\textcolor{black}{#1}}

\newcommand{\Reals}{\mathbb{R}} 

\newcommand{\complex}{\mathbb{C}}

 \DeclareMathOperator{\posi}{posi}

  \DeclareMathOperator{\cl}{cl}
\newcommand{\pspace}{\Omega}

\newcommand{\reals}{\Reals}

\newcommand{\gambles}{\mathcal{L}}

\newcommand{\He}{\mathcal{H}}

\newcommand{\theory}{\mathscr{T}}
\newcommand{\btheory}{\mathscr{T}^\star}
\newcommand{\domain}{\mathscr{K}}
\newcommand{\assess}{\mathcal{G}}
\newcommand{\bdomain}{\mathscr{C}}

\newcommand{\nonnegative}{\gambles^{\geq} }
\newcommand{\bnonnegative}{\Sigma^{\geq}}

\newcommand{\negative}{\gambles^{<} }
\newcommand{\bnegative}{\Sigma^{<} }
\newcommand{\states}{\mathscr{S}}

\newtheorem{remark}{Remark}
\newtheorem{lemma}{Lemma}

\newtheorem{theorem}{Theorem}
\newtheorem{definition}{Definition}
\newtheorem{proposition}{Proposition}
\newtheorem{example}{Example}